\documentclass[journal]{IEEEtran}

\usepackage{bbm, graphicx, caption, subcaption, float, mathtools, xcolor, comment, braket, amssymb}

\DeclarePairedDelimiter{\floor}{\lfloor}{\rfloor}
\usepackage[font=small,skip=0pt]{caption}

\usepackage{amsfonts, booktabs, siunitx}

\newcommand{\E}{\mathbb{E}}
\newcommand{\dx}[1]{\textrm{d}#1}

\usepackage{amsthm,cite}
\newtheorem*{example*}{Example}
\newtheorem{lemma}{Lemma}
\newtheorem{theorem}{Theorem}
\newtheorem{corollary}{Corollary}

\allowdisplaybreaks
\makeatletter
\newcounter{longaligned}
\newenvironment{longaligned}[1][]
 {%
  \stepcounter{longaligned}%
  \refstepcounter{equation}%
  \label{longaligned@\thelongaligned}%
  #1%
  \start@align\@ne\st@rredtrue\m@ne 
 }
 {\endalign}
\newcommand{\longalignedtag}{\tag{\ref{longaligned@\thelongaligned}}}

\makeatother

\begin{document}
\title{Straggler Mitigation at Scale}

\author{%
  \IEEEauthorblockN{Mehmet Fatih Akta\c{s} and Emina Soljanin} \\
  \IEEEauthorblockA{Department of Electrical and Computer Engineering, Rutgers University \\
  Email: \{mehmet.aktas, emina.soljanin\}@rutgers.edu}
  \thanks{A preliminary version of this paper appeared in the ACM SIGMETRICS Performance Evaluation Review \cite{MAMA:AktasPS17, IFIP:AktasPS18}.}%
}

\maketitle

\begin{abstract}
  Runtime performance variability has been a major issue, hindering predictable and scalable performance in modern distributed systems.
  Executing requests or jobs redundantly over multiple servers have been shown to be effective for mitigating variability, both in theory and practice. Systems that employ redundancy has drawn significant attention, and numerous papers have analyzed the pain and gain of redundancy under various service models and assumptions on the runtime variability.
  This paper presents a cost (pain) vs.\ latency (gain) analysis of executing jobs of many tasks by employing replicated or erasure coded redundancy. The tail heaviness of service time variability is decisive on the pain and gain of redundancy and we quantify its effect by deriving expressions for cost and latency.
  Specifically, we try to answer four questions: 
  1) How do replicated and coded redundancy compare in the cost vs.\ latency tradeoff?
  2) Can we introduce redundancy after waiting some time and expect it to reduce the cost?
  3) Can relaunching the tasks that appear to be straggling after some time help to reduce cost and/or latency?
  4) Is it effective to use redundancy and relaunching together?
  We validate the answers we found for each of these questions via simulations that use empirical distributions extracted from a Google cluster data.
\end{abstract}


\section{Introduction}
Providing \emph{predictable} performance is an important ongoing challenge for distributed computing systems.
In distributed settings, a \emph{job} is split into multiple smaller \emph{tasks}, which get spread over separate resources for parallel execution.
Task execution times in modern systems are known to exhibit significant runtime variability due to many factors such as power management, software or hardware failures, maintenance, and most importantly, resource sharing \cite{Dryad:IsardBY07, MapReduce:DeanG08, Mantri:AnanthanarayananKG10, ResilientDistributedDatasets:ZahariaCD12, TailAtScale:DeanB13, StragglerRootCauseAnalysisInDatacenters:OuyangGY16, RootCauseAnalysisOfStragglersInBigDataSystem:ZhouLY18}.
Runtime variability may cause some tasks to \emph{straggle} and take much longer to complete than other tasks in the job. Since a distributed job completes only when all its tasks complete, straggler tasks significantly delay the job completion. As the number of tasks in a job increases so does the chance that at least one of them will be a straggler, thus the impact of stragglers on the job completion time is greater at scale \cite{AchievingRapidResponseTimesInLargeOnlineServices:Dean12, TailAtScale:DeanB13}.

Straggler problem has received significant attention from the systems research community.
Existing solution techniques fall into two categories: 
i) Squashing runtime variability via preventive actions such as blacklisting faulty machines that frequently exhibit high variability \cite{MapReduce:DeanG08, AchievingRapidResponseTimesInLargeOnlineServices:Dean12} or learning the characteristics of task-to-node assignments that lead to high variability and avoiding such problematic task-node pairings \cite{ProactiveStragglerAvoidance:YadwadkarC12},
ii) Speculative execution by launching the tasks together with replicas and waiting only for the fastest copy to complete \cite{ImprovingMapReducePerformance:ZahariaKJ08, Mantri:AnanthanarayananKG10, Dremel:MelnikGL10, AttackOfClones:AnanthanarayananGS13, LowLatencyviaRed:VulimiriGM13}.
Because runtime variability is caused by intrinsically complex reasons, preventive measures for stragglers could not fully solve the problem and runtime variability continued plaguing the compute workloads \cite{AchievingRapidResponseTimesInLargeOnlineServices:Dean12, AttackOfClones:AnanthanarayananGS13}.
Speculative task execution on the other hand has proved to be an effective remedy, and indeed the most widely deployed solution for stragglers \cite{TailAtScale:DeanB13, DecentralizedSpeculationAwareClusterScheduling:RenAW15}.
For instance with task replication, median runtime slowdown experienced by the tasks within a job is brought down from 8 (and 7) to 1.08 (and 1.1) in Facebook's production Hadoop cluster (and Bing’s Dryad cluster) \cite{DecentralizedSpeculationAwareClusterScheduling:RenAW15}.

Executing tasks with greater number of copies will surely reduce the chance of having to wait for a straggler.
However, task replicas occupy system resources that could otherwise be used to execute other tasks. Furthermore, if task replicas are employed excessively, they can overburden the system and further aggravate the runtime variability, given that the primary cause of runtime variability is resource sharing.
Therefore, replicas are employed with care in practice, e.g., replica tasks are used only for jobs with a few tasks \cite{AttackOfClones:AnanthanarayananGS13}, or only tasks that straggle beyond some threshold are replicated \cite{ImprovingMapReducePerformance:ZahariaKJ08}.
More recently, replicas are proposed to be dispatched for single-task jobs only if any server is found idle, which is shown, with a queueing theoretic analysis, to not drive the system to instability by dispatching excessive number of replicas \cite{DecouplingSlowdownJobsize:GardnerHS17}.

This paper focuses on two important performance metrics for distributed job execution: 1) \emph{Latency,} measuring the time to complete the job, and 2) \emph{Cost,} measuring the total resource time spent to execute the job.
Job execution is desired to be fast and with low cost, but these are often conflicting objectives.
Cost of executing a job depends on the number of tasks\footnote{Resource usage of tasks vary across different jobs or might vary even within the same job in practice \cite{Kubernetes:BurnsGO16}. We abstract this complexity by assuming that each task uses one unit of resource per unit time.} and the time each task takes to finish.
Executing a job with task replicas is expected to reduce the time spent by the tasks in the system, while also increasing the total number of tasks involved in completing the job, which is likely to increase the cost.
It is important to understand the effect of added redundancy not only on the latency but also on the cost because the load exerted on the system by a job execution is determined by its cost (as elaborated in Sec.~\ref{subsec:subsec_on_the_cost}).

Erasure coding implements a more general form of redundancy than simple replication, and has been considered for straggler mitigation both in data download \cite{Codes&Qs:JoshiLS12, Codes&Qs:HuangPZ12, CodesQs:KadheSS15_Allerton}, and more recently in distributed computing context \cite{ShortDot:DuttaCG16, MachineLearningWithCodes:LeeLP17, GradientCoding:TandonLD16, CodedGradientDescent:LiKA17, StragglerMitigationWithDataEncoding:KarakusSD17, CodedMatrixMultiplication:YuMA18, CodedGradientDescent:HalbawiAS18}.
With coding, a job of $k$ tasks is expanded into a job of $n$ tasks with $n-k$ \textit{parity} tasks. Parity tasks are constructed by encoding the initial $k$ tasks, which is done by embedding redundancy either in the computational procedure collaboratively implemented by the tasks (e.g., \cite{ShortDot:DuttaCG16}) or in the data the tasks consume during execution (e.g., \cite{StragglerMitigationWithDataEncoding:KarakusSD17}).
If coded tasks are created with \textit{MDS} code, the most commonly used encoding model, any $k$ of the $n$ tasks would be sufficient to recover the desired outcome of the job, thus only the fastest $k$ tasks would be sufficient for completing the job.

Modeling task execution times and the variability they exhibit is crucial for the theoretical analysis of straggler mitigation techniques to match with the experimental measurements.
In the analysis of straggler mitigation techniques, variability in execution times is commonly expressed with a fixed straggling factor. The straggling factor for each task is typically assumed to be independently drawn from a fixed random variable, which we also adopt in this paper.
However, runtime variability is known to be to a large part caused by resource sharing in practice \cite{TailAtScale:DeanB13, StragglerRootCauseAnalysisInDatacenters:OuyangGY16, RootCauseAnalysisOfStragglersInBigDataSystem:ZhouLY18}, and the redundant tasks added into system exert additional load on the system resources, which is expected to aggravate the runtime variability.
Therefore, we believe that the model of variability should account for the redundancy added into system.
In Sec.~\ref{sec:sec_when_red_changes_tail}, we consider a model where the tail of task execution times changes with the level of redundancy added into system, and we study the cost and latency of redundancy under this model.

There are various decisions to make while employing redundancy for straggler mitigation.
The first natural step is to decide adding whether replicated or coded tasks, and how many of them.
Secondly, waiting for some time before launching the replica tasks has been considered to reduce the cost of redundancy \cite{RepedComputing:WangJW15}. A natural question is that does waiting before launching the redundant (replicated or coded) tasks help in general to reduce cost.
In this paper, we analyze the cost vs.\ latency tradeoff to find out the best practice in making these decisions.
As an alternative to adding redundancy, cancelling and relaunching the tasks that appear to be straggling after waiting some time has been considered \cite{ImprovingMapReduceInHeteroEnvironments:Zaharia08}. This is justified by the heavy tailed nature of task execution times as observed in practice \cite{TailAtScale:DeanB13, GoogleTraceAnalysis:ReissTG12, GRASS:AnanthanarayananHR14}.
We quantify the effect of straggler relaunch on the cost vs.\ latency tradeoff in terms of the tail heaviness pronounced by the service time variability. We also consider employing straggler relaunch together with redundancy, and analyze its effects on cost and latency.
Parts of the results presented in this paper were published in  \cite{MAMA:AktasPS17, IFIP:AktasPS18}.

This paper is structured as follows.
In Sec.~\ref{sec:sec_system_model}, we explain the system model that is used for the presented analysis, and formally define the cost and latency of distributed job execution.
In Sec.~\ref{sec:sec_coding_vs_rep}, we examine the effect of the type and level of redundancy, and the launch time of redundant tasks on the cost vs.\ latency tradeoff.
In Sec.~\ref{sec:sec_when_red_changes_tail}, we evaluate the performance of job execution with redundancy when the redundant tasks added into system changes the tail of service time variability.
In Sec.~\ref{sec:sec_straggler_relaunch}, we study straggler relaunch and investigate its impact on the cost and latency.
In Sec.~\ref{sec:sec_red_togetherwith_relaunch}, we consider employing straggler relaunch together with redundancy.
In Sec.~\ref{sec:sec_conclusions}, we summarize our key findings, discuss the shortcomings of our analysis and possible future directions.

\vspace{0.5em}
\noindent
\textbf{Summary of observations.} Coding allows increasing the level of added redundancy with finer steps than replication, which translates into greater achievable cost vs.\ latency region.
Waiting for some time before launching the redundant tasks is not effective in trading off latency for reduced cost when the employed redundancy is coding, that is, one can obtain lower latency for the same cost by launching less number of coded tasks rather than delaying their launch time. When the employed redundancy is replication, some cost reduction is possible by launching the replica tasks after waiting some time.
Coding is more efficient than replication in the cost vs.\ latency tradeoff; adding coded tasks into job execution yields higher reduction in latency per incurred cost (hence per incurred additional load on the system) compared to adding replicated tasks.
Execution with redundancy reduces the cost and latency together when enough tail heaviness is pronounced by the service time variability. The required tail heaviness is smaller
when coding is employed compared to replication.
The advantage of coding over replication becomes greater when the job is executed at higher scale, i.e., when the job consists of greater number of parallel tasks.

Relaunching tasks that appear to be straggling after some time reduces the cost and latency when relaunching is performed at the right time and enough tail heaviness is pronounced by the service time variability.
Redundancy and straggler relaunch serve the same purpose of mitigating stragglers, hence employing both together require greater tail heaviness in service time variability in order to reduce the cost and latency.

\section{System Model}
\label{sec:sec_system_model}
We adopt a system model that is an extension of what is adopted in \cite{StragglerRep:WangJW15}; execution time (duration from its launch time to completion) of each task is modeled with a single random variable. All $k$ tasks of a job are launched simultaneously and the execution time of each is assumed to be identically and independently distributed (i.i.d.).
We use two canonical distributions to model task execution times: 1) Shifted exponential $\mathrm{SExp}(s, \mu)$ with a positive minimum value $s$ and a tail decaying exponentially at rate $\mu$, and 2) $\mathrm{Pareto}(s, \alpha)$ with a positive minimum value $s$ and a power law tail with index $\alpha$.
Minimum value of the distribution models the minimum service time of the tasks (i.e., task size), while tail of the distribution models the slowdown due to runtime variability; smaller $\mu$ or $\alpha$ implies greater chance for stragglers.

Task execution times in modern compute systems are known to exhibit heavy tail \cite{TailAtScale:DeanB13, GoogleTraceAnalysis:ReissTG12, GRASS:AnanthanarayananHR14}.
In Fig.~\ref{fig:figs/plot_google_empiricaltail}, we plot the tail distribution of the task execution times\footnote{Task execution times are calculated as the difference between the timestamps for SCHEDULE and FINISH events for each task as given in \cite{GoogleTraceAnalysis:ReissTG12}.} that we extracted from a Google Trace data for jobs with $15$, $400$, or $1050$ tasks \cite{GoogleTraceAnalysis:ReissTG12}. Note that both axes in the plots are in log scale, hence an exponential tail would have appeared as an exponentially decaying curve, while a true power law tail (e.g., tail of $\mathrm{Pareto}$) would have pronounced a linear decay at a constant rate \cite{FundamentalsOfHeavyTails:NairWZ13}. Tail distributions shown in the figure exhibit exponential decay at small values and a linearly decaying trend at larger values, which indicates a heavy tailed runtime variability \cite{PerfEvalWithHeavyTails:Crovella01}. Note that the steep decay of the tail at the far right edge is due to the bounded support of the distributions.

Assuming execution times to be identically distributed for tasks within the same job is appropriate since jobs in practice are known to be a collection of one or more usually identical tasks \cite{GoogleClusterTrace:ReissWH11, GoogleTraceAnalysis:ReissTG12}.
However, when task execution times are modeled using a distribution with a minimum value of zero (e.g., exponential distribution), assuming independent execution times across the initial and redundant tasks proved to be problematic because added redundancy in this case can make job execution time arbitrarily small. This is in contrast to reality where tasks have an inherent size and due to this, job execution times are lower bounded by a positive value regardless of the level of added redundancy. Modeling task execution times with a minimum value of zero have previously led to theoretical results that are at odds with experimental measurements. Implications of this are discussed in detail in \cite{DecouplingSlowdownJobsize:GardnerHS17} and authors propose a better model for service times in which the time due to task size is decoupled from the time due to runtime variability. Specifically, service times are modeled as $s \times Sl$ where $s$ represents the inherent task size and $Sl$ is the slowdown factor, which is assumed to be i.i.d. across tasks and servers with a minimum value of $1$.
Distributions that we adopt for modeling task execution times can be expressed using the decoupling method introduced in \cite{DecouplingSlowdownJobsize:GardnerHS17}; $\mathrm{SExp}(s, \mu)$ can be written as $s \times \mathrm{SExp}(1, \mu)$ or $\mathrm{Pareto}(s, \alpha)$ as $s \times \mathrm{Pareto}(1, \alpha)$.

In our model, redundant tasks are added into execution only if the job does not complete within some time $\Delta$.
Redundancy is introduced either in the form of task replicas or coded parity tasks. When replication is employed, $c$ replicas are launched for every remaining task at time $\Delta$. When coding is employed, $n-k$ MDS coded parity tasks are launched at time $\Delta$ (see Fig.~\ref{fig:fig_delayed_red}).
When straggler relaunch is implemented, tasks (initial or redundant) remaining at time $\Delta$ are canceled and fresh replacements are immediately launched in their place.
\begin{figure}[t]
  \centering
  \includegraphics[width=0.5\textwidth, keepaspectratio=true]{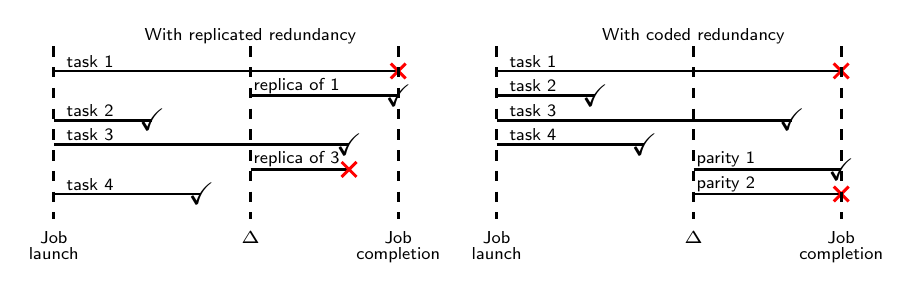}
  \caption{A job of four tasks is executed by launching replicated (Left) or coded (Right) tasks, some time $\Delta$ after launching job's initial tasks. Check marks represent task completions while crosses represent cancellations. With replication, exact clones of the remaining tasks are launched, while with coding, parity tasks can be used as a ``clone'' for any task, therefore, stragglers do not have to be tracked down.}
  \label{fig:fig_delayed_red}
\end{figure}

\begin{figure*}[t]
  \centering
  \begin{subfigure}[]{.25\textwidth}
    \centering
    \includegraphics[width=1\textwidth, keepaspectratio=true]{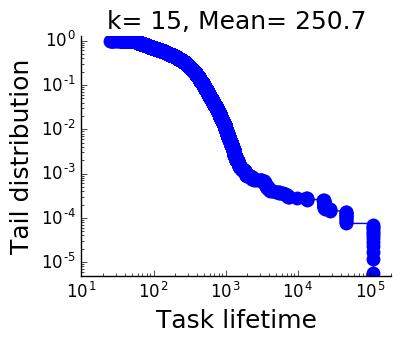}
  \end{subfigure}
  \hspace{1em}
  \begin{subfigure}[]{.25\textwidth}
    \centering
    \includegraphics[width=1\textwidth, keepaspectratio=true]{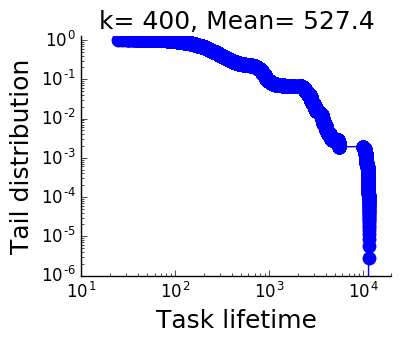}
  \end{subfigure}
  \hspace{1em}
  \begin{subfigure}[]{.25\textwidth}
    \centering
    \includegraphics[width=1\textwidth, keepaspectratio=true]{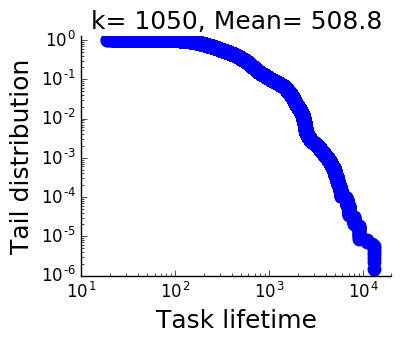}
  \end{subfigure}
  \caption{Empirical tail distribution of task execution times for Google cluster jobs with number of tasks $k=15, 400, 1050$.}
  \label{fig:figs/plot_google_empiricaltail}
\end{figure*}

We define the cost of executing a job as the sum of the lifetimes of all the tasks (including the redundant ones) involved in its execution. Lifetime of a task is the duration from its launch to its completion or cancellation.
Depending on the application domain, there are two possible cost definitions: 1) \textit{Cost with task cancellation}; outstanding redundant tasks are canceled as soon as the job completes (as illustrated in Fig.~\ref{fig:fig_delayed_red}), which is a viable option for distributed job execution, 2) \textit{Cost without task cancellation}; outstanding redundant tasks are left to run until they complete, which for instance is the only option for routing messages with redundancy in an opportunistic network \cite{ErasureCodingBasedRoutingForOpportunisticNetworks:WangJM05}.
We assume that task cancellation takes place instantly and does not incur any delay.
In the following subsection, we elaborate on the meaning and consequences of the job execution cost.

\subsection{On the cost of job execution}
\label{subsec:subsec_on_the_cost}
Cost, as is defined here, reflects the total resource time spent while executing a job. Lower cost translates into executing the same job by occupying less area in $\textnormal{system capacity}\times\textnormal{time}$ space.
Thus, reducing the cost of job executions allows fitting more jobs per area and leads to higher system throughput \cite{BoostingThroughput:Joshi17}.

As we show in the following sections, adding redundant tasks into a job execution can increase or \emph{decrease} the cost depending on the variability pronounced in task execution times, and the type and level of introduced redundancy.
Since redundancy can lead to higher cost, it should be employed with care. Executing jobs at a higher cost implies occupying greater portion of system's overall capacity per job, which increases the load on the system. This may translate into greater congestion in the system resources, hence aggravate job slowdowns or even drive the system to instability.
For instance, \cite{DecouplingSlowdownJobsize:GardnerHS17} shows, with a queueing theoretic analysis, how excessive replication of single-task job arrivals can drive system to instability.
As another example, \cite{AttackOfClones:AnanthanarayananGS13} introduces a system, named as \emph{Dolly}, which launches replicas only for small jobs that consist of $\leq 10$ tasks. This is shown to achieve significant reduction in latency without overburdening the system according to the traces collected on two clusters at Facebook and Microsoft Bing. The underlying reason for the success of their replication scheme is the workload characteristics; small jobs were observed to tend to have short duration, thus, replicating them did not introduce substantial cost overhead in the system, while returning substantial reduction in the latency of short jobs.

The workload and system characteristics considered in \cite{AttackOfClones:AnanthanarayananGS13} are not universal; execution with redundancy is relevant in general not only for small jobs but also for jobs that run at higher scale for large duration.
Job slowdowns due to stragglers is an emerging problem for future high performance computing (HPC) systems. Exascale computing is expected to be implemented by systems that are much larger in size and will enable execution at unprecedented levels of parallelism. These future systems are anticipated to be prone to much higher node level runtime variability \cite{ExascaleDOE:BrownMB10}. Moreover, to implement high resource utilization, resource scheduling in these systems is suggested to be realized with time-sharing rather than today's de facto batch scheduling \cite{TimeSharingInHPC:HofmeyrIC16}. As resource sharing is pointed out as the primary cause of stragglers in data centers \cite{TailAtScale:DeanB13}, performance of future HPC systems is likely to greatly suffer from stragglers. Simulations over the traces collected on Edison Supercomputer demonstrate that jobs with larger number of tasks and shorter duration experience higher slowdowns due to runtime variability under batch scheduling, while under time-sharing based resource scheduling, slowdowns are observed to be relatively uniform regardless of the number of tasks or the job duration \cite{TimeSharingInHPC:HofmeyrIC16}.

\subsection{Notation and Tools for Analysis}
Expected cost ($C$) and latency ($T$) are the two metrics that we use to quantify the pain and gain of distributed execution of jobs with redundancy and/or straggler relaunch. Thus, the cost and latency by themselves imply their expected values throughout, and any other quantity associated with them is made explicit. Note that the cost and latency depend on the number of tasks $k$ that constitute the job, task sizes $s$, runtime variability (determined by $\mu$ or $\alpha$), the level of redundancy added into the job ($c$ task replicas or $n-k$ coded tasks), as well as the time $\Delta$ at which redundant tasks are launched and/or straggler relaunch is performed.

Derivations of the cost and latency expressions make frequent use of the law of total probability since we consider adding redundancy and/or performing straggler relaunch after waiting some time $\Delta$. Results from order statistics are essential for the derivations since only a subset of the launched tasks is necessary for job completion when redundancy is employed. Derivations presented in the paper require tedious algebra at times and the expressions involve some special functions that commonly appear while working with order statistics. For completeness, we kept every non-trivial step in the proofs. This made some proofs lengthy and we placed them in the Appendix that is made available as a supplement to this paper.

We here give an overview of the notation and special functions that appear throughout the paper.
For their detailed definitions and interesting properties, we refer the reader to \cite{NIST:DLMF}.
$X_{n:i}$ denotes the $i$th order statistic of $n$ i.i.d. samples drawn from a random variable $X$.
$H_n$, the $n$th harmonic number, is defined as $\sum_{i=1}^n 1/i$ for $n \in \mathbb{Z}^+$ or as $\int_0^1 (1-x^n)/(1-x) \dx{x}$ for $n \in \mathbb{R}$.
$H_{n^2}$, the $n$th generalized harmonic number of order two, is defined as $\sum_{i=1}^n 1/i^2$.
Incomplete Beta function $B(q;m,n)$ is defined for $q \in [0,1]$, $m, n \in \mathbb{R}^+$ as $\int_0^q u^{m-1}(1-u)^{n-1} \dx{u}$, Beta function $B(m,n)$ as $B(1;m,n)$ and its regularized form $I(q;m,n)$ as $B(q;m,n)/B(m,n)$.
Gamma function $\Gamma(x)$ is defined as $\int_0^{\infty} u^{x-1}e^{-u} \dx{u}$ for $x \in \mathbb{R}$ or as $(x-1)!$ for $x \in \mathbb{Z}^+$.


\section{Coding vs.\ Replication}
\label{sec:sec_coding_vs_rep}
In this section, we study the cost vs.\ latency tradeoff in executing a distributed job by adding task replicas or coded tasks after waiting some time $\Delta$. Note that we do not consider straggler relaunch until Sec.~\ref{sec:sec_straggler_relaunch}.
Theorems given below firstly present expressions for the cost and latency assuming exponential task execution times, which we then use to derive the cost and latency for shifted-exponential task execution times.

\vspace{2ex}
\noindent
\textbf{Consider executing a job of $k$ tasks by adding $c$ replicas for each remaining task after waiting some time $\Delta$.}
\begin{theorem}
  Suppose task execution times are i.i.d. with $\mathrm{Exp}(\mu)$.
  Distribution of job execution time is given as
  \begin{equation}
  \begin{split}
    \Pr\{T \leq t\} &= \Bigl(1 - \mathbbm{1}(t \leq \Delta)(e^{-\mu t} - e^{-\mu\Delta}) \\
    &\qquad - \mathbbm{1}(t > \Delta)e^{-\mu\left((c+1)(t-\Delta) + \Delta\right)} \Bigr)^k
  \end{split}
  \label{eqn:eq_k_cd_Exp_tail}
  \end{equation}
  
  Latency is well approximated as
  \begin{equation}
    \E[T] \approx \frac{1}{\mu}\left(H_k - \frac{c}{c+1}H_{k-kq}\right).
  \label{eqn:eq_k_cd_Exp_ET}
  \end{equation}
  
  Cost with ($C^c$) or without ($C$) task cancellation is given as
  \begin{equation}
    \E[C^c] = \frac{k}{\mu}, \quad\quad \E[C] = \left(c(1-q) + 1\right)\frac{k}{\mu}.
  \label{eqn:eq_k_cd_Exp_EC}
  \end{equation}
  where $q = 1 - e^{-\mu\Delta}$.
\label{thm_k_cd_Exp_T_C}
\end{theorem}

\begin{theorem}
  Suppose task execution times are i.i.d. with $\mathrm{SExp}(s, \mu)$.
  Distribution and the expected value of job execution time are given as
  \begin{equation}
  \begin{split}
    \Pr\{T > t\} &= \Pr\{T_e > t - s\}, \\
    \E[T] &= s + \E[T_e].
  \end{split}
  \label{eqn:eq_k_cd_SExp_tail__ET}
  \end{equation}
  where $T_e$ is the job execution time when task execution times are distributed as $\mathrm{Exp}(\mu)$, for which the distribution and expected value are given in Thm.~\ref{thm_k_cd_Exp_T_C}.
  
  Cost with task cancellation is given as
  \begin{equation}
    \E[C^c] = 
    \begin{cases}
      \begin{aligned}
        & k(c+1)\Biggl(s + \frac{1}{\mu} \\
        &~ \times \left(1 - \frac{c}{c+1}(e^{-\mu\Delta} + \mu\Delta) \right)\Biggr)
      \end{aligned} & \Delta \leq s, \\
      k\left(s + \frac{1}{\mu}\left(1 + c\left(1-q-e^{-\mu\Delta}\right)\right)\right) & o.w.
    \end{cases}
  \label{eqn:eq_k_cd_SExp_ECwcancel}
  \end{equation}
  
  Cost without task cancellation is given as
  \begin{equation}
    \E[C] = k\left(c(1-q)+1\right)(s + 1/\mu).
  \label{eqn:eq_k_cd_SExp_EC}
  \end{equation}
  where $q = \mathbbm{1}(\Delta > s)\left(1 - e^{-\mu(\Delta - s)}\right)$.
\label{thm_k_cd_SExp_T_C}
\end{theorem}

\vspace{2ex}
\noindent
\textbf{Consider executing a job of $k$ tasks by adding $n-k$ coded tasks after waiting some time $\Delta$.}
\begin{theorem}
  Suppose task execution times are i.i.d. with $\mathrm{Exp}(\mu)$.
  Distribution and the expected value of job execution time are well approximated as
  \begin{longaligned}[\label{eqn:eq_k_nd_Exp_tail__ET}]
    & \Pr\{T > t\} \approx \mathbbm{1}(t \leq \Delta)\left(q^k - (1-e^{-\mu t})^k\right) \longalignedtag \\
    &\qquad\qquad\quad + I\left(\mathbbm{1}(t > \Delta)e^{-\mu(t-\Delta)}; n-k+1, k(1-q)\right) \\
    &\qquad\qquad\quad - q^k I\left(\mathbbm{1}(t > \Delta)e^{-\mu(t-\Delta)}; n-k+1, 0\right), \\
    & \E[T] \approx \Delta - \frac{1}{\mu}\left(B(q;k+1,0) + H_{n-kq} - H_{n-k}\right).
  \end{longaligned}
  
  Cost with ($C^{c}$) or without ($C$) task cancellation is given as
  \begin{equation}
  \begin{split}
  	\E[C^c] = \frac{k}{\mu}, \qquad \E[C] = \frac{k}{\mu}q^k + \frac{n}{\mu}\left(1-q^k\right),
  \end{split}
  \label{eqn:eq_k_nd_Exp_EC}
  \end{equation}
  where $q = 1 - e^{-\mu\Delta}$.
  \label{thm_k_nd_Exp_T_C}
\end{theorem}

\begin{theorem}
  Suppose task execution times are i.i.d. with $\mathrm{SExp}(s, \mu)$. 
  Distribution and the expected value of job execution time are given as
  \begin{equation}
  \begin{split}
    \Pr\{T > t\} &= \Pr\{T_e > t - s\}, \\
  	\E[T] &= s + \E[T_e],
  \end{split}
  \label{eqn:eq_k_nd_SExp_tail__ET}
  \end{equation}
  where $T_e$ is the job execution time when task execution times are distributed as $\mathrm{Exp}(\mu)$, for which the distribution and the expected value are given in Thm.~\ref{thm_k_nd_Exp_T_C}.
  
  Cost with ($C^{c}$) or without ($C$) task cancellation is given as
  \begin{equation*}
    \E[C] = 
    \begin{cases}
      n\left(s + 1/\mu\right) & \Delta \leq s, \\
      \left(k + (1-\tilde{q}^k)(n-k)\right)\left(s + 1/\mu\right) & o.w.
    \end{cases}
  \end{equation*}
  
  \begin{equation*}
    \E[C^c] = 
    \begin{cases}
      \begin{aligned}
        & k/\mu + ns - (n-k)q^k \\
        &~ \times \left(\Delta + k\mu\left(\frac{\zeta}{\mu q^k} - \Delta\left(\frac{1}{q}-1\right)\right)\right)
      \end{aligned} & \Delta \leq s, \\
      \begin{aligned}
        & (\approx)~ \E[C] - \frac{n-k}{\mu}\Bigl(1-q^k + \zeta^{-k(1-q)} \\
        &~ \times B(\zeta; k-kq+1, 0)\left(\tilde{q}^k-q^k\right) \Bigr)
      \end{aligned} & o.w.
    \end{cases}
  \end{equation*}
  where $q = \mathbbm{1}(\Delta > s)\left(1-e^{-\mu(\Delta-s)}\right)$, $\tilde{q} = 1-e^{-\mu\Delta}$ and $\zeta = 1-e^{-\mu s}$.
  
  
\label{thm_k_nd_SExp_T_C}
\end{theorem}

\begin{figure}[ht]
  \centering
  \begin{subfigure}[]{.4\textwidth}
    \centering
    \includegraphics[width=1\textwidth, keepaspectratio=true]{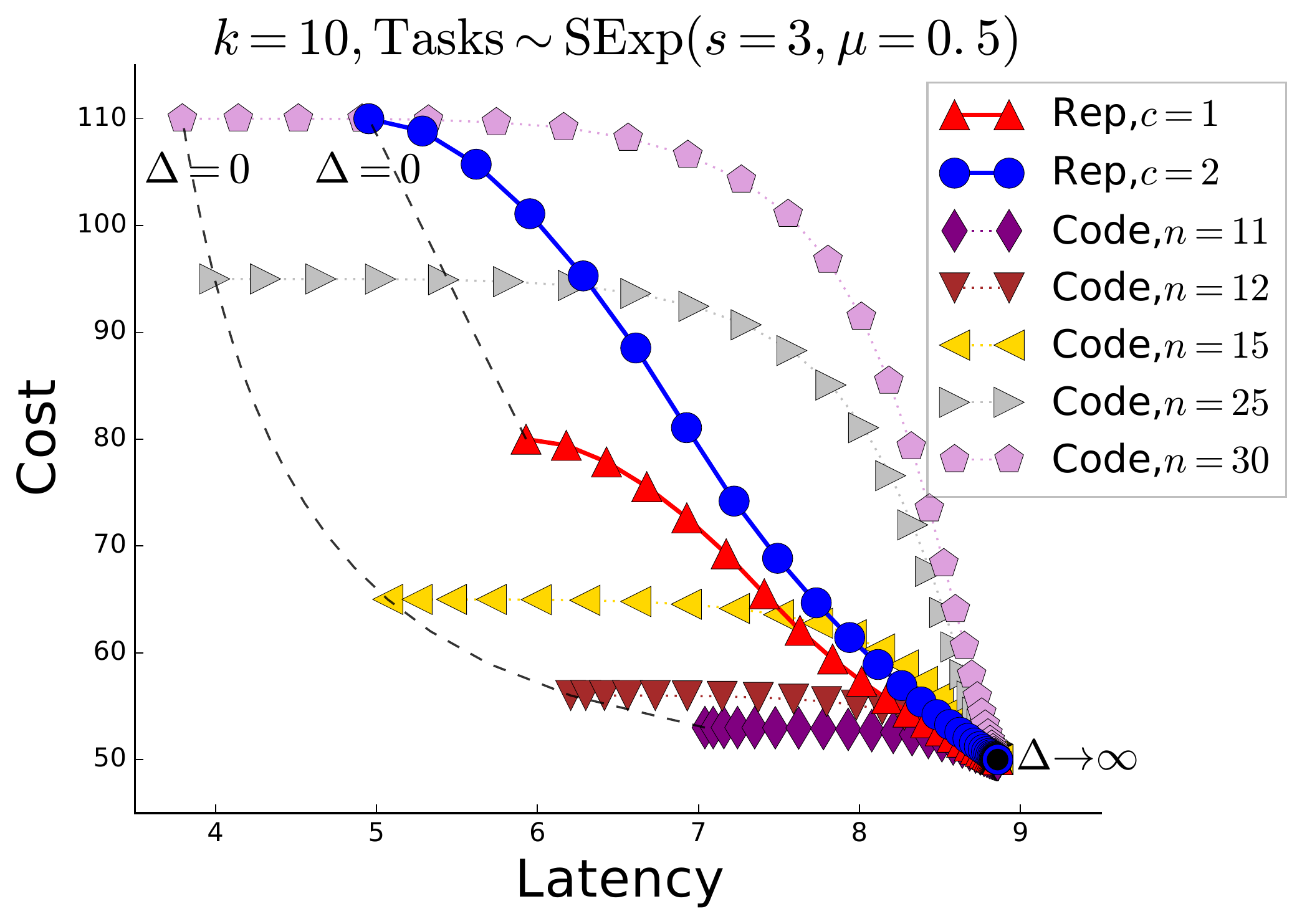}
  \label{fig:plot_EC_vs_ET_wdelay_SExp_k_10}
  \end{subfigure}
  \begin{subfigure}[]{.4\textwidth}
    \centering
    \includegraphics[width=1\textwidth, keepaspectratio=true]{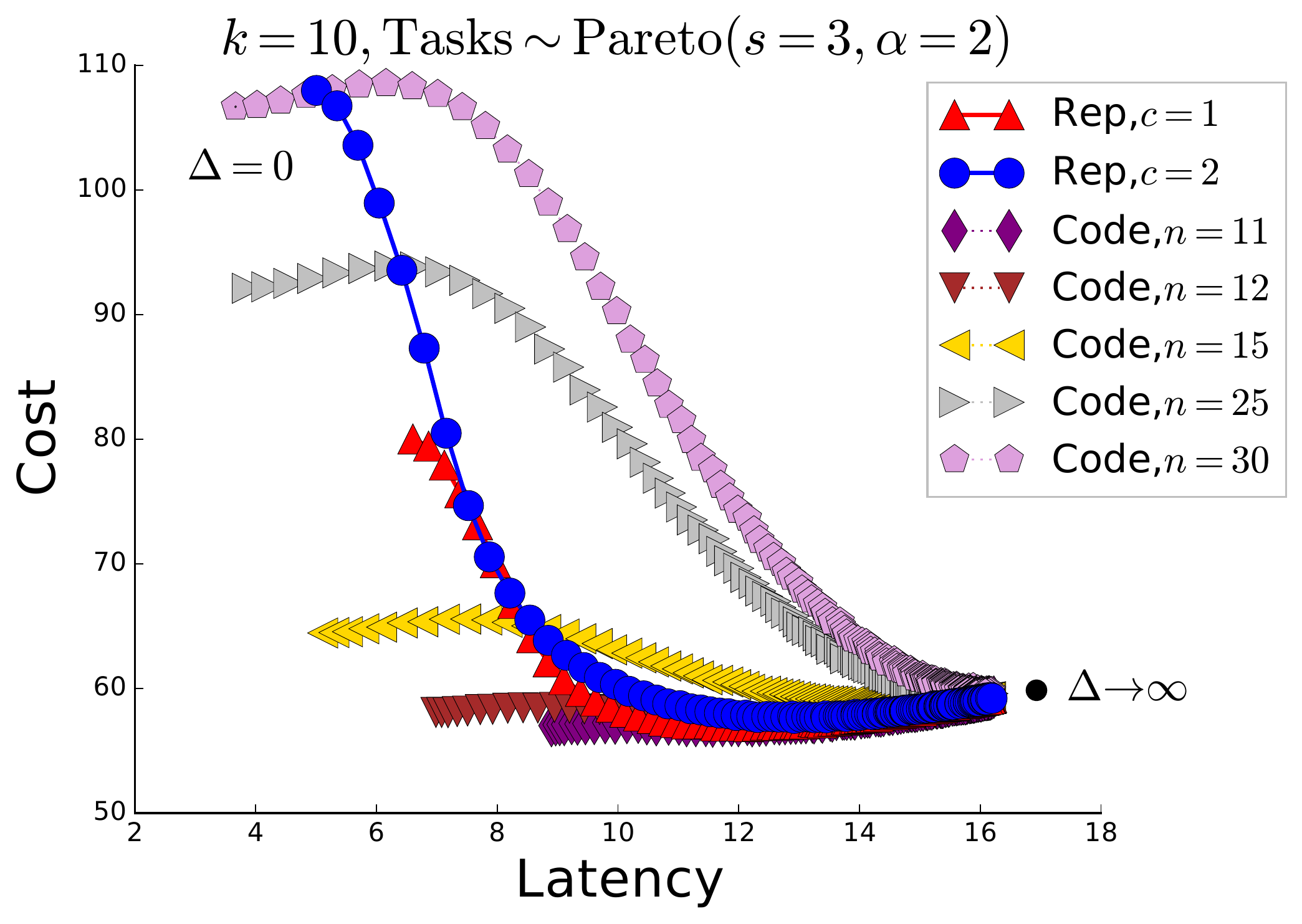}
  \label{fig:figs/plot_EC_vs_ET_wdelay_Pareto_k_10}
  \end{subfigure}
  \caption{Achievable cost (with task cancellation) vs.\ latency in executing a job of $k$ tasks with replicated ($c=1,2$) or coded ($n \in [k+1, 3k]$) redundancy. Each cost vs.\ latency curve is drawn for a fixed number of redundant tasks by varying the launch time $\Delta$ of redundant tasks. Task execution times are i.i.d. with $\mathrm{SExp}$ (Top) and $\mathrm{Pareto}$ (Bottom).}
  \label{fig:figs/plot_EC_vs_ET_wdelay}
\end{figure}

In distributed computing systems, outstanding redundant tasks can be canceled by signaling the computing nodes as soon as the job completes, hence we always refer to the cost with task cancellation in the discussions throughout the paper.

\begin{figure*}[t]
  \centering
  \begin{subfigure}[]{.32\textwidth}
    \centering
    \includegraphics[width=1\textwidth, keepaspectratio=true]{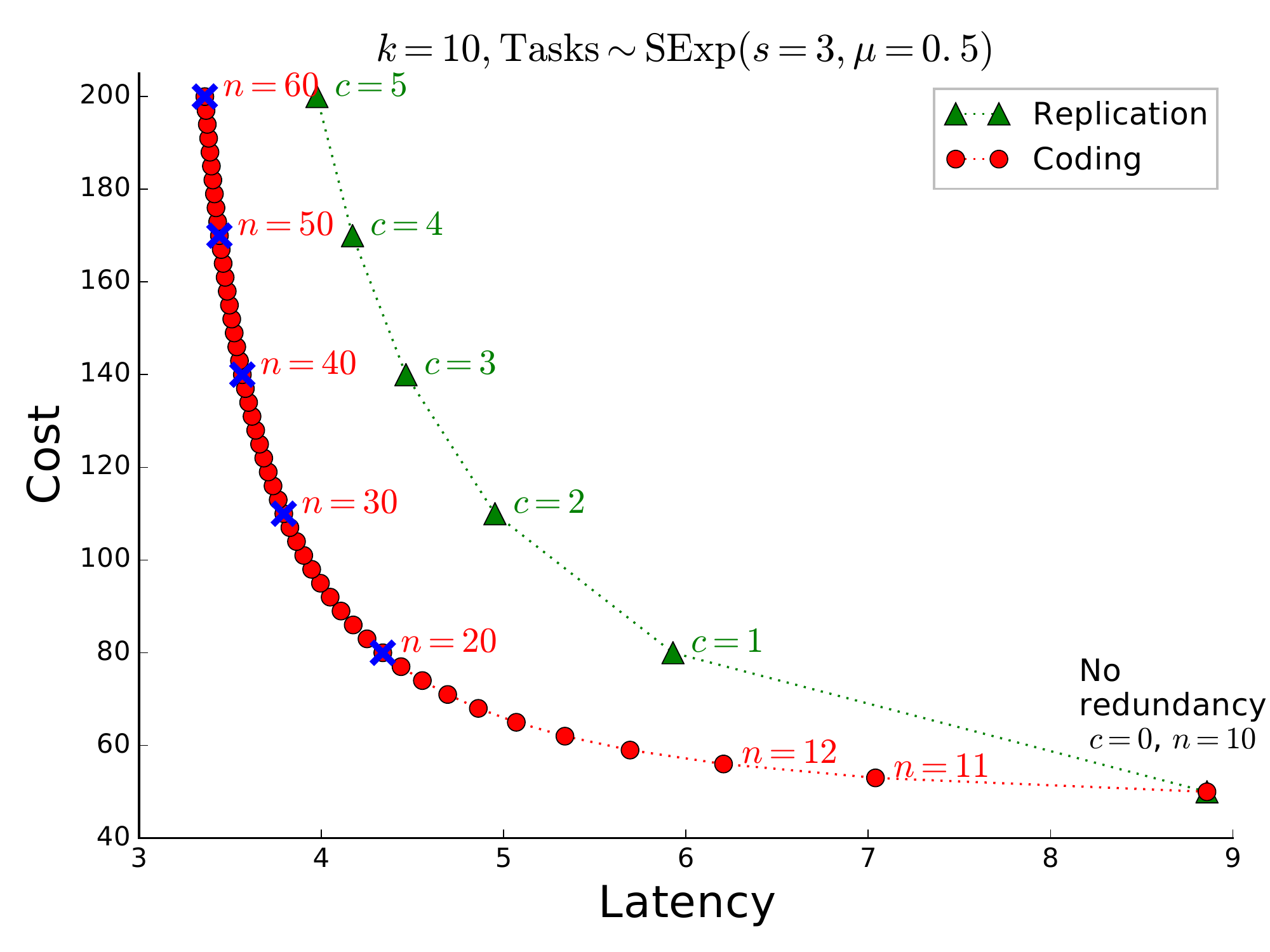}
  \end{subfigure}
  \begin{subfigure}[]{.32\textwidth}
    \centering
    \includegraphics[width=1\textwidth, keepaspectratio=true]{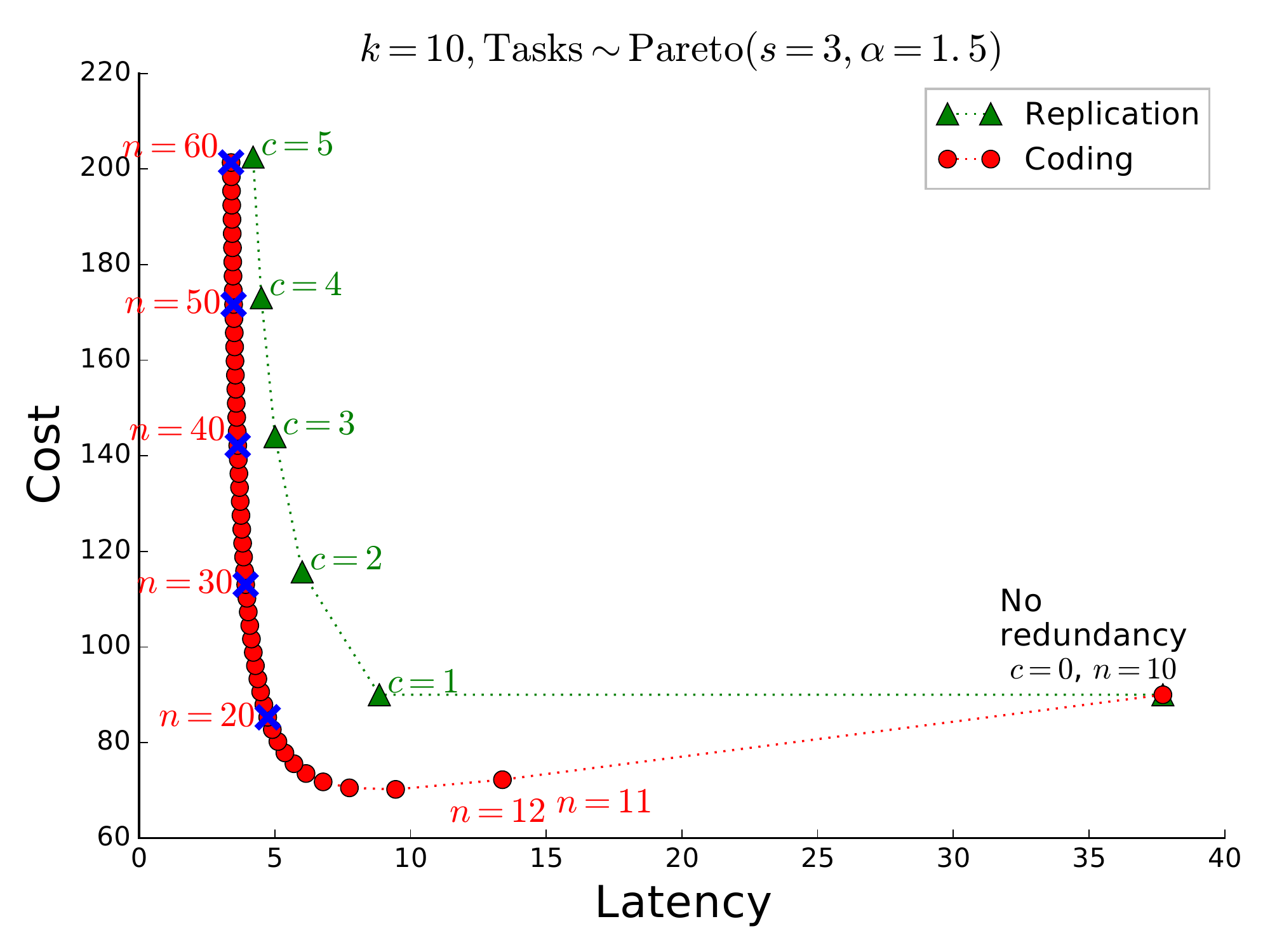}
  \end{subfigure}
  \begin{subfigure}[]{.32\textwidth}
    \centering
    \includegraphics[width=1\textwidth, keepaspectratio=true]{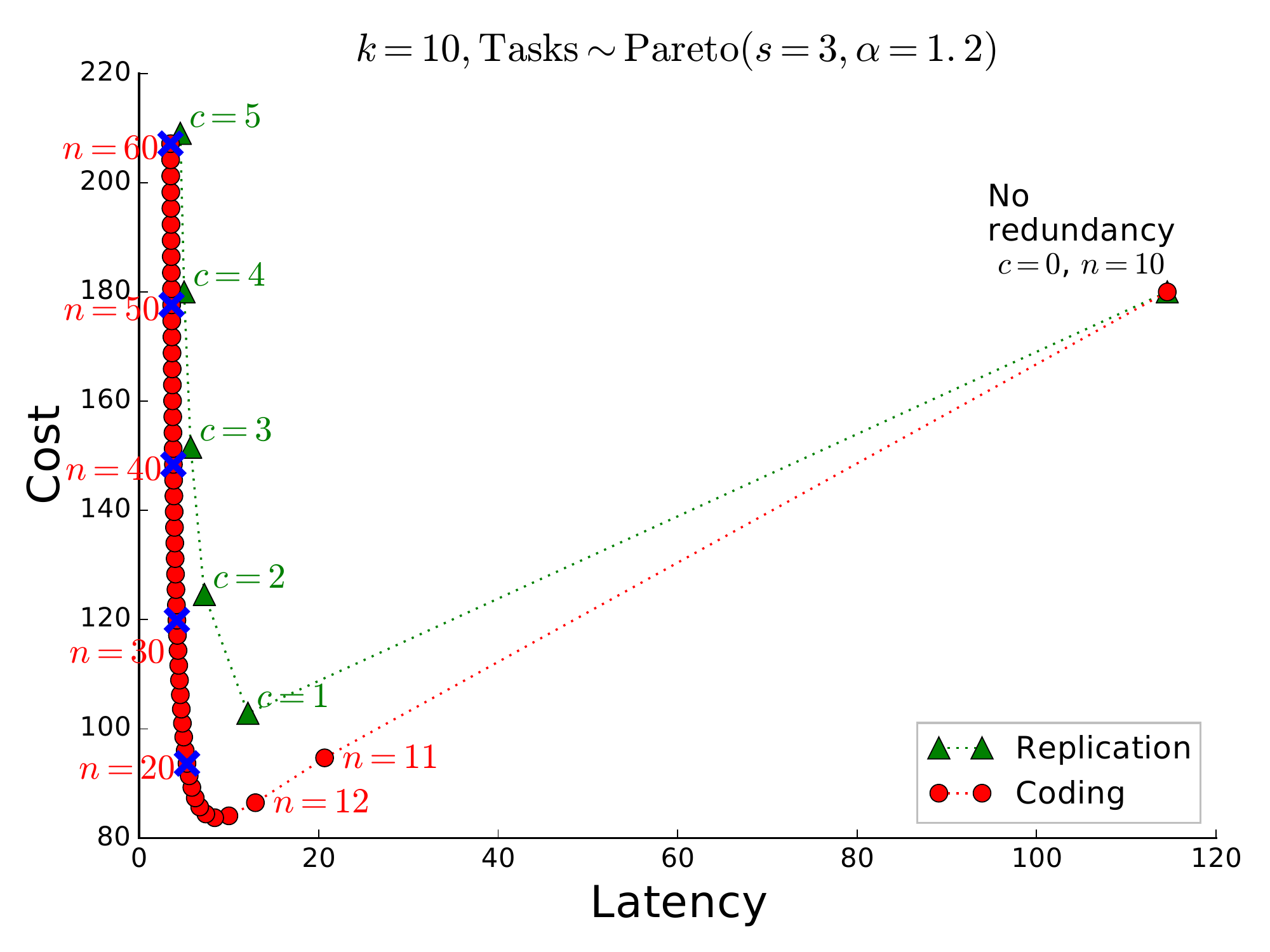}
  \end{subfigure}
  \caption{Cost vs.\ latency of executing a job of $k=10$ tasks by employing zero-delay replicated or coded redundancy. The level of employed redundancy, $c$ for replication and $n$ for coding, varies along each curve. Tail heaviness of task execution times increases from left to right.}
  \label{fig:plot_zerodelay_reped_vs_coded_k_10}
\end{figure*}

When task execution times are exponentially distributed, expressions in Thm.~\ref{thm_k_cd_Exp_T_C} and \ref{thm_k_nd_Exp_T_C} tell us that job execution cost neither depends on the time $\Delta$ at which redundant tasks are launched nor the level of employed replicated ($c$) or coded ($n$) redundancy. Therefore, according to our model with exponentially distributed task execution times, launching all the available redundant tasks at the beginning (i.e., $\Delta=0$) achieves the minimum latency with zero penalty in cost.

Recent research proposes waiting for some time before replicating the tasks to reduce the cost of redundancy \cite{StragglerRep:WangJW15}.
Using the expressions given in Thm.~\ref{thm_k_cd_SExp_T_C} and \ref{thm_k_nd_SExp_T_C}, Fig.~\ref{fig:figs/plot_EC_vs_ET_wdelay} plots the cost vs.\ latency tradeoff in executing the same job with different levels of replicated or coded redundancy, by varying the launch time $\Delta$ of the redundant tasks between $0$ and $\infty$.
First, let us focus on the case with $\mathrm{SExp}$ task execution times as shown in Fig.~\ref{fig:figs/plot_EC_vs_ET_wdelay} (Top).
Cost monotonically decreases while latency monotonically increases with $\Delta$.
Let $C(c, \Delta)$, $T(c, \Delta)$ be the cost and latency when $c$ replicas are added for each remaining task after waiting some time $\Delta$. 
Increasing $\Delta$ initially allows significant reduction in cost while causing a slight increase in latency. However, as soon as $T(c, \Delta)$ exceeds $T(c-1, 0)$ (plot shows this for $c=2$) increasing $\Delta$ further does not make sense; one can achieve less cost for the same latency by reducing $c$ rather than increasing $\Delta$.
This behavior of cost vs. latency tradeoff is more apparent when coded redundancy is employed. Increasing $\Delta$ from zero initially returns no visible cost reduction while incurring significant increase in the latency, while it does yield visible reduction in cost only after $\Delta$ reaches a certain value. In other words, adding coded tasks with delay can yield significant cost reduction only after significant sacrifice in latency.
Consider the cost and latency value at a sufficiently large value of $\Delta$ on a curve for a number of coded tasks $n-k = r > 1$, then the curve below for $n-k = r-1$ attains the same latency at less cost at a smaller value of $\Delta$.
Thus, given a job execution with a sufficiently large value of $\Delta$ and $r > 1$ coded tasks, same latency can be achieved for less cost by reducing $\Delta$ and decrementing $r$.
The same conclusions hold for the case with $\mathrm{Pareto}$ task execution times (shown in Fig.~\ref{fig:figs/plot_EC_vs_ET_wdelay} (Bottom)).


The remainder of this section is concerned with the cost vs.\ latency tradeoff when redundant tasks are launched together with the original tasks (i.e., $\Delta = 0$), which we refer to as \emph{zero-delay redundancy}.
For the case with $\Delta > 0$, we could derive the cost and latency expressions only when the distribution of task execution times, which we refer to as $X$ here, has exponential tail.
This is because in the absence of memoryless property (e.g., when $X$ is heavy tailed), derivations require working with the order statistics of samples drawn from two different distributions\footnote{This issue disappears when the remaining tasks at time $\Delta$ are relaunched. This is why in Sec.~\ref{sec:sec_straggler_relaunch}, we will be able to derive the cost and latency expressions for the case of jointly performing straggler relaunch and launching redundant tasks after waiting some time $\Delta$.}; residual execution time of the remaining task copies after time $\Delta$ is distributed as $X|X > \Delta$, while the execution time of copies that are newly launched at time $\Delta$ is distributed as $X$.
Order statistics of independent but non-identical random variables has been studied in the literature \cite{OrderStatisticsForINID:BapatB89}.
Using the results available in the literature, cost and latency expressions in the case with non-exponential $X$ could be written out, however, the expressions are unwieldy (relatively bearable for job execution with replicas compared to execution with coded tasks).
We did not pursue deriving such cumbersome expressions because our purpose was to observe the effect of $\Delta$ on the cost vs. latency tradeoff, which was very well served by the expressions we derived for the case with shifted-exponential $X$.
When $\Delta = 0$, cost and latency expressions can be derived with fairly tractable steps when $X$ has either exponential or heavy tail.
\begin{theorem}
  Let the cost (with task cancellation) and latency of executing a job of $k$ tasks be $C_c$, $T_c$ when each task is launched together with $c$ replicas, and let them be $C_n$, $T_n$ when job is launced with $n-k$ additional coded tasks.
  When task execution times are i.i.d. with $\mathrm{SExp}(s, \mu)$,
  \begin{equation*}
  \begin{split}
    \E[T_c] &= s + \frac{H_k}{(c+1)\mu}, \qquad \E[C_c] = k\left((c+1)s + \frac{1}{\mu}\right), \\
    \E[T_n] &= s + \frac{1}{\mu}(H_n-H_{n-k}), \quad \E[C_n] = ns + \frac{k}{\mu}.
  \end{split}
  \label{eqn:eq_k_c_SExp_ET__EC}
  \end{equation*}
   When task execution times are i.i.d. with $\mathrm{Pareto}(s, \alpha)$,
  \begin{equation*}
  \begin{split}
    \E[T_c] &= s k!\frac{\Gamma\left(1-1/\left((c+1)\alpha\right)\right)}{\Gamma\left( k+1-1/\left((c+1)\alpha\right)\right)}, \\
    \E[C_c] &= s k(c+1)\frac{\alpha}{\alpha-1/(c+1)}, \\
    \E[T_n] &= s\frac{n!}{(n-k)!}\frac{\Gamma(n-k+1-1/\alpha)}{\Gamma(n+1-1/\alpha)}, \\
    \E[C_n] &= s\frac{n}{\alpha-1}\left(\alpha - \frac{\Gamma(n)}{\Gamma(n-k)}\frac{\Gamma(n-k+1-1/\alpha)}{\Gamma(n+1-1/\alpha)}\right).
  \end{split}
  \label{eqn:eq_k_c_Pareto_ET__EC}
  \end{equation*}
\label{thm_k_cn_ET__EC}
\end{theorem}

Using the expressions given in Thm.~\ref{thm_k_cn_ET__EC}, Fig.~\ref{fig:plot_zerodelay_reped_vs_coded_k_10} plots the cost vs.\ latency tradeoff in executing the same job by introducing varying levels of zero-delay replicated or coded redundancy. Under both $\mathrm{SExp}$ and $\mathrm{Pareto}$ task execution times, coding always achieves less latency for the same cost compared to replication. This observation is formally stated in Thm~\ref{thm_coding_vs_rep_less_latency_cost}.
\begin{theorem}
  Consider launching a job of $k$ tasks with redundant tasks.
  Cost and latency is lower when $kc$ MDS coded tasks are added compared to adding $c$ replicas for each task.
\label{thm_coding_vs_rep_less_latency_cost}
\end{theorem}

When task execution times are light tailed, adding redundant tasks into the job reduces its latency but increases its cost. In \cite{RepedComputing:WangJW15}, replication is demonstrated to reduce the cost and latency together when task execution times are heavy tailed. Using the exact expressions in Thm.~\ref{thm_k_cn_ET__EC}, Fig.~\ref{fig:plot_zerodelay_reped_vs_coded_k_10} illustrates that redundancy can reduce cost and latency together when the tail of task execution times is \textit{heavy enough}. Reduction in the cost and latency is greater with coding compared to replication. We elaborate at the end of this section on the tail heaviness required to achieve reduction in the cost and latency together.

Although closed form expressions are formidable to derive, second moments of the cost and latency can be exactly computed as described in Thm.~\ref{thm_k_cn_ET_2__EC_2}, which enables us to compute the standard deviation of the cost and latency. Fig.~\ref{fig:plot_zerodelay_reped_vs_coded_wstdev_k_10} plots the expected cost and latency values with error bars of width equal to the standard deviation in respective dimensions. Variability in the cost and latency naturally decreases with increasing levels of redundancy. Fixing the number of added redundant tasks, coding achieves less variability compared to replication.

\begin{theorem}
  Consider launching a job of $k$ tasks with redundant tasks.
  Let us denote the cost and latency as $C_c$, $T_c$ when $c$ replicas are added for each task, and as $C_n$, $T_n$ when $n-k$ coded tasks are added.
  For $X \sim \mathrm{Exp}(\mu)$ and $j \geq i$, we have
  \begin{equation*}
  \begin{split}
    \E[X_{n:i}X_{n:j}] &= \frac{1}{\mu^2}\bigl(H_{n^2} - H_{(n-i)^2} \\
    &\qquad\quad + (H_n - H_{n-i})(H_n - H_{n-j})\bigr).
  \end{split}
  \label{eqn:eq_Exp_orderstat_joint_moment}
  \end{equation*}
  as given in \cite[Pg.~73]{OrderStat:Arnold08}. Let $Y \sim \mathrm{Exp}((c+1)\mu)$.
  When task execution times are i.i.d. with $\mathrm{SExp}(s, \mu)$, second moments of the cost and latency are given as
  \begin{equation*}
  \begin{split}
	\E[T_c^2] &= \left(s + \frac{H_k}{(c+1)\mu}\right)^2 + \frac{H_{k^2}}{(c+1)^2\mu^2}, \\
    \E[C_c^2] &= \left(k(c+1)s\right)^2 + 2k(c+1)s\frac{k}{\mu} \\
    &\quad + (c+1)^2 \sum_{i,j=1}^k \E[Y_{n:i}Y_{n:j}] \\
    \E[T_n^2] &= \frac{H_{n^2} - H_{(n-k)^2}}{\mu^2} + \left(s + \frac{H_n - H_{n-k}}{\mu}\right)^2, \\
    \E[C_n^2] &= \left(n s\right)^2 + 2ns\frac{k}{\mu} + (n-k)^2 \E[X_{n:k}^2] \\
    &\quad + 2(n-k)\sum_{i=1}^k \E[X_{n:i}X_{n:k}] + \sum_{i,j=1}^k \E[X_{n:i}X_{n:j}].
  \end{split}
  \label{eqn:eq_k_cn_SExp_ET_2__EC_2}
  \end{equation*}
  
  For $X \sim \mathrm{Pareto}(s, \alpha)$, given $\alpha > \max\{2/(n-i+1), 1/(n-j+1)\}$ and $j \geq i$, we have
  \begin{equation*}
  \begin{split}
    \E[X_{n:i}X_{n:j}] = &s^2\frac{n!}{\Gamma(n+1-2/\alpha)} \\
    &\times \frac{\Gamma(n-i+1-2/\alpha)}{\Gamma(n-i+1-1/\alpha)} \frac{\Gamma(n-j+1-2/\alpha)}{\Gamma(n-j+1)}.
  \end{split}
  \label{eqn:eq_Pareto_orderstat_joint_moment}
  \end{equation*}
  as given in \cite[Pg.~62]{Pareto:Arnold15}. Let $Y \sim \mathrm{Pareto}(s, (c+1)\alpha)$.
  When task execution times are distributed as $\mathrm{Pareto}(s, \alpha)$, second moments of the cost and latency are given as
  \begin{equation*}
  \begin{split}
	\E[T_c^2] &= \E[Y_{k:k}^2], \\
    \E[C_c^2] &= (c+1)^2 \sum_{i,j=1}^k \E[Y_{k:i}Y_{k:j}], \\
    \E[T_n^2] &= \E[X_{n:k}^2] \\
    \E[C_n^2] &= (n-k)^2 \E[X_{n:k}^2] + 2(n-k)\sum_{i=1}^k \E[X_{n:i}X_{n:k}] \\
    &\quad + \sum_{i,j=1}^k \E[X_{n:i}X_{n:j}].
  \end{split}
  \label{eqn:eq_k_cn_Pareto_ET_2__EC_2}
  \end{equation*}
  \label{thm_k_cn_ET_2__EC_2}
\end{theorem}
\begin{IEEEproof}[Proof Sketch]
  Derivations follow from the cost and latency formulation given in the proof of Thm.~\ref{thm_k_cn_ET__EC}.
\end{IEEEproof}

When task execution times are heavy tailed, it is possible to reduce latency by adding redundant tasks and still pay for the baseline cost of executing the job with no redundancy (cf. Fig.~\ref{fig:plot_zerodelay_reped_vs_coded_k_10}). We refer to this as \textit{latency reduction at no cost}.

\begin{figure}[t]
  \centering
  \includegraphics[width=0.36\textwidth, keepaspectratio=true]{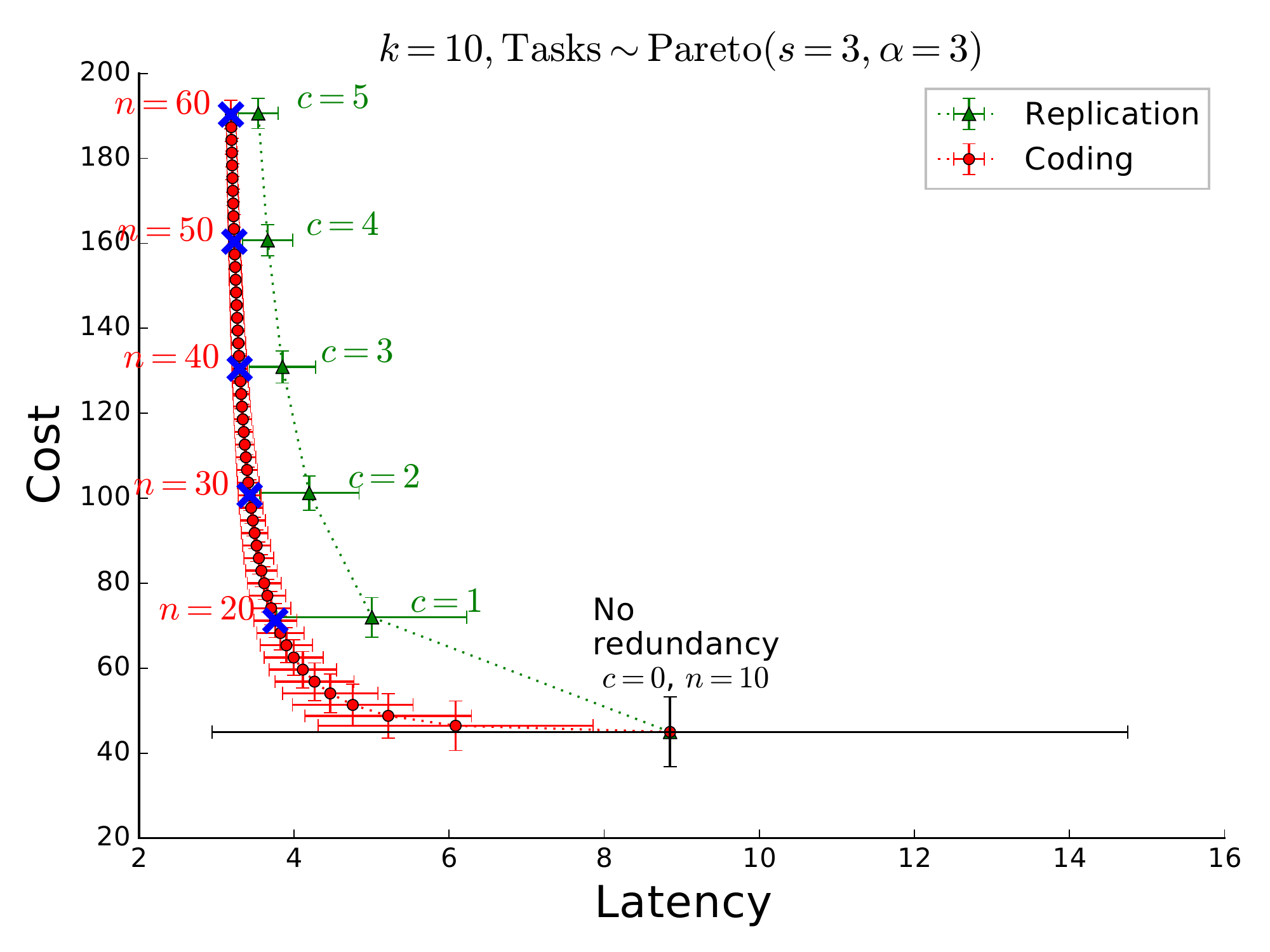}
  \caption{Cost vs.\ latency for zero-delay redundancy systems. The width of horizontal error bars is equal to the standard deviation of latency and the width of vertical bars is equal to the standard deviation of cost.}
  \label{fig:plot_zerodelay_reped_vs_coded_wstdev_k_10}
\end{figure}

\begin{corollary}
  Suppose task execution times are i.i.d. with $\mathrm{Pareto}(s, \alpha)$.
  Launching a job of $k$ tasks by adding $c$ replicas for each task can reduce its latency up to a minimum value $\E[T_{\min}]$ without incurring any additional cost if and only if $\alpha < 1.5$, and for $c_{\max} = \max\Set{\floor*{1/(\alpha-1)}-1, 0}$, we have
  \begin{equation}
    \E[T_{\min}] = s k!\frac{\Gamma\left(1-1/\left(\alpha(c_{\max}+1)\right)\right)}{\Gamma\left(k+1-1/\left(\alpha(c_{\max}+1)\right)\right)}.
  \label{eqn:eq_k_c_Pareto_reduc_in_ET_for_base_EC}
  \end{equation}
  
  A sufficient condition to reduce latency with no additional cost by adding $n-k$ coded tasks is given as
  \begin{equation}
    \alpha^{\alpha} \leq \frac{n}{n-k+1},
  \label{eqn:eq_k_n_suffcond_on_a}
  \end{equation}
  a necessary condition is given as
  \begin{equation}
    \alpha^{\alpha} \leq \frac{n+1}{n-k},
  \label{eqn:eq_k_n_neccond_on_a}
  \end{equation}
  the minimum latency at no additional cost is given as
  \begin{equation}
  \begin{split}
    \E[T_{\min}] = f(n_{\max}).
  \end{split}
  \label{eqn:eq_k_n_Pareto_reduc_in_ET_for_base_EC}
  \end{equation}
  such that
  \begin{equation*}
  \begin{split}
    & f(n) = s\frac{n!}{(n-k)!}\frac{\Gamma(n-k+1-1/\alpha)}{\Gamma(n+1-1/\alpha)}, \\
    & n_{\max} = \max\Set{n~|~f(n) - \frac{f(k)}{(n-k)} - \alpha \leq 0},
  \end{split}
  \end{equation*}
  or it is bounded as follows
  \begin{equation}
    \E[T_{\min}] < s\left(\alpha + k!\frac{\Gamma(1-1/\alpha)}{\Gamma(k+1-1/\alpha)}\right).
  \label{eqn:eq_k_n_Pareto_reduc_in_ET_for_base_EC__ineq}
  \end{equation}
\label{cor_k_cn_Pareto_reduc_in_ET_for_baseline_EC}
\end{corollary}

Fig.~\ref{plot_k_cn_pareto_reduc_in_ET_for_baseline_EC} plots the maximum relative latency reduction at no cost in executing the same job under varying degree of tail heaviness in task execution times. Maximum relative latency reduction at no cost is defined as $\left(\E[T_0]-\E[T_{\min}]\right)/\E[T_0]$ where $\E[T_{\min}]$ is the minimum possible latency at no cost, and $\E[T_0]$ is the baseline latency of executing the job with no redundancy.
As stated in Cor.~\ref{cor_k_cn_Pareto_reduc_in_ET_for_baseline_EC}, when the employed redundancy is replication, latency reduction at no cost is possible only when the tail index of task execution times is less than $1.5$, that is, only when the tail of task execution times is quite heavy. Employing coded redundancy relaxes this requirement on the tail heaviness, as also shown in the plot. 
When the employed redundancy is replication, the tail heaviness requirement is independent of the number of tasks $k$ that constitute the job, while employing coded redundancy relaxes the requirement on the tail index further at larger $k$, i.e., the upper threshold on the tail index increases with $k$.
This can be explained as follows. A task replica can only replace its original copy, while a coded task can replace any of the $k$ initial tasks. Thus, coded tasks can mitigate stragglers more effectively when the job is executed at higher scale (larger $k$), while the effectiveness of task replicas is not associated with the scale of execution. Consequently for jobs that run at higher scale, coding can reduce latency at no cost even under lighter tailed task execution times, while the scale of execution does not change the tail heaviness requirement for replication.
\begin{figure}[ht]
  \centering
  \includegraphics[width=0.4\textwidth, keepaspectratio=true]{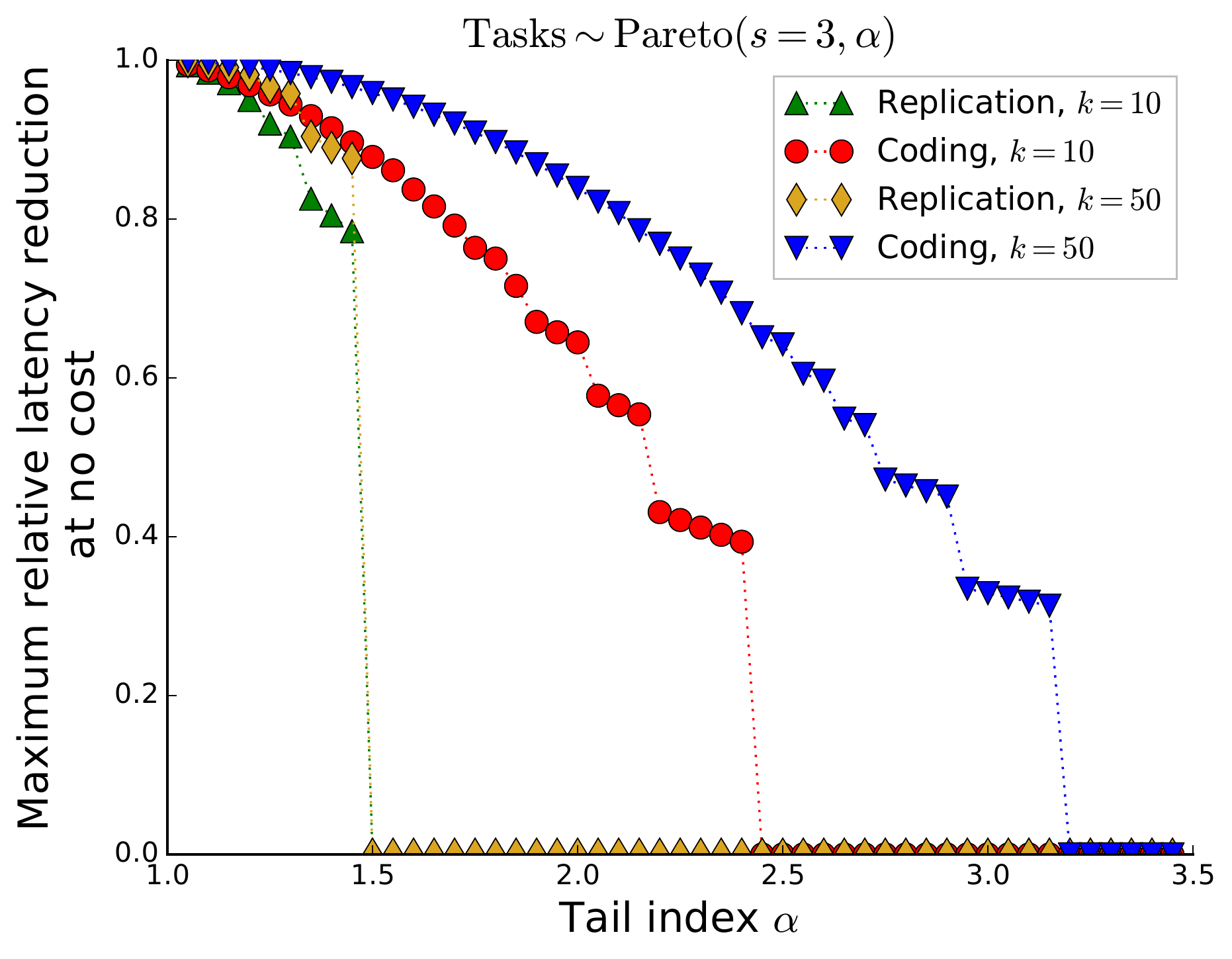}
  \caption{Maximum relative latency reduction at no cost by employing replicated or coded redundancy vs. the tail of task execution times.}
  \label{plot_k_cn_pareto_reduc_in_ET_for_baseline_EC}
\end{figure}

\begin{figure*}[t]
  \centering
  \begin{subfigure}[]{.32\textwidth}
    \centering
    \includegraphics[width=1\textwidth, keepaspectratio=true]{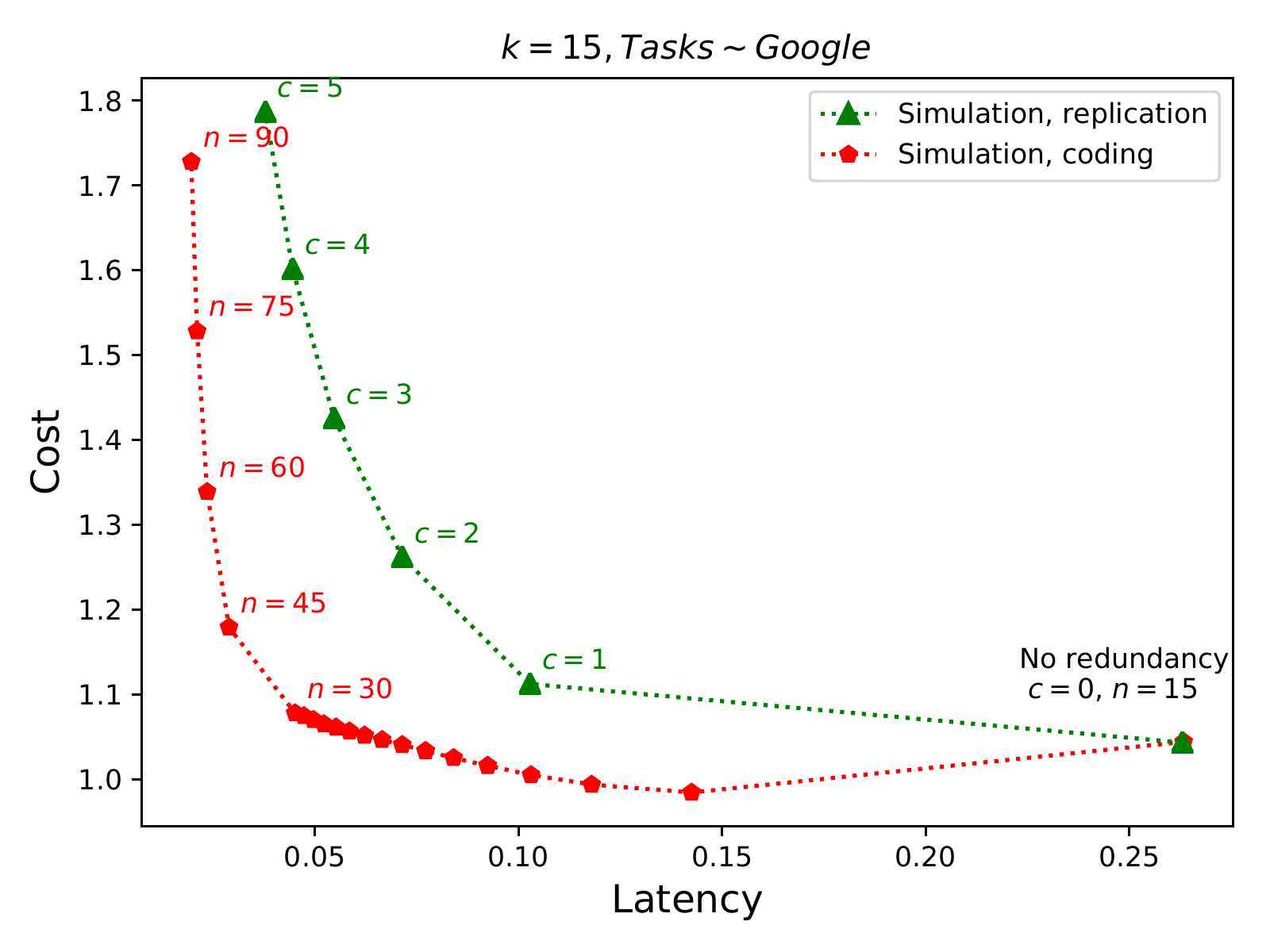}
  \end{subfigure}
  \begin{subfigure}[]{.32\textwidth}
    \centering
    \includegraphics[width=1\textwidth, keepaspectratio=true]{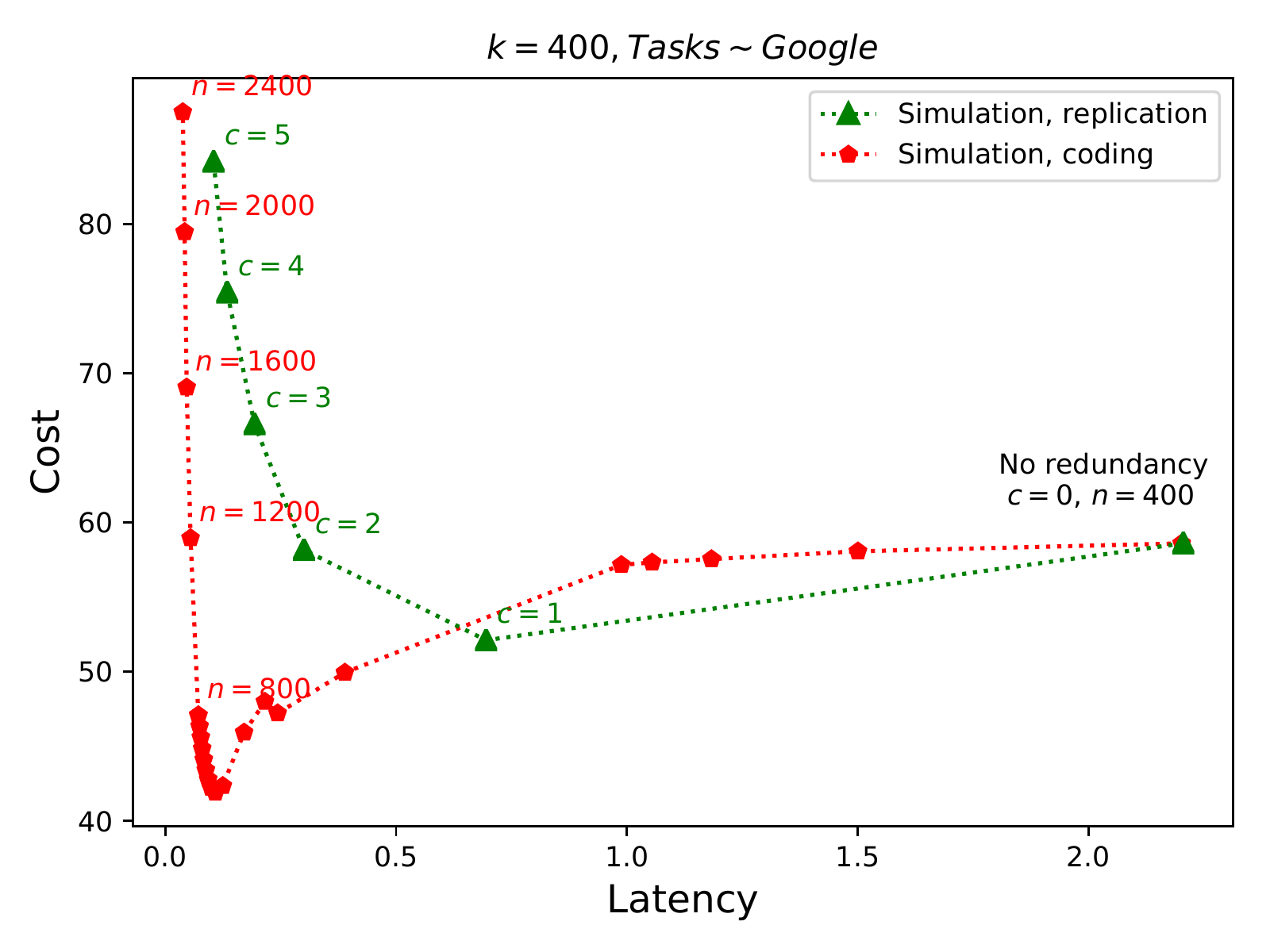}
  \end{subfigure}
  \begin{subfigure}[]{.32\textwidth}
    \centering
    \includegraphics[width=1\textwidth, keepaspectratio=true]{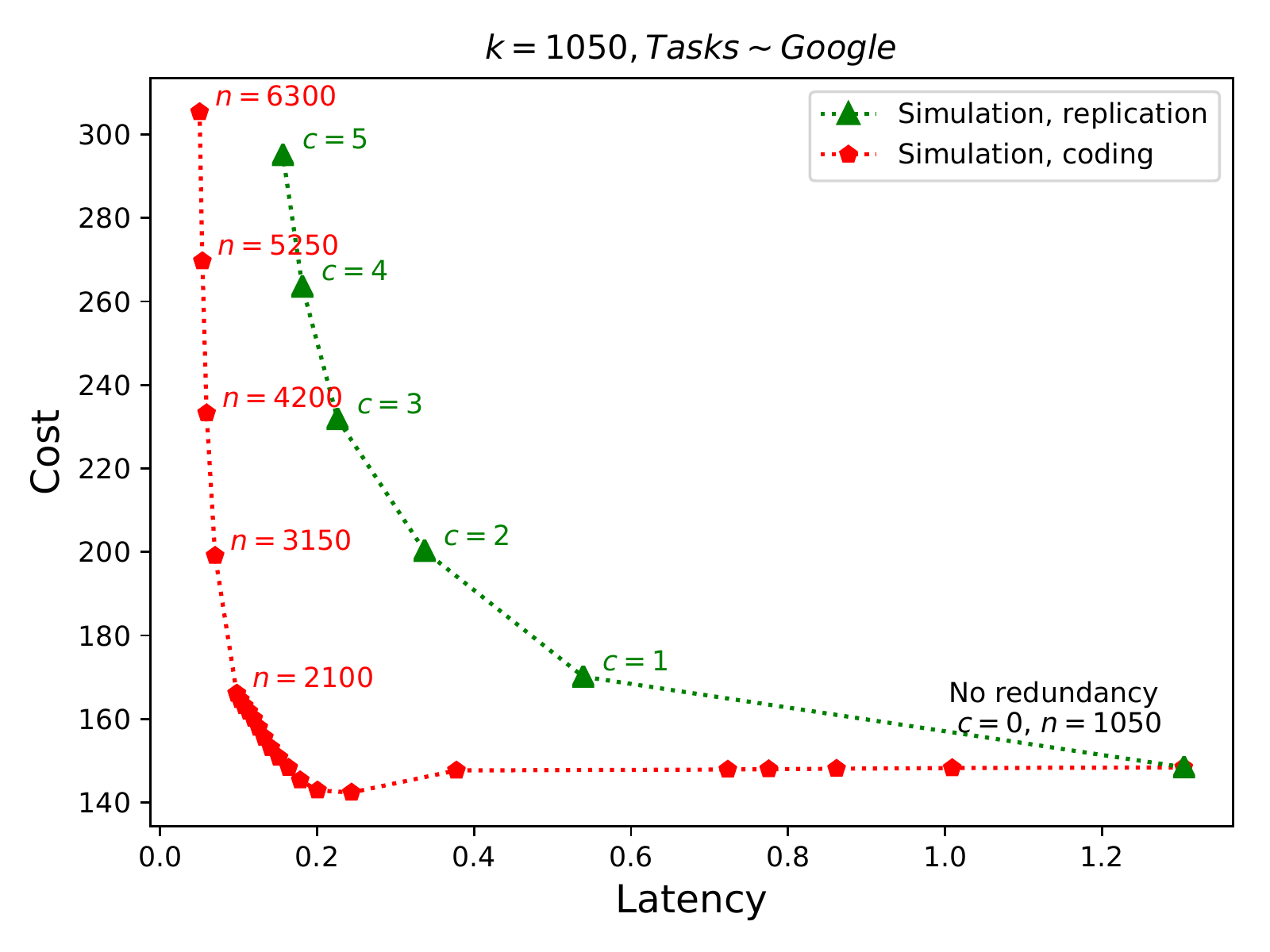}
  \end{subfigure}
  \caption{Simulated cost vs.\ latency curves for executing jobs with $k=15, 400, 1050$ tasks by employing zero-delay replicated or coded redundancy. Task execution time distributions used in the simulations are extracted from a Google Cluster Trace data \cite{GoogleTraceAnalysis:ReissTG12}.}
  \label{fig:figs/plot_reped_vs_coded_Google}
\end{figure*}

\vspace{1ex}
\noindent\textbf{Demonstration using Google cluster data.} We simulated job executions with replicated or coded redundancy by using task execution time distributions extracted from a Google Cluster data \cite{GoogleClusterTrace:ReissWH11}. 
Google released this data from a cluster running a mixed workload of short or long running MapReduce batch jobs, services and interactive queries \cite{GoogleTraceAnalysis:ReissTG12}.

Fig.~\ref{fig:figs/plot_reped_vs_coded_Google} plots the cost vs. latency curves using the three empirical distributions for jobs with $k=15, 400, 1050$ tasks that were previously illustrated in Fig.~\ref{fig:figs/plot_google_empiricaltail}.
In all three, coding is doing better than replication in the cost vs. latency tradeoff. Execution with redundancy could reduce the cost and latency together because each of these distributions pronounces heavy tail at large values (cf. Fig.~\ref{fig:figs/plot_google_empiricaltail}). In the execution of jobs with $15$ or $1050$ tasks, employing replication does not allow for latency reduction at no cost but coding does. In the execution of job with $400$ tasks, although replication seems to achieve less cost and latency at first, coding outperforms replication beyond a certain level of redundancy.

\section{When redundancy changes the tail}
\label{sec:sec_when_red_changes_tail}
So far we have ignored the impact of redundancy on the system.
Redundant tasks exert extra load on the system, which is likely to aggravate the existing contention in the system resources.
Given that resource contention is the primary cause of runtime variability \cite{TailAtScale:DeanB13}, the added redundant tasks are likely to increase the variability in task execution times.
Compute servers are typically shared by the tasks of jobs that simultaneously execute on the cluster \cite{Kubernetes:BurnsGO16}.
Two canonical server sharing strategies are
1) Processor sharing: tasks time-share the server according to a round-robin scheduling,
2) Queueing: tasks wait in a queue and are accepted into service one at a time.
Modern Operating Systems implement a mix of processor sharing and First-come First-served (FCFS) queueing to host multiple processes on a server, e.g., scheduling classes SCHED\_FIFO and SCHED\_RR in the Linux Kernel \cite{UnderstandingLinuxKernel:BovetC05}.
A compute server in reality hosts several shared resources (e.g., CPU, memory, I/O bus, etc.) and each with its own scheduling scheme. For simplicity, we here model servers to host only CPU.
We adopt \emph{limited processor sharing} model in which tasks are allowed to time-share the server (while being served over multiple CPU cores or threads) until a limited number of them accumulate, beyond which the remaining tasks wait in a FCFS queue.
Limited processor sharing is shown to implement robust performance (in terms of the tail of response time) for both heavy and light tailed task sizes \cite{LimitedPS:Nair10}.

In order to understand the impact of added redundancy on the system's runtime variability, we simulated a cluster of servers, each implementing a limited processor sharing queue.
Jobs of varying number of tasks and size (minimum task execution time) arrive to cluster according to a Poisson process.
Distribution of task sizes and number of tasks within real compute jobs are known to exhibit heavy tail \cite{HeavyTailedJobs:Leland86, HeavyTailedJobs:Harchol97, GoogleClusterDataAnalysis:ChenGG10, GoogleTraceAnalysis:ReissTG12}.
Therefore in our simulation:
i) Task size for each arriving job is independently sampled from a Truncated-Pareto (a canonical continuous heavy tailed) distribution with minimum value of $1$, maximum value of $10^{10}$ and tail index of $1.1$. The choice of Truncated-Pareto distribution and the values for its parameters come from the distribution of real compute task sizes presented in \cite{UnfairnessSRPT:BansalH01}.
ii) Number of tasks that constitute each arriving job is independently sampled from a Zipf (a canonical discrete heavy tailed) distribution. 
Each arriving job is expanded with the same rate $r > 1$; a job of $k$ tasks gets expanded into $n = \floor{rk}$ tasks by adding $\floor{rk} - k$ coded tasks, and the resulting $n$ tasks are dispatched to the $n$ servers with the least number of tasks in the cluster.
As soon as any $k$ of the $n$ tasks of a job is completed, the job completes and its remaining $n-k$ outstanding tasks (either in service or waiting in a queue) get immediately removed from the cluster.
Expanding jobs with the same rate $r$ ensures fairness by introducing redundancy in proportion to the scale $k$ at which a job is executed.

Cost and latency values for a particular type of job with a fixed number of tasks of unit size are plotted in Fig.~\ref{fig:plot_EC_vs_ET_red_affects_load} for increasing values of $r$. Each simulated server in the cluster implements limited processor sharing queue with a limit of 8 tasks.
The latency of the job is the time span between its arrival and departure to and from the system. The cost of the job is the sum of the service time of every task involved in its execution.
Simulated curve shows that redundancy initially reduces latency significantly with little change in cost, then reduces latency but increases cost, and finally beyond a level increases latency with little change in cost.
In order to evaluate the appropriateness of modeling task execution times with canonical heavy tailed distributions, we first fitted Pareto and Truncated-Pareto distributions on the task execution times sampled from the simulation, then substituted these fitted models in the analytical cost and latency expressions.
We presented cost and latency expressions for Pareto task execution times in Thm.~\ref{thm_k_cn_ET__EC}. Cost and latency are formidable to derive in closed form for Truncated-Pareto task execution times, but their computation involves a single integral which we evaluate numerically (refer to \cite[Pg.~63]{Pareto:Arnold15}).
Parameters of the Pareto (minimum value and tail index) and Truncated-Pareto (minimum and maximum values, and tail index) models are estimated using the unbiased MLE estimators that are respectively presented in \cite{ParetoEstimation:Rytgaard90} and \cite[Thm.~1]{TParetoEstimation:AbanMP06}.

The comparison given in Fig.~\ref{fig:plot_EC_vs_ET_red_affects_load} between the simulated and fitted values of cost and latency shows that modeling task execution times with Pareto distribution is fairly appropriate to study the cost vs. latency tradeoff.
This is not surprising since the asymptotic approximations of the tail of waiting times in FCFS or processor sharing queues have demonstrated that heavy tailed task sizes result in heavy tailed delay~\cite{QueueingWithHeavyTails:Zwart01, MG1AsymTail:Sakurai04, MG1HeavyTailAsymp:OlveraBG11}.
However one caveat of the model is that it cannot capture the case we observe in (the top right of) Fig.~\ref{fig:plot_EC_vs_ET_red_affects_load} in which adding more redundancy increases latency with little change in cost.
In the remainder of this section, we study the cost vs.\ latency tradeoff by adopting a Pareto task execution time model that is dependent on the rate $r$ at which redundancy is added into all the jobs executing in the system.
Expansion of a job with task replicas at (an integer) rate of $r$ refers to launching $r-1$ replicas for each of the $k$ tasks within the job.
\begin{figure}[t]
  \centering
  \includegraphics[width=0.45\textwidth, keepaspectratio=true]{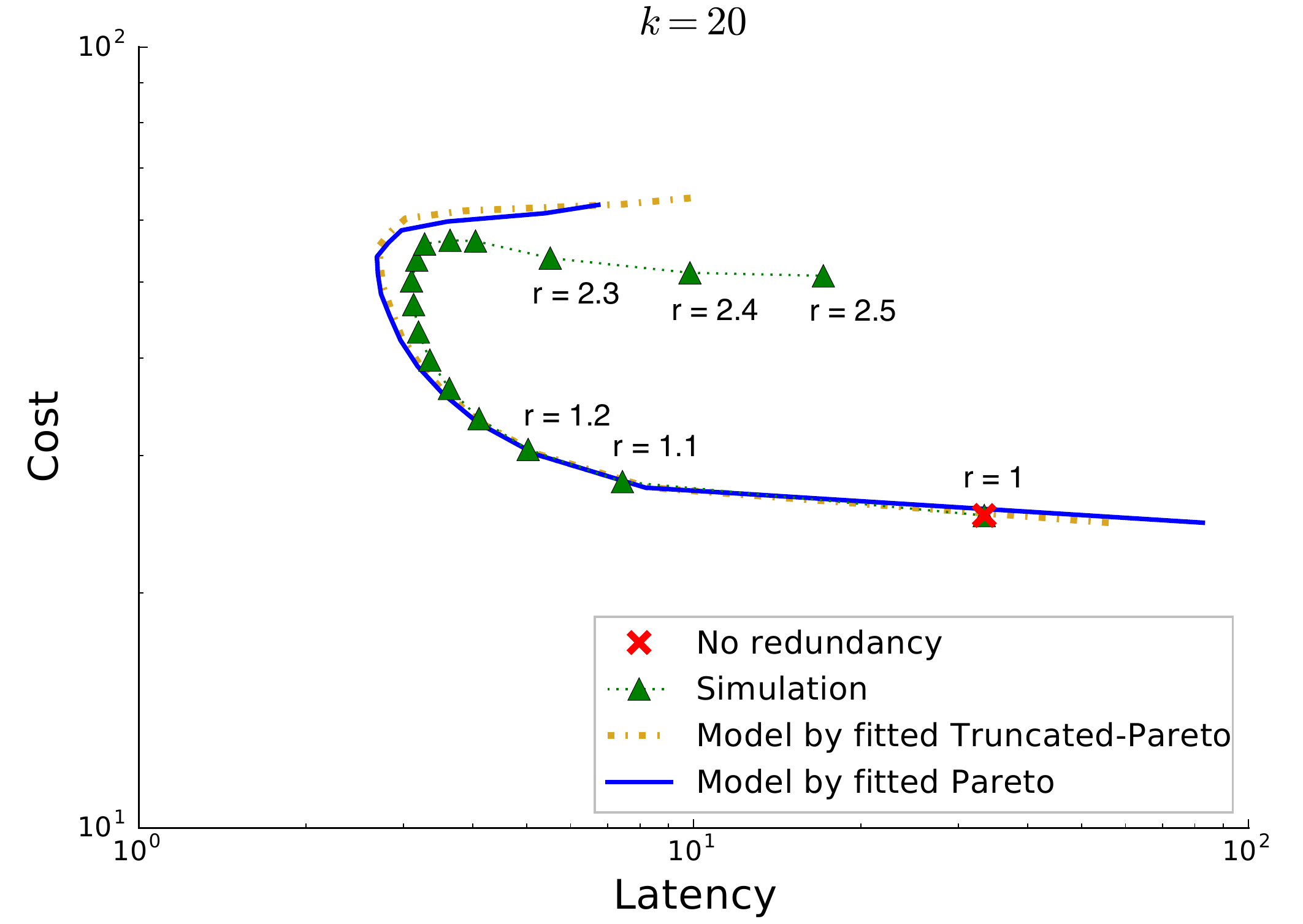}
  \caption{Cost vs. latency for a particular type of job with 20 tasks of unit size.
  Arriving jobs are expanded with coded tasks at a multiplicative factor of $r$.
  Simulated values are given for increasing values of $r$; we start with $r=1$ (No redundancy) and then increase $r$ by $0.1$ at each step.}
\label{fig:plot_EC_vs_ET_red_affects_load}
\end{figure}

Redundant tasks added into the system are expected to aggravate resource contention, and consequently increase the variability in task execution times.
Therefore, the impact of redundant load exerted on the system should be incorporated in the $\mathrm{Pareto}(s, \alpha)$ distribution that we use to model task execution times.
Under stability, an arriving job, with nonzero probability, can find the system empty and complete execution without having to share servers with any other job. Thus, we assume that minimum task execution time $s$ solely reflects the task size and is not affected by resource contention.
Then, the impact of added redundant load should be captured by the only remaining parameter, the tail index $\alpha$. Smaller $\alpha$ implies greater variability (implying greater chance and impact of resource contention), so $\alpha$ is expected to get smaller as more redundancy is added into the jobs, which is indeed what we observe in the simulations. We directly use the job expansion rate $r$ to quantify the level of added redundancy and model $\alpha$ as a function of $r$.
Note that we do not study the exact trend which describes how $\alpha$ changes with $r$, but rather try to understand the requirements on the relationship between $\alpha$ and $r$ that leads to gain or pain in the cost vs.\ latency tradeoff.
We firstly present sufficient conditions in terms of $\alpha$ and $r$ to yield a reduction or incur an increase in latency.
\begin{theorem}
  Suppose that task execution times are i.i.d. with $\mathrm{Pareto}$ with tail index $\alpha_i$ when jobs arriving to the system are expanded with redundant tasks by a multiplicative factor of $r_i > 1$. 
  Consider increasing $r_i$ to $r_j$.
  If jobs are expanded with coded tasks, a sufficient condition to reduce the latency of a job of $k$ tasks by the change $r_i \to r_j$ is
  \begin{equation}
    \alpha_i/\alpha_j \leq \log\left(\frac{n_i}{n_i-k+1}\right)/\log\left(\frac{n_j+1}{n_j-k}\right),
  \label{eq:eq_suffcond_ETred_coding}
  \end{equation}
  a sufficient condition to incur an increase in job's latency is
  \begin{equation}
    \alpha_i/\alpha_j \geq \log\left(\frac{n_i+1}{n_i-k}\right)/\log\left(\frac{n_j}{n_j-k+1}\right),
  \label{eq:eq_suffcond_ETinc_coding}
  \end{equation}
  where $n_i = \floor{k r_i}$ and $n_j = \floor{k r_j}$.
  If jobs are expanded with task replicas, a necessary and sufficient condition for the change $r_i \to r_j$  to reduce latency is given for any job as
  \begin{equation}
    \alpha_i/\alpha_j < r_j/r_i.
  \label{eq:eq_necessandsuffcond_ETred_rep}
  \end{equation}
\label{thm_suffcond_ETgainpain}
\end{theorem}

Condition \eqref{eq:eq_suffcond_ETred_coding} for the case of expanding jobs with coded tasks is sufficient to reduce latency, but it may not give tight guarantees.
It can be made easier to interpret by expressing the expansion rate $r$ as $n/k$ for a given job of $k$ tasks. Increasing the rate from $n/k$ to $(n+1)/k$, the condition \eqref{eq:eq_suffcond_ETred_coding} becomes
\[ \frac{\alpha_n}{\alpha_{n+1}} \leq \log\left(\frac{n}{n-k+1}\right)/\log\left(\frac{n+2}{n-k+1}\right) < 1. \]
This says that if the tail heaviness of task execution times (or $\alpha$) stays the same or becomes lighter as $r$ increases, increasing $r$ reduces latency for all jobs regardless of $k$.
This is not informative since we already know that latency monotonically decreases in $n$ when the tail heaviness of task execution times stays the same let alone when it gets lighter (cf.\ Thm.~\ref{thm_k_cn_ET__EC}).

Next we derive an \emph{approximate} necessary and sufficient condition to reduce latency of a particular job by increasing $r$, in the case where jobs are expanded with coded tasks. Presented approximation yields close estimates for large enough values of $r$, in particular when $r > 2$.
Approximating the quotient of Gamma functions with Sterling's approximation \cite{AsymptoticApproxOfQuotientOfGamma:Tricomi51}, latency of executing a job of $k$ tasks in a system with coded expansion rate $r = n/k$ is approximately given as
\[ \E[T_n] \approx s\left(1 + \frac{k}{n-k+1}\right)^{1/\alpha_n}, \]
which gives us the following approximation for the ratio
\[ \frac{\E[T_{n+1}]}{\E[T_n]} \approx \left(1 + \frac{k}{n-k+2}\right)^{1/\alpha_{n+1}}\left(1 + \frac{k}{n-k+1}\right)^{-1/\alpha_n}. \]
This gives us the following approximate necessary and sufficient condition on the growth of tail index to reduce the latency for jobs of $k$ tasks by increasing $r$ from $n/k$ to $(n+1)/k$,
\[ \frac{\E[T_{n+1}]}{\E[T_n]} \lesssim 1 \iff \frac{\alpha_{n+1}}{\alpha_n} \gtrsim \frac{\log(1 + k/(n-k+2))}{\log(1 + k/(n-k+1))}. \]

The condition above and the ones given in Thm.~\ref{thm_suffcond_ETgainpain} are quantitative expressions of our intuition; when the redundant load exerted on the system increases the runtime variability, it gets harder to reduce latency by 
executing jobs with more redundancy as the level of employed redundancy gets higher.
When jobs are expanded with task replicas, the condition to reduce latency with more redundancy does not depend on the number of tasks $k$ (scale) within the job; higher level of replication achieves less latency as long as the relative growth in the job expansion rate $r$ is larger than the relative reduction in the tail index $\alpha$ (i.e., relative growth in tail heaviness) of task execution times.
In Sec.~\ref{sec:sec_coding_vs_rep}, coded redundancy is shown to be more effective for jobs that run at greater scale. Similarly here when jobs are expanded with coded tasks, increased runtime variability due to redundant load can be better compensated by jobs that run at greater scale.
In addition, the threshold for redundancy to start incurring higher latency grows at a slower rate in $r$ when coded tasks are used compared to using task replicas. This is due to the fact that coded redundancy is more efficient; it yields greater reduction in latency per introduced redundant task compared to replication.

\section{Straggler Relaunch}
\label{sec:sec_straggler_relaunch}
Throughout this section, we assume task execution times are heavy tailed.
There are two properties of heavy tailed task execution times that greatly affect the distributed job execution \cite{PerfEvalWithHeavyTails:Crovella01}.
Firstly, the longer a task has taken to execute, the longer its average residual lifetime is expected to be.
Secondly, the majority of the mass in a set of sample observations drawn from a heavy tailed distribution is contributed by few samples.
This suggests that among all tasks within a job, few of them are expected to be stragglers with much longer completion time compared to the non-stragglers.

After launching a job, let us wait for a reasonably large $\Delta$ amount of time and check whether the job is completed or not. If the job is still running, we expect only a few tasks remaining which we refer to as stragglers.
Heavy tailed nature of the task execution times suggests that the tasks straggling beyond time $\Delta$ are expected to take at least $\Delta$ more to complete on average. It also suggests that if a fresh copy is launched at time $\Delta$ for each straggling task, fresh copies are likely to complete before their corresponding old copies.

In this section, we show that \textit{straggler relaunch}, that is, replacing the straggling tasks with fresh copies after waiting for some time, can yield significant reduction in cost and latency when the task execution times are \textit{heavy tailed enough}. We investigate the level of tail heaviness required for straggler relaunch to be effective.
The selection of the tasks to be relaunched is decided by the time $\Delta$ we wait before relaunching the remaining tasks.
Untimely relaunch might be either late and cause delayed cancellation of the stragglers, or might be early and cause killing the non-straggler tasks as well. We find an approximation for the optimal time to perform straggler relaunch, which turns out to have a simple and insightful form.
Lastly, we consider  performing straggler relaunch jointly with adding redundant tasks into the job execution.

Exact expressions for the cost and latency of job execution with straggler relaunch are given in Thm.~\ref{thm_k_wrelaunch_T_C}. Note that we assume relaunching tasks takes place instantly and does not incur any additional delay.
Performing straggler relaunch before the minimum task completion time $s$ causes meaningless work loss and further delays the job completion, while performing straggler relaunch at the right time significantly reduces the latency. The cost is a direct function of the latency in the absence of redundant tasks, hence reduced latency implies reduced cost as well (as illustrated in Fig.~\ref{fig:fig_k_wrelaunch_ET__EC}).

\begin{theorem}
  Suppose task execution times are i.i.d. with $\mathrm{Pareto}(s, \alpha)$.
  Consider executing a job of $k$ tasks by relaunching all the remaining tasks after waiting some time $\Delta$.
  Then, the distribution of job completion time is given as
  \begin{longaligned}[\label{eqn:eq_k_wrelaunch_tail}]
    \Pr\{T &> t\} = 1 - \left(\mathbbm{1}(t > s)\left(1 - (s/t)^{\alpha}\right)\right)^k \longalignedtag \\
      &+ \left(q + \mathbbm{1}(t > \Delta)\left(1 - (\Delta/t)^{\alpha}\right)(1-q)\right)^k \\
      &+ \left(q + \mathbbm{1}(t > \Delta+s)\left(1 -\left(s/(t-\Delta)\right)^{\alpha}\right)(1-q)\right)^k.
  \end{longaligned}
  
  Latency is given as
  \begin{equation}
    \E[T] = 
    \begin{cases}
      \Delta + L & \Delta \leq s, \\
      \begin{aligned}
        & \Delta(1-q^k) + L\bigl((s/\Delta-1) \\
        &\qquad \times I(1-q; 1-1/\alpha, k) + 1\bigr)
      \end{aligned} & o.w.
    \end{cases}
  \label{eqn:eq_k_wrelaunch_tail__ET}
  \end{equation}
  
  Cost with ($C^c$) or without ($C$) task cancellation is given as
  \begin{longaligned}[\label{eqn:eq_k_wrelaunch_EC}]
    &\E[C^c] = 
      \begin{cases}
        k\Delta + \frac{1}{\alpha-1}\left(ks\alpha - L\right) + k(1-q)\Delta & \Delta \leq s, \\
        \frac{\alpha}{\alpha-1}\left(k(1-q)(s-\Delta) + ks\right) & o.w.
      \end{cases} \\
    &\E[C] = 
      \begin{cases}
        k\Delta + ks\frac{\alpha}{\alpha-1} & \Delta \leq s, \longalignedtag \\
        \frac{\alpha}{\alpha-1}\left(ks(2-q)\right) - \frac{k\Delta(1-q)}{\alpha-1} & o.w.
      \end{cases}
  \end{longaligned}
  where $q = \mathbbm{1}(\Delta > s)\left(1 - (s/\Delta)^{\alpha}\right)$, and $L = s k!\Gamma(1-1/\alpha)/\Gamma(k+1-1/\alpha)$ is the baseline latency of executing the job without straggler relaunch.
\label{thm_k_wrelaunch_T_C}
\end{theorem}

\captionsetup[subfigure]{labelformat=empty}
\begin{figure}[t]
  \centering
  \includegraphics[width=.4\textwidth, keepaspectratio=true]{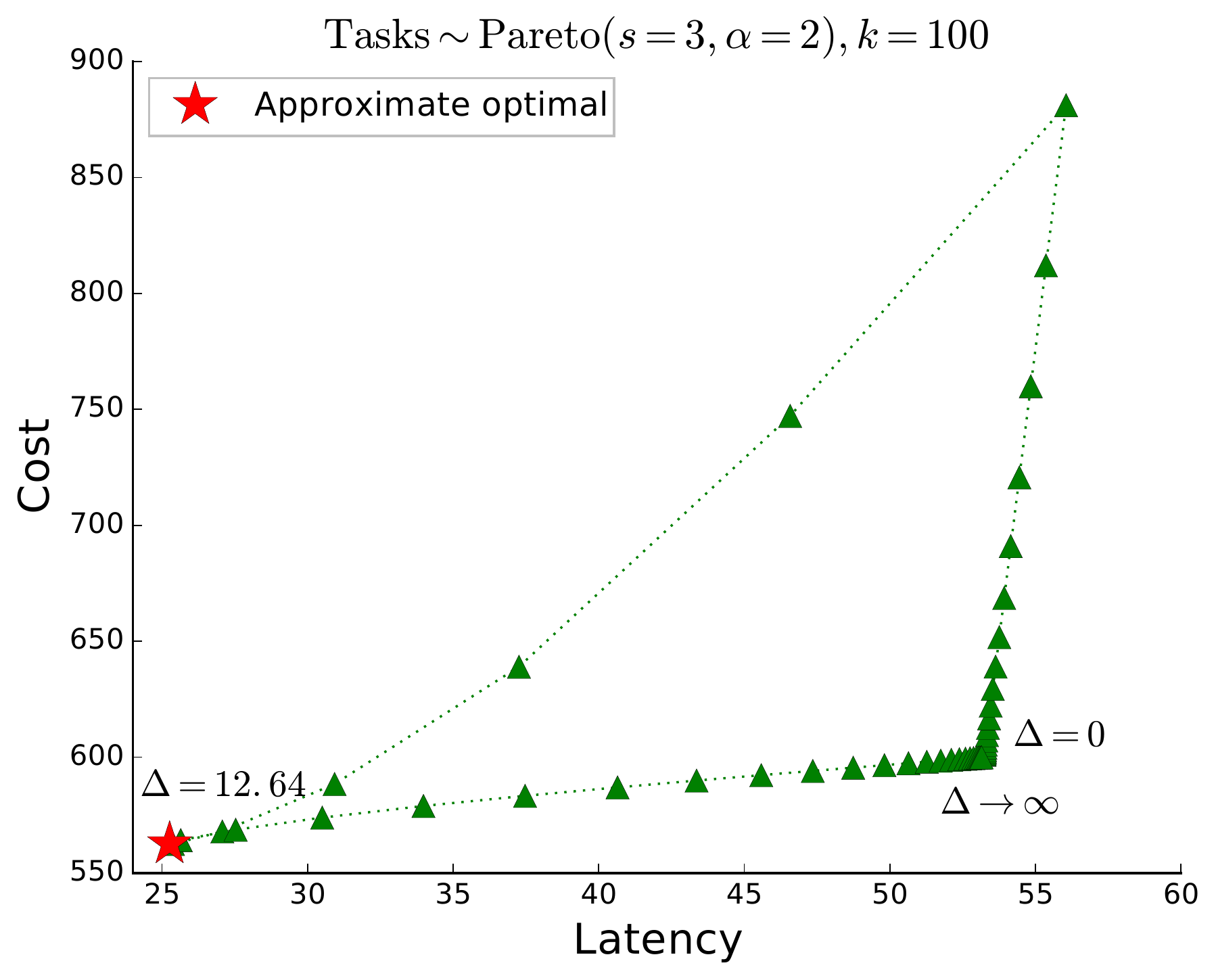}
  \caption{Cost vs. latency of executing a job of $100$ tasks by relaunching the remaining tasks after waiting some time $\Delta$. Relaunch time $\Delta$ is varied from $0$ to $\infty$ along the curve.}
  \label{fig:fig_k_wrelaunch_ET__EC}
\end{figure}

\begin{lemma}
  Suppose task execution times are distributed as $\mathrm{Pareto}(s, \alpha)$, and let $T_{no rel}$ denote the baseline completion time for executing a job of $k$ tasks without straggler relaunch. A sufficient condition for reducing the cost and latency of job execution by performing straggler relaunch is given by
  \begin{equation}
    \E[T_{no rel}] > 4 s.
  \label{eqn:eq_k_wrelaunch_suffcond}
  \end{equation}
  This gives a looser sufficient condition on the tail index as
  \begin{equation}
    \alpha < \ln(k)/\ln(4).
  \label{eqn:eq_k_wrelaunch_suffcond_a}
  \end{equation}
  Optimal relaunch time to execute the job with minimum cost and latency is approximately given as
  \begin{equation}
  	\Delta^* \approx \sqrt{s \E[T_{no rel}]} = s\sqrt{\frac{k!\Gamma(1-1/\alpha)}{\Gamma(k+1-1/\alpha)}}.
  \label{eqn:eq_k_wrelaunch_approx_optd}
  \end{equation}
  This implies that average fraction of the tasks that are relaunched by the optimal strategy is approximately given as
  \begin{equation}
    p^* \approx \left(s/\E[T_{no rel}]\right)^{\alpha/2} \approx \frac{\Gamma(1-1/\alpha)^{-\alpha/2}}{\sqrt{k+1}}.
  \label{eqn:eq_k_wrelaunch_optimal_fraction}
  \end{equation}
  Sufficient conditions and the approximations given above are asymptotic and become exact in the limit $k \to \infty$.
\label{lm_k_wrelaunch_ET__opt_d_suff_a}
\end{lemma}

An approximate expression for the relaunch time $\Delta^*$ that minimizes the cost and latency of job execution is given in Lemma~\ref{lm_k_wrelaunch_ET__opt_d_suff_a}. The given approximation converges to the true optimal as $k$ gets larger, e.g., approximate $\Delta^*$ is very close to the true optimal for the case shown in Fig.~\ref{fig:fig_k_wrelaunch_ET__EC} with $k=100$.
Optimal relaunch time is an increasing function of the number of tasks $k$ and the task sizes $s$, which intuitively makes sense. Also it is a decreasing function of $\alpha$, meaning that it is better to relaunch earlier when the tail of task execution times is lighter, while for heavier tail, delaying relaunch further helps to identify the stragglers better.
This is because relaunching tasks is a choice of canceling the work that is already completed in order to get possibly lucky and execute the replacement copies much faster. When the task execution times are heavier in tail, the residual lifetime of the straggler tasks is expected to be much larger, and the gain of relaunching stragglers can compensate for the work loss. However with lighter tailed task execution times, it is better to try our chance with relaunching earlier and keep the work loss limited.

As discussed above $\Delta^*$ gets smaller as $\alpha$ increases, but this does not imply relaunching a larger fraction of the job's tasks.
When relaunching is performed after waiting $\Delta^*$, fraction $p^*$ of the tasks that are relaunched monotonically decreases\footnote{$p^*$ is a monotonically decreasing function of $\alpha$. As the tail of task execution times becomes very heavy; $\lim_{\alpha \rightarrow 1} \Gamma(1-1/\alpha)^{-\alpha/2} = 1$, and as the tail becomes very light; $\lim_{\alpha \rightarrow \infty} \Gamma(1-1/\alpha)^{-\alpha/2} \approx 0.749$.} with $\alpha$, that is, fewer tasks get relaunched on average by the optimal strategy as the tail gets lighter.
In addition, $p^*$ decreases with $k$, which means for jobs with larger number of tasks, optimal strategy dictates relaunching smaller fraction of the tasks. For instance, suppose $\alpha=2$ and $k=10$, then $p^* \approx 0.17$, which implies $17\%$ of the tasks would need to be relaunched on average with the optimal strategy, while if $k=100$, then only $6\%$ of the tasks would need to be relaunched on average.


We assume relaunching tasks does not introduce any additional delay. Given that, the cost of job execution directly changes with the latency, thus, optimal relaunch time that minimizes latency also minimizes the cost.
Note that cost of relaunching may not be ignored in practice, which is why results presented here on the performance of straggler relaunch can only be taken as optimistic guidelines.

\begin{figure}[b]
  \centering
  \includegraphics[width=0.35\textwidth, keepaspectratio=true]{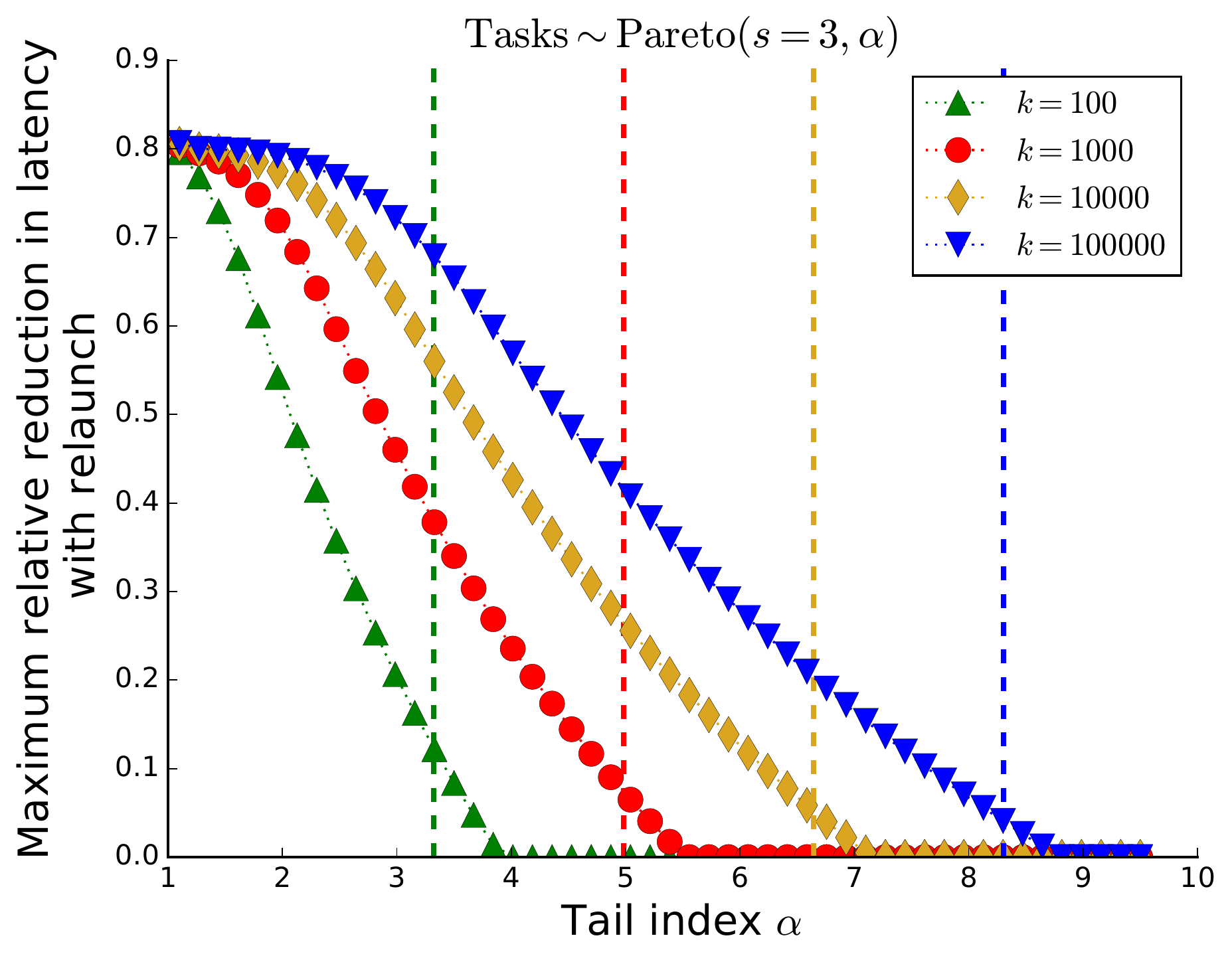}
  \caption{Maximum reduction in the latency of executing a job of $k$ tasks with straggler relaunch (relative to the baseline without relaunch) depends on the tail of the task execution times. Vertical dashed lines indicate the sufficient condition given on $\alpha$ in Lemma~\ref{lm_k_wrelaunch_ET__opt_d_suff_a}.}
  \label{fig:fig_k_wrelaunch__reduc_in_ET_vs_a}
\end{figure}

For relaunching to be effective, work loss due to the cancellation of already running tasks should be compensated by the gain of not having to wait very long for the stragglers. In other words, straggler relaunch is effective only if the tail of task execution times is heavy beyond a level. Otherwise relaunching tasks hurts performance; it incurs additional cost and latency in the job execution. For instance, relaunching always hurts when task execution times have light tail, i.e., when the tail decays at least exponentially fast.

Lemma~\ref{lm_k_wrelaunch_ET__opt_d_suff_a} presents an asymptotic sufficient condition for straggler relaunch to be effective, which has a particularly nice form; if the baseline latency without relaunching is greater than 4 times the minimum task completion time $s$, then relaunching stragglers at the right time will reduce the cost and latency of job execution.
Reformulation of this condition in terms of the tail index $\alpha$ suggests that straggler relaunch is effective as long as $\alpha$ is less than a threshold, which is the same as saying the tail of task execution times should be heavy beyond a level.
Note that this condition on the tail index $\alpha$ does not depend on the minimum task completion time $s$ and is only proportional with the logarithm of the scale $k$ of job execution, which we also validate by numerically computing the exact necessary and sufficient condition on $\alpha$ (see Fig.~\ref{fig:fig_k_wrelaunch__reduc_in_ET_vs_a}).

\section{Redundancy together with Relaunch}
\label{sec:sec_red_togetherwith_relaunch}
In this section, we consider employing redundant tasks and straggler relaunch jointly for straggler mitigation.

\subsection{Zero-delay redundancy with relaunch}
Firstly, we consider launching the redundant tasks ($c$ replicas for each task or $n-k$ coded tasks) together with the $k$ initial tasks of a job, then relaunching each remaining task (initial or redundant) after waiting some time $\Delta$.
Thm.~\ref{thm_k_cn_wrelaunch_ET} gives exact expressions for the latency and Lemma~\ref{lm_k_cn_wrelaunch_ET__opt_d_suff_a} presents an asymptotic sufficient condition for relaunching to be effective in reducing latency, and also presents an approximate value for the optimal relaunch time.
Relaunch time $\Delta$ does not affect the level of added redundancy, hence the optimal relaunch time that minimizes latency also minimizes cost.

\begin{theorem}
  Suppose task execution times are i.i.d. with $\mathrm{Pareto}(s, \alpha)$. Consider launching a job of $k$ tasks together with redundant tasks then relaunching all remaining tasks after waiting some time $\Delta$.
  Let $\E[T_{norel}]$ denote the baseline latency without straggler relaunch as given in Thm~\ref{thm_k_cn_ET__EC}.
  
  When job is launched by adding $c$ replicas for each task,
  \begin{longaligned}[\label{eqn:eq_k_c_wrelaunch_ET}]
    \E[T] =
    \begin{cases}
      \Delta + \E[T_{norel}] & \Delta \leq s, \longalignedtag \\
      \begin{aligned}
        & \Delta(1 - q^k) + \E[T_{norel}]\Bigl(1 + \\
        &~ (s/\Delta - 1)I(1-q; 1-1/\alpha, k) \Bigr)
      \end{aligned} & o.w.
    \end{cases}
  \end{longaligned}
  where $q = \mathbbm{1}(\Delta > s)\left(1 - (s/\Delta)^{(c+1)\alpha}\right)$.
  
  When job is launched by adding $n-k$ coded tasks,
  \begin{longaligned}[\label{eqn:eq_k_n_wrelaunch_ET}]
    \E[T] =
    \begin{cases}
      \Delta + \E[T_{norel}] & \Delta \leq s, \longalignedtag \\
      \begin{aligned}
        & \Delta I(1-q; n-k+1, k) + \E[T_{norel}]\Bigl(1 + \\
        &~ (s/\Delta-1)I(1-q; n-k+1-1/\alpha, k) \Bigr)
      \end{aligned} & o.w.
    \end{cases}
  \end{longaligned}
  where $q = \mathbbm{1}(\Delta > s)\left(1 - (s/\Delta)^{\alpha}\right)$.
\label{thm_k_cn_wrelaunch_ET}
\end{theorem}

\begin{lemma}
  Suppose task execution times are i.i.d. with $\mathrm{Pareto}(s, \alpha)$.
  Let $\E[T_{no rel}]$ denote the latency for a job of $k$ tasks that is launched together with task replicas or coded tasks (without straggler relaunch), which is given in Thm.~\ref{thm_k_cn_ET__EC}.
  A sufficient condition that guarantees reduction in cost and latency by also performing straggler relaunch is given as
  \begin{equation}
    \E[T_{no rel}] > 4s.
  \label{eqn:eq_k_cn_wrelaunch_suffcond}
  \end{equation}
  A looser sufficient condition is, when task replicas are used
  \begin{equation}
    \alpha < \frac{\ln(k)}{(c+1)\ln(4)},
  \label{eqn:eq_k_c_wrelaunch_suffcond_a}
  \end{equation}
  or when coded tasks are used
  \begin{equation}
    \alpha < \frac{\ln\left(n/(n-k+1)\right)}{\ln(4)}.
  \label{eqn:eq_k_n_wrelaunch_suffcond_a}
  \end{equation}
  
  Optimal relaunch time for minimum cost and latency (either when task replicas or coded tasks are used) is approximately
  \begin{equation}
    \Delta^* \approx \sqrt{s \E[T_{no rel}]}.
  \label{eqn:eq_k_cn_wrelaunch_approx_optd}
  \end{equation}
  Sufficient conditions and the approximations given above are asymptotic and becomes exact in the limit $k \to \infty$.
  \label{lm_k_cn_wrelaunch_ET__opt_d_suff_a}
\end{lemma}
\begin{proof}[Proof Sketch]
  Very similar to the proof of Lemma~\ref{lm_k_wrelaunch_ET__opt_d_suff_a}.
\end{proof}

\begin{figure*}[t]
  \centering
  \begin{subfigure}[]{.32\textwidth}
    \centering
    \includegraphics[width=1\textwidth, keepaspectratio=true]{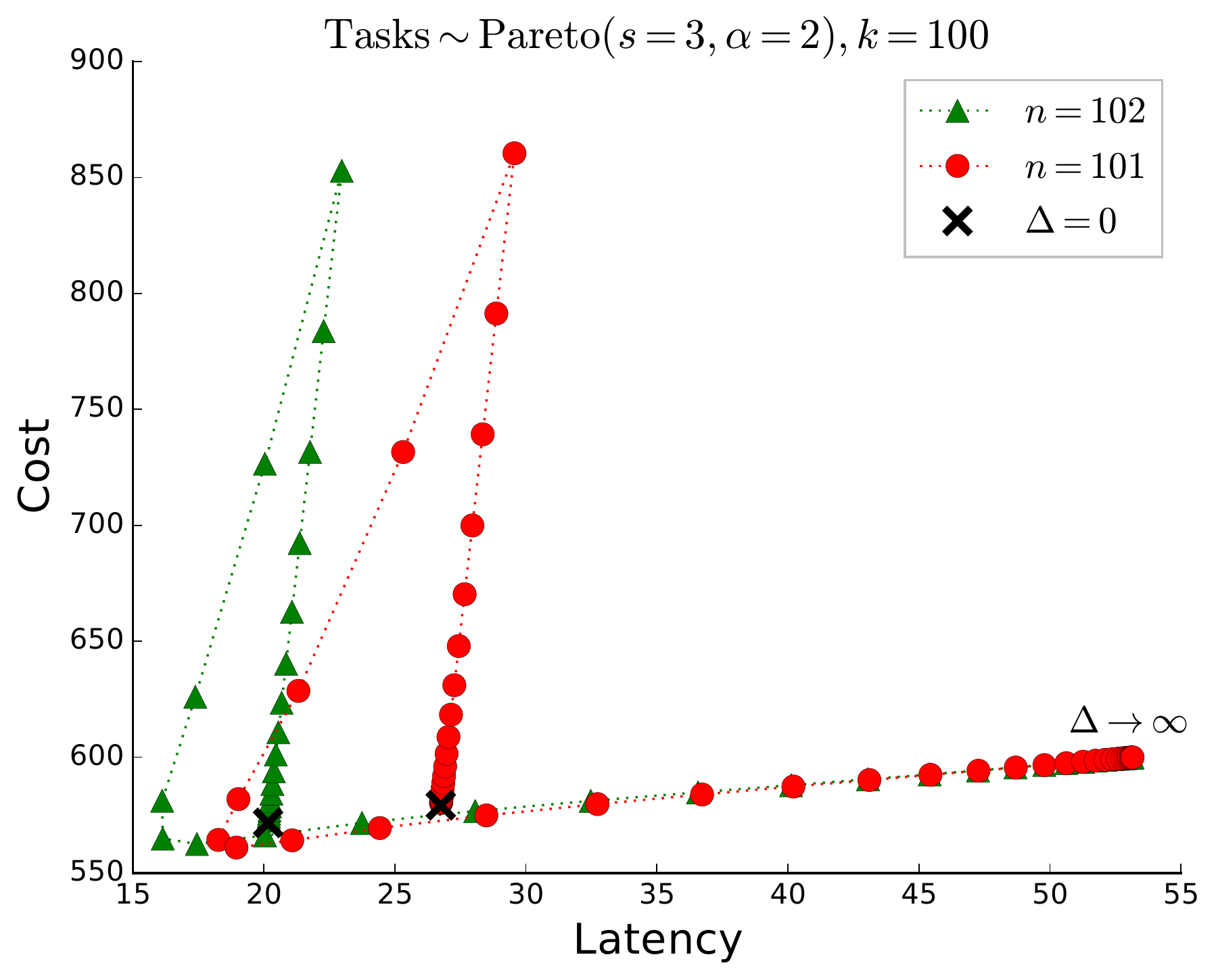}
  \end{subfigure}
  \begin{subfigure}[]{.32\textwidth}
    \centering
    \includegraphics[width=1\textwidth, keepaspectratio=true]{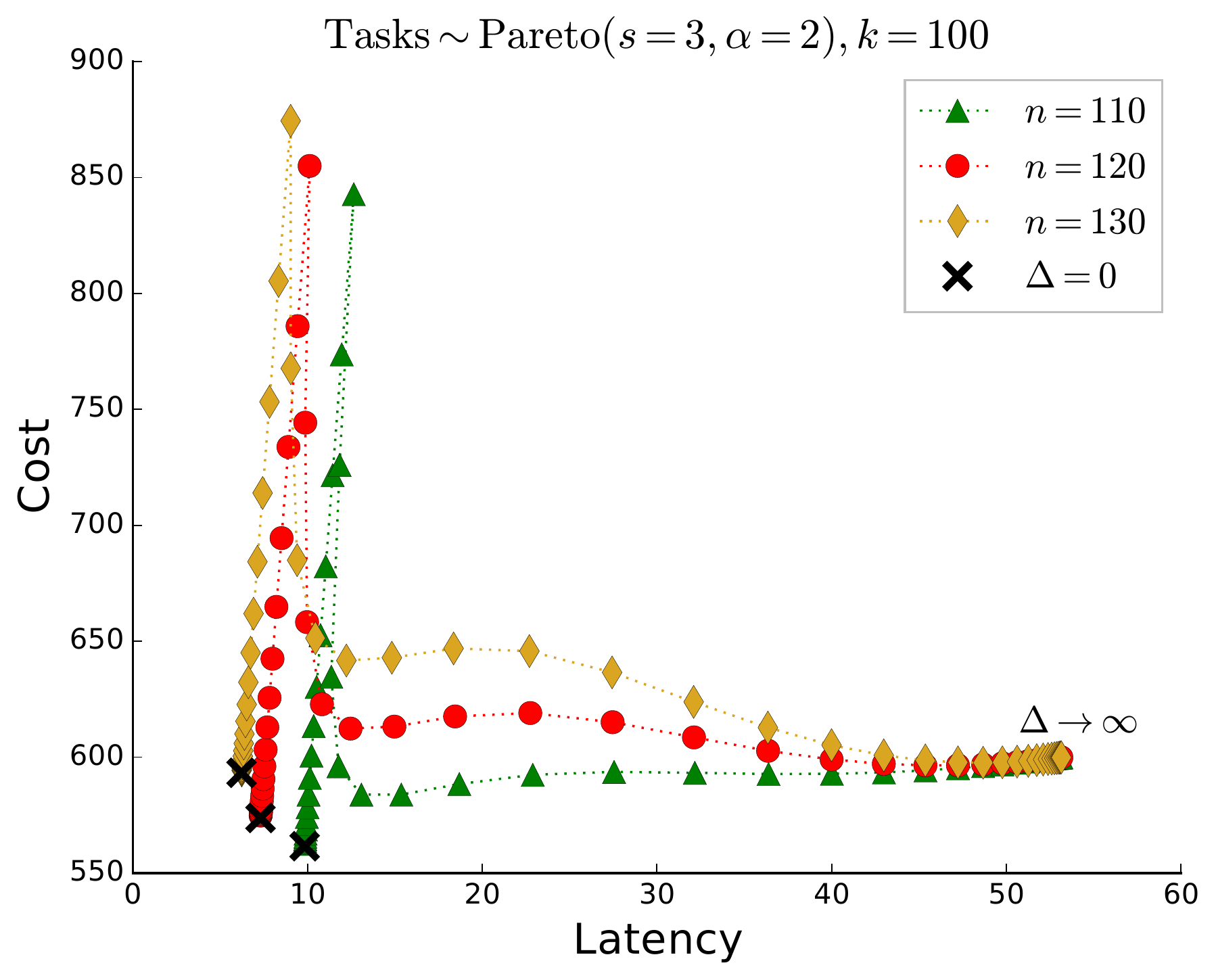}
  \end{subfigure}
  \begin{subfigure}[]{.32\textwidth}
    \centering
    \includegraphics[width=1\textwidth, keepaspectratio=true]{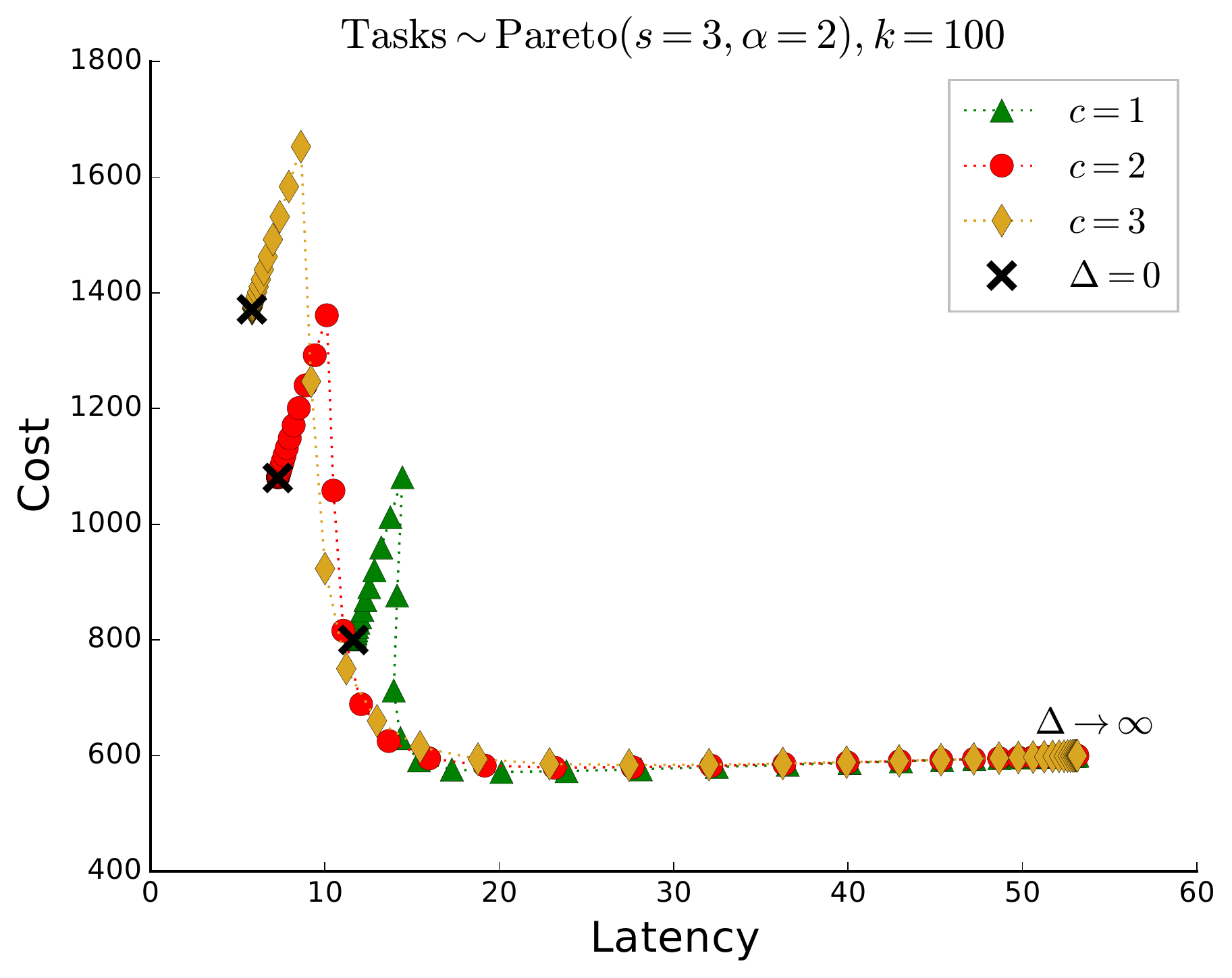}
  \end{subfigure}
  \caption{Cost vs.\ latency curves for executing a job of $100$ tasks by adding redundancy and performing straggler relaunch after time $\Delta$. Each curve is plotted by interpolating between the incremental steps of time $\Delta$.}
  \label{fig:fig_k_nc_wrelaunch_EC_vs_ET}
\end{figure*}

Launching a job with redundant tasks mitigates the effect of stragglers, so does relaunching the stragglers after waiting some time.
Therefore, the relative latency and cost reduction harvested from straggler relaunch decreases when it is used jointly with redundancy. 
Straggler relaunch is effective only when the effect of stragglers is significant, i.e., when the tail of task execution times is heavy beyond a level (cf.\ Lemma~\ref{lm_k_wrelaunch_ET__opt_d_suff_a}). Adding redundant tasks into the job execution already ``cuts'' some of the tail, hence the initial tail heaviness that is required for relaunching to be effective increases with the level of added redundancy.
Sufficient conditions \eqref{eqn:eq_k_c_wrelaunch_suffcond_a} and \eqref{eqn:eq_k_n_wrelaunch_suffcond_a} given on the tail heaviness are asymptotic representations of this observation. The upper threshold given on the tail index as the sufficient condition decays (i.e., required tail heaviness increases) with the level of added redundancy faster when task replicas are used (decays as $1/(c+1)$) compared to using coded tasks (decays as $\ln\left(n/(n-k+1)\right)$).

\subsection{Delayed redundancy with relaunch}
Secondly, we consider adding redundant tasks and performing straggler relaunch jointly after waiting some time $\Delta$.
Cost and latency of job execution in this case are are presented in Thm.~\ref{thm_k_cnd_wrelaunch_ET__EC}.
Using these expressions, Fig.~\ref{fig:fig_k_nc_wrelaunch_EC_vs_ET} plots the cost vs.\ latency tradeoff by varying $\Delta$ from $0$ to $\infty$ for different levels of added redundancy.

When the number of added coded tasks is low, there exists an optimal time $\Delta$ that minimizes the cost and latency of job execution (Left, Fig.~\ref{fig:fig_k_nc_wrelaunch_EC_vs_ET}). This is the same observation that we previously made for the case of performing straggler relaunch without adding any redundancy (cf.\ Fig.~\ref{fig:fig_k_wrelaunch_ET__EC}).
As the number of added coded tasks increases, redundancy becomes a greater effect on the cost and latency than straggler relaunch, hence waiting for some time before adding redundant tasks becomes ineffective to reduce the cost (Middle, Fig.~\ref{fig:fig_k_nc_wrelaunch_EC_vs_ET}). This is the same observation that we made previously for the case of employing delayed redundancy without straggler relaunch (cf.\ Fig.~\ref{fig:fig_delayed_red}).
When task replicas are used rather than coded tasks, regardless of the number of added replicas, delaying $\Delta$ is ineffective to reduce the cost (Right, Fig.~\ref{fig:fig_k_nc_wrelaunch_EC_vs_ET}). This is because replicating each remaining task after some time $\Delta$ even by one is enough to dominate the effect of straggler relaunch on the cost vs.\ latency tradeoff.

\begin{theorem}
  Suppose task execution times are i.i.d. with $\mathrm{Pareto}(s, \alpha)$.
  Let $\E[T_{no red}]$ denote the latency of executing a job of $k$ tasks by relaunching each remaining task after some time $\Delta$ (without adding redundant task) as given in Thm.~\ref{thm_k_wrelaunch_T_C}.
  
  Consider relaunching and adding $c$ replicas for each remaining task after some time $\Delta$.
  Then, latency is given as
  \begin{longaligned}[\label{eqn:eq_k_cd_wrelaunch_ET}]
    \E[T] \approx
    \begin{cases}
      \Delta + s k! \frac{\Gamma(1-1/\tilde{\alpha})}{\Gamma(k+1-1/\tilde{\alpha})} & \Delta \leq s, \\
      \begin{aligned}
        & \E[T_{no red}] + f(\tilde{\alpha}) - f(\alpha).
      \end{aligned} & o.w. \longalignedtag
    \end{cases}
  \end{longaligned}
  for $f(\alpha) = s\;\frac{\Gamma(1-1/\alpha)}{\Gamma(-1/\alpha)}B(k-kq+1, -1/\alpha)$.
  
  Cost with ($C^c$) or without ($C$) task cancellation is given as
  \begin{equation*}
  \begin{split}
    \E[C^c] &=
    \begin{cases}
      k\Delta + ks(c+1)\frac{\tilde{\alpha}}{\tilde{\alpha}-1} & \Delta \leq s, \\
      \begin{split}
        & \frac{k\alpha}{(\alpha-1)}(s - \Delta(1-q)) \\
        &+ k(1-q)\Delta + ks(c+1)(1-q)\frac{\tilde{\alpha}}{\tilde{\alpha}-1}
      \end{split}
      & o.w.
    \end{cases} \\
    \E[C] &=
    \begin{cases}
      k\Delta + ks(c+1)\frac{\alpha}{\alpha-1} & \Delta \leq s, \\
      \begin{split}
        & \frac{k\alpha}{(\alpha-1)}(s - \Delta(1-q)) \\
        &+ k(1-q)\Delta + ks(c+1)(1-q)\frac{\alpha}{\alpha-1}
      \end{split}
      & o.w.
    \end{cases}
  \end{split}
  \label{eqn:eq_k_cd_wrelaunch_EC}
  \end{equation*}
  where $\tilde{\alpha} = (c+1)\alpha$ and $q = \mathbbm{1}(\Delta > s)(1 - (s/\Delta)^{\alpha})$.
  
  Consider adding $n-k$ coded tasks instead of task replicas.
  Then, latency is given as
  \begin{equation}
  \begin{split}
    \E[T] \approx
    \begin{cases}
      \Delta + s \frac{n!}{(n-k)!}\frac{\Gamma(n-k+1-1/\alpha)}{\Gamma(n+1-1/\alpha)} & \Delta \leq s, \\
      \begin{split}
        & \Delta(1-q^k) + s\Bigl(\frac{B(n-kq+1,-1/\alpha)}{B(n-k+1,-1/\alpha)} \\
        &+ kB(q;k,1-1/\alpha) - q^k\Bigr)
      \end{split} & o.w.
    \end{cases}
  \end{split}
  \label{eqn:eq_k_nd_wrelaunch_ET}
  \end{equation}
  
  Cost with ($C^c$) or without ($C$) task cancellation is given as
  \begin{longaligned}[\label{eqn:eq_k_nd_wrelaunch_EC}]
    \E[C^c] &=
    \begin{cases}
      k\Delta + s\frac{n}{\alpha-1}\left(\alpha - \frac{\Gamma(n)}{\Gamma(n-k)}\frac{\Gamma(n-k+1-1/\alpha)}{\Gamma(n+1-1/\alpha)}\right) & \Delta \leq s, \\
      \begin{aligned}
        & \frac{\alpha}{\alpha-1}(k(1-q)(s-\Delta) + ns) \\
        &+ k(1-q)\Delta - s(n-k)q^k \\
        &- \frac{s}{\alpha-1}(n-k)\frac{B(n-kq+1, -1/\alpha)}{B(n-k+1, -1/\alpha)}.
      \end{aligned} & o.w.
    \end{cases} \\
    \E[C] &= 
    \begin{cases}
      k\Delta + ns/(1-1/\alpha) & \Delta \leq s, \\
      \begin{aligned}
      & \frac{\alpha}{\alpha-1}\bigl(ks(1-q+q^k) \\
      &\quad + ns(1-q^k)\bigr) - \frac{k\Delta(1-q)}{\alpha-1}
      \end{aligned} & o.w. \longalignedtag
    \end{cases}
  \end{longaligned}
  where $q = \mathbbm{1}(\Delta > s)(1 - (s/\Delta)^{\alpha})$.
\label{thm_k_cnd_wrelaunch_ET__EC}
\end{theorem}

\section{Conclusions}
\label{sec:sec_conclusions}
This paper presented a theoretical performance evaluation of the two most widely deployed straggler mitigation techniques for distributed job execution: i) adding redundant tasks (together with the original tasks or after waiting some time) into the job and waiting only for a sufficient subset of all launched tasks for job completion, ii) waiting for some time after launching the job and relaunching its remaining tasks.
We derived the cost and latency expressions for executing the job by applying either one of these techniques or both jointly.
Using the derived expressions, we found the following guidelines for the application of these techniques:
i) Waiting for some time before launching redundant tasks is not effective to reduce the cost of redundancy.
ii) Launching a job with redundant tasks can reduce not only its latency but also its cost.
iii) Launching a job with MDS coded tasks achieves less cost (hence incurs less additional load on the system) and latency than using task replicas.
iv) Relaunching remaining tasks after waiting some time is effective only if the tail of task execution times is heavier beyond a level, and employing redundant tasks together with straggler relaunch increases this tail heaviness requirement.

In our system model, we abstract away the job dispatching and resource sharing dynamics by modeling execution times of tasks within a job as i.i.d. random variables.
This rather lumped model allows deriving insightful expressions that allows evaluating and comparing widely deployed straggler mitigation techniques. However, application of these techniques modifies the system dynamics and it is necessary to augment the model to reflect the impact of this modification on task execution times.
This is an ongoing challenge for us and Sec.~\ref{sec:sec_when_red_changes_tail} presented a simulation driven attempt in this direction.

\bibliographystyle{IEEEtran}
\bibliography{references}

\begin{thebibliography}{10}
\providecommand{\url}[1]{#1}
\csname url@samestyle\endcsname
\providecommand{\newblock}{\relax}
\providecommand{\bibinfo}[2]{#2}
\providecommand{\BIBentrySTDinterwordspacing}{\spaceskip=0pt\relax}
\providecommand{\BIBentryALTinterwordstretchfactor}{4}
\providecommand{\BIBentryALTinterwordspacing}{\spaceskip=\fontdimen2\font plus
\BIBentryALTinterwordstretchfactor\fontdimen3\font minus
  \fontdimen4\font\relax}
\providecommand{\BIBforeignlanguage}[2]{{%
\expandafter\ifx\csname l@#1\endcsname\relax
\typeout{** WARNING: IEEEtran.bst: No hyphenation pattern has been}%
\typeout{** loaded for the language `#1'. Using the pattern for}%
\typeout{** the default language instead.}%
\else
\language=\csname l@#1\endcsname
\fi
#2}}
\providecommand{\BIBdecl}{\relax}
\BIBdecl

\bibitem{MAMA:AktasPS17}
M.~F. Aktas, P.~Peng, and E.~Soljanin, ``Effective straggler mitigation: Which
  clones should attack and when?'' \emph{ACM SIGMETRICS Performance Evaluation
  Review}, vol.~45, no.~2, pp. 12--14, 2017.

\bibitem{IFIP:AktasPS18}
------, ``Straggler mitigation by delayed relaunch of tasks,'' \emph{ACM
  SIGMETRICS Performance Evaluation Review}, vol.~45, no.~3, pp. 224--231,
  2018.

\bibitem{Dryad:IsardBY07}
M.~Isard, M.~Budiu, Y.~Yu, A.~Birrell, and D.~Fetterly, ``Dryad: distributed
  data-parallel programs from sequential building blocks,'' in \emph{ACM SIGOPS
  operating systems review}.\hskip 1em plus 0.5em minus 0.4em\relax ACM, 2007.

\bibitem{MapReduce:DeanG08}
J.~Dean and S.~Ghemawat, ``Mapreduce: simplified data processing on large
  clusters,'' \emph{Communications of the ACM}, 2008.

\bibitem{Mantri:AnanthanarayananKG10}
G.~Ananthanarayanan, S.~Kandula, A.~G. Greenberg, I.~Stoica, Y.~Lu, B.~Saha,
  and E.~Harris, ``Reining in the outliers in map-reduce clusters using
  mantri.'' in \emph{Osdi}, vol.~10, no.~1, 2010, p.~24.

\bibitem{ResilientDistributedDatasets:ZahariaCD12}
M.~Zaharia, M.~Chowdhury, T.~Das, A.~Dave, J.~Ma, M.~McCauley, M.~J. Franklin,
  S.~Shenker, and I.~Stoica, ``Resilient distributed datasets: A fault-tolerant
  abstraction for in-memory cluster computing,'' in \emph{Proceedings of the
  9th USENIX conference on Networked Systems Design and Implementation}.\hskip
  1em plus 0.5em minus 0.4em\relax USENIX Association, 2012, pp. 2--2.

\bibitem{TailAtScale:DeanB13}
J.~Dean and L.~A. Barroso, ``The tail at scale,'' \emph{Communications of the
  ACM}, vol.~56, no.~2, pp. 74--80, 2013.

\bibitem{StragglerRootCauseAnalysisInDatacenters:OuyangGY16}
X.~Ouyang, P.~Garraghan, R.~Yang, P.~Townend, and J.~Xu, ``Reducing late-timing
  failure at scale: Straggler root-cause analysis in cloud datacenters,'' in
  \emph{Fast Abstracts in the 46th IEEE/IFIP International Conference on
  Dependable Systems and Networks}.\hskip 1em plus 0.5em minus 0.4em\relax DSN,
  2016.

\bibitem{RootCauseAnalysisOfStragglersInBigDataSystem:ZhouLY18}
H.~Zhou, Y.~Li, H.~Yang, J.~Jia, and W.~Li, ``Bigroots: An effective approach
  for root-cause analysis of stragglers in big data system,'' \emph{arXiv
  preprint arXiv:1801.03314}, 2018.

\bibitem{AchievingRapidResponseTimesInLargeOnlineServices:Dean12}
\BIBentryALTinterwordspacing
J.~Dean, ``Achieving rapid response times in large online services,'' 2012.
  [Online]. Available:
  \url{https://research.google.com/people/jeff/latency.html}
\BIBentrySTDinterwordspacing

\bibitem{ProactiveStragglerAvoidance:YadwadkarC12}
N.~J. Yadwadkar and W.~Choi, ``Proactive straggler avoidance using machine
  learning,'' \emph{White paper, University of Berkeley}, 2012.

\bibitem{ImprovingMapReducePerformance:ZahariaKJ08}
M.~Zaharia, A.~Konwinski, A.~D. Joseph, R.~H. Katz, and I.~Stoica, ``Improving
  mapreduce performance in heterogeneous environments.'' in \emph{Osdi},
  vol.~8, no.~4, 2008, p.~7.

\bibitem{Dremel:MelnikGL10}
S.~Melnik, A.~Gubarev, J.~J. Long, G.~Romer, S.~Shivakumar, M.~Tolton, and
  T.~Vassilakis, ``Dremel: interactive analysis of web-scale datasets,''
  \emph{Proceedings of the VLDB Endowment}, vol.~3, pp. 330--339, 2010.

\bibitem{AttackOfClones:AnanthanarayananGS13}
G.~Ananthanarayanan, A.~Ghodsi, S.~Shenker, and I.~Stoica, ``Effective
  straggler mitigation: Attack of the clones.'' 2013.

\bibitem{LowLatencyviaRed:VulimiriGM13}
A.~Vulimiri, P.~B. Godfrey, R.~Mittal, J.~Sherry, S.~Ratnasamy, and S.~Shenker,
  ``Low latency via redundancy,'' in \emph{Proceedings of the ninth ACM
  conference on Emerging networking experiments and technologies}.\hskip 1em
  plus 0.5em minus 0.4em\relax ACM, 2013, pp. 283--294.

\bibitem{DecentralizedSpeculationAwareClusterScheduling:RenAW15}
X.~Ren, G.~Ananthanarayanan, A.~Wierman, and M.~Yu, ``Hopper: Decentralized
  speculation-aware cluster scheduling at scale,'' in \emph{ACM SIGCOMM
  Computer Communication Review}.\hskip 1em plus 0.5em minus 0.4em\relax ACM,
  2015.

\bibitem{DecouplingSlowdownJobsize:GardnerHS17}
K.~Gardner, M.~Harchol-Balter, A.~Scheller-Wolf, and B.~Van~Houdt, ``A better
  model for job redundancy: Decoupling server slowdown and job size,''
  \emph{IEEE/ACM Transactions on Networking}, 2017.

\bibitem{Kubernetes:BurnsGO16}
B.~Burns, B.~Grant, D.~Oppenheimer, E.~Brewer, and J.~Wilkes, ``Borg, omega,
  and kubernetes,'' 2016.

\bibitem{Codes&Qs:JoshiLS12}
G.~Joshi, Y.~Liu, and E.~Soljanin, ``Coding for fast content download,'' in
  \emph{Communication, Control, and Computing (Allerton), 2012 50th Annual
  Allerton Conference on}.\hskip 1em plus 0.5em minus 0.4em\relax IEEE, 2012,
  pp. 326--333.

\bibitem{Codes&Qs:HuangPZ12}
L.~Huang, S.~Pawar, H.~Zhang, and K.~Ramchandran, ``Codes can reduce queueing
  delay in data centers,'' in \emph{Proceed.\ 2012 IEEE International Symposium
  on Information Theory (ISIT'12)}.

\bibitem{CodesQs:KadheSS15_Allerton}
S.~Kadhe, E.~Soljanin, and A.~Sprintson, ``When do the availability codes make
  the stored data more available?'' in \emph{Communication, Control, and
  Computing (Allerton), 2015 53rd Annual Allerton Conference on}.\hskip 1em
  plus 0.5em minus 0.4em\relax IEEE, 2015, pp. 956--963.

\bibitem{ShortDot:DuttaCG16}
S.~Dutta, V.~Cadambe, and P.~Grover, ``Short-dot: Computing large linear
  transforms distributedly using coded short dot products,'' in \emph{Advances
  In Neural Information Processing Systems}, 2016.

\bibitem{MachineLearningWithCodes:LeeLP17}
K.~Lee, M.~Lam, R.~Pedarsani, D.~Papailiopoulos, and K.~Ramchandran, ``Speeding
  up distributed machine learning using codes,'' \emph{IEEE Transactions on
  Information Theory}, 2017.

\bibitem{GradientCoding:TandonLD16}
R.~Tandon, Q.~Lei, A.~G. Dimakis, and N.~Karampatziakis, ``Gradient coding,''
  \emph{arXiv preprint arXiv:1612.03301}, 2016.

\bibitem{CodedGradientDescent:LiKA17}
S.~Li, S.~M.~M. Kalan, A.~S. Avestimehr, and M.~Soltanolkotabi, ``Near-optimal
  straggler mitigation for distributed gradient methods,'' \emph{arXiv preprint
  arXiv:1710.09990}, 2017.

\bibitem{StragglerMitigationWithDataEncoding:KarakusSD17}
C.~Karakus, Y.~Sun, S.~Diggavi, and W.~Yin, ``Straggler mitigation in
  distributed optimization through data encoding,'' in \emph{Advances in Neural
  Information Processing Systems}, 2017, pp. 5438--5446.

\bibitem{CodedMatrixMultiplication:YuMA18}
Q.~Yu, M.~A. Maddah-Ali, and A.~S. Avestimehr, ``Straggler mitigation in
  distributed matrix multiplication: Fundamental limits and optimal coding,''
  \emph{arXiv preprint arXiv:1801.07487}, 2018.

\bibitem{CodedGradientDescent:HalbawiAS18}
W.~Halbawi, N.~Azizan, F.~Salehi, and B.~Hassibi, ``Improving distributed
  gradient descent using reed-solomon codes,'' in \emph{IEEE International
  Symposium on Information Theory}, 2018.

\bibitem{RepedComputing:WangJW15}
D.~Wang, G.~Joshi, and G.~Wornell, ``Using straggler replication to reduce
  latency in large-scale parallel computing,'' \emph{ACM SIGMETRICS Performance
  Evaluation Review}, vol.~43, no.~3, pp. 7--11, 2015.

\bibitem{ImprovingMapReduceInHeteroEnvironments:Zaharia08}
M.~Zaharia, A.~Konwinski, A.~D. Joseph, R.~H. Katz, and I.~Stoica, ``Improving
  mapreduce performance in heterogeneous environments.'' in \emph{Osdi},
  vol.~8, no.~4, 2008, p.~7.

\bibitem{GoogleTraceAnalysis:ReissTG12}
C.~Reiss, A.~Tumanov, G.~R. Ganger, R.~H. Katz, and M.~A. Kozuch, ``Towards
  understanding heterogeneous clouds at scale: Google trace analysis,''
  \emph{Intel Science and Technology Center for Cloud Computing, Tech. Rep},
  p.~84, 2012.

\bibitem{GRASS:AnanthanarayananHR14}
G.~Ananthanarayanan, M.~C.-C. Hung, X.~Ren, I.~Stoica, A.~Wierman, and M.~Yu,
  ``$\{$GRASS$\}$: Trimming stragglers in approximation analytics,'' in
  \emph{11th $\{$USENIX$\}$ Symposium on Networked Systems Design and
  Implementation}, 2014.

\bibitem{StragglerRep:WangJW15}
D.~Wang, G.~Joshi, and G.~Wornell, ``Using straggler replication to reduce
  latency in large-scale parallel computing,'' \emph{ACM SIGMETRICS Performance
  Evaluation Review}, vol.~43, no.~3, pp. 7--11, 2015.

\bibitem{FundamentalsOfHeavyTails:NairWZ13}
J.~Nair, A.~Wierman, and B.~Zwart, ``The fundamentals of heavy-tails:
  properties, emergence, and identification,'' in \emph{ACM SIGMETRICS
  Performance Evaluation Review}.\hskip 1em plus 0.5em minus 0.4em\relax ACM,
  2013.

\bibitem{PerfEvalWithHeavyTails:Crovella01}
M.~E. Crovella, ``Performance evaluation with heavy tailed distributions,'' in
  \emph{Workshop on Job Scheduling Strategies for Parallel Processing}, 2001.

\bibitem{GoogleClusterTrace:ReissWH11}
C.~Reiss, J.~Wilkes, and J.~L. Hellerstein, ``Google cluster-usage traces:
  format+ schema,'' \emph{Google Inc., White Paper}, 2011.

\bibitem{ErasureCodingBasedRoutingForOpportunisticNetworks:WangJM05}
Y.~Wang, S.~Jain, M.~Martonosi, and K.~Fall, ``Erasure-coding based routing for
  opportunistic networks,'' in \emph{Proceedings of the 2005 ACM SIGCOMM
  workshop on Delay-tolerant networking}, 2005.

\bibitem{BoostingThroughput:Joshi17}
G.~Joshi, ``Boosting the throughput of a multi-server system via adaptive task
  replication,'' 2017.

\bibitem{ExascaleDOE:BrownMB10}
D.~L. Brown, P.~Messina, P.~Beckman, D.~Keyes, J.~Vetter, M.~Anitescu, J.~Bell,
  R.~Brightwell, B.~Chamberlain, D.~Estep \emph{et~al.}, ``Cross cutting
  technologies for computing at the exascale,'' \emph{US DoE Office of Advanced
  Scientific Computing Research and the National Nuclear Security
  Administration, Tech. Rep}, 2010.

\bibitem{TimeSharingInHPC:HofmeyrIC16}
S.~Hofmeyr, C.~Iancu, J.~Colmenares, E.~Roman, and B.~Austin, ``Time-sharing
  redux for large-scale hpc systems,'' in \emph{High Performance Computing and
  Communications; IEEE 14th International Conference on Smart City; IEEE 2nd
  International Conference on Data Science and Systems (HPCC/SmartCity/DSS),
  2016 IEEE 18th International Conference on}.\hskip 1em plus 0.5em minus
  0.4em\relax IEEE, 2016, pp. 301--308.

\bibitem{NIST:DLMF}
\BIBentryALTinterwordspacing
``{\it NIST Digital Library of Mathematical Functions},''
  http://dlmf.nist.gov/, f.~W.~J. Olver, A.~B. {Olde Daalhuis}, D.~W. Lozier,
  B.~I. Schneider, R.~F. Boisvert, C.~W. Clark, B.~R. Miller and B.~V.
  Saunders, eds. [Online]. Available: \url{http://dlmf.nist.gov/}
\BIBentrySTDinterwordspacing

\bibitem{OrderStatisticsForINID:BapatB89}
R.~Bapat and M.~Beg, ``Order statistics for nonidentically distributed
  variables and permanents,'' \emph{Sankhy{\=a}: The Indian Journal of
  Statistics, Series A}, pp. 79--93, 1989.

\bibitem{OrderStat:Arnold08}
B.~C. Arnold, N.~Balakrishnan, and H.~N. Nagaraja, \emph{A first course in
  order statistics}.\hskip 1em plus 0.5em minus 0.4em\relax SIAM, 2008.

\bibitem{Pareto:Arnold15}
B.~C. Arnold, \emph{Pareto distribution}.\hskip 1em plus 0.5em minus
  0.4em\relax Wiley Online Library, 2015.

\bibitem{UnderstandingLinuxKernel:BovetC05}
D.~P. Bovet and M.~Cesati, \emph{Understanding the Linux Kernel: from I/O ports
  to process management}.\hskip 1em plus 0.5em minus 0.4em\relax "O'Reilly
  Media, Inc.", 2005.

\bibitem{LimitedPS:Nair10}
J.~Nair, A.~Wierman, and B.~Zwart, ``Tail-robust scheduling via limited
  processor sharing,'' \emph{Performance Evaluation}, 2010.

\bibitem{HeavyTailedJobs:Leland86}
W.~Leland and T.~J. Ott, \emph{Load-balancing heuristics and process
  behavior}.\hskip 1em plus 0.5em minus 0.4em\relax ACM, 1986, vol.~14, no.~1.

\bibitem{HeavyTailedJobs:Harchol97}
M.~Harchol-Balter and A.~B. Downey, ``Exploiting process lifetime distributions
  for dynamic load balancing,'' \emph{ACM Transactions on Computer Systems
  (TOCS)}, vol.~15, no.~3, pp. 253--285, 1997.

\bibitem{GoogleClusterDataAnalysis:ChenGG10}
Y.~Chen, A.~S. Ganapathi, R.~Griffith, and R.~H. Katz, ``Analysis and lessons
  from a publicly available google cluster trace,'' \emph{EECS Department,
  University of California, Berkeley, Tech. Rep. UCB/EECS-2010-95}, vol.~94,
  2010.

\bibitem{UnfairnessSRPT:BansalH01}
N.~Bansal and M.~Harchol-Balter, \emph{Analysis of SRPT scheduling:
  Investigating unfairness}.\hskip 1em plus 0.5em minus 0.4em\relax ACM, 2001,
  vol.~29, no.~1.

\bibitem{ParetoEstimation:Rytgaard90}
M.~Rytgaard, ``Estimation in the pareto distribution,'' \emph{ASTIN Bulletin:
  The Journal of the IAA}, vol.~20, no.~2, pp. 201--216, 1990.

\bibitem{TParetoEstimation:AbanMP06}
I.~B. Aban, M.~M. Meerschaert, and A.~K. Panorska, ``Parameter estimation for
  the truncated pareto distribution,'' \emph{Journal of the American
  Statistical Association}, vol. 101, no. 473, pp. 270--277, 2006.

\bibitem{QueueingWithHeavyTails:Zwart01}
A.~P. Zwart, \emph{Queueing systems with heavy tails}.\hskip 1em plus 0.5em
  minus 0.4em\relax Technische Universiteit Eindhoven, 2001.

\bibitem{MG1AsymTail:Sakurai04}
T.~Sakurai, ``Approximating m/g/1 waiting time tail probabilities,'' 2004.

\bibitem{MG1HeavyTailAsymp:OlveraBG11}
M.~Olvera-Cravioto, J.~Blanchet, P.~Glynn \emph{et~al.}, ``On the transition
  from heavy traffic to heavy tails for the m/g/1 queue: the regularly varying
  case,'' \emph{The Annals of Applied Probability}, 2011.

\bibitem{AsymptoticApproxOfQuotientOfGamma:Tricomi51}
F.~Tricomi and A.~Erd{\'e}lyi, ``The asymptotic expansion of a ratio of gamma
  functions,'' \emph{Pacific Journal of Mathematics}, pp. 133--142, 1951.

\bibitem{SumOfGammaRatios:Garrappa07}
R.~Garrappa, ``Some formulas for sums of binomial coefficients and gamma
  functions,'' in \emph{International Mathematical Forum}, 2007.

\bibitem{GammaIneqs:Gautschi59}
W.~Gautschi, ``Some elementary inequalities relating to the gamma and
  incomplete gamma function,'' \emph{Studies in Applied Mathematics}, 1959.

\end{thebibliography}

\newpage
\section{Appendix}
\label{sec:sec_appendix}
We use the following Lemma frequently in the derivations.
\begin{lemma}\cite[eq.~4]{SumOfGammaRatios:Garrappa07}
  For $\alpha \in [1, \infty]$, we have
  \begin{equation}
    \sum_{m=0}^n \frac{\Gamma(m-\beta)}{\Gamma(m)} = \frac{\Gamma(n+1-\beta)}{(1-\beta)\Gamma(n)}
  \label{eqn:eq_sum_ratio_gamma}
  \end{equation}
  \label{lm_sum_gamma_ratios}
\end{lemma}

We next present two approximations that proved to be very useful in approximating cost and latency expressions. They render tedious formulas into relatively simpler and intuitively appealing forms. We used these simpler forms of the expressions to answer, at least asymptotically, the questions of interest, such as optimal straggler relaunch time or the shape of cost and latency in terms of important system parameters.
\begin{lemma}
  Given $R \sim \mathrm{Binomial}(k, q)$,
  \begin{equation}
  \begin{split}
    \E[I(z; x-R, y)] &\approx I(z; x-kq, y), \\
    \E[H_{n-R}] &\approx H_{n-kq}.
  \end{split}
  \label{eqn:eq_approx_binomial_mean_of_harmonic__ibeta}
  \end{equation}
  for $z \in [0,1]$ and $x, y, n \in \mathbb{R}^+$. 
  \label{lm_approx_bin_mean_of_harmonic__ibeta}
\end{lemma}
\begin{proof}
  Using the integral representation of the regularized incomplete beta function \cite[eq.~8.17]{NIST:DLMF}
  \[ I(z; x-R, y) = \frac{(1-z)^y}{2\pi i} \int_{c-i\infty}^{c+i\infty} \frac{(z/s)^{x-R}}{(s-z)(1-s)^y} ds. \]
  we get
  \begin{equation}
  \E[I(z; x-R, y)] = \frac{(1-z)^y}{2\pi i} \int_{c-i\infty}^{c+i\infty} \frac{\E[(z/s)^{x-R}]}{(s-z)(1-s)^y} ds.
  \label{eq:eq_irrbeta}
  \end{equation}
  Using the probability generating function of $R$,
  \[ \E\left[(z/s)^{-R} \right] = (q s/z+1-q)^k \approx (z/s)^{-kq}. \]
  which is exact in the limits as
  \begin{equation*}
  \begin{split}
    & \lim_{q \rightarrow 0} \;(q s/z+1-q)^k = 1, \\
    & \lim_{q \rightarrow 1} \;(q s/z+1-q)^k = (z/s)^{-k}.
  \end{split}
  \end{equation*}
  Substituting this approximation in \eqref{eq:eq_irrbeta} gives
  \begin{equation*}
  \begin{split}
    \E[I(z; x-R, y)] &\approx \frac{(1-z)^y}{2\pi i} \int_{c-i\infty}^{c+i\infty} \frac{(z/s)^{x-kq}}{(s-z)(1-s)^y} ds \\
    &= I(z;x-kq,y).
  \end{split}
  \end{equation*}
  which gives the first approximation in \eqref{eqn:eq_approx_binomial_mean_of_harmonic__ibeta}.
  
  In order to derive the other approximation given in \eqref{eqn:eq_approx_binomial_mean_of_harmonic__ibeta}, we use the Euler representation of the harmonic number as
  \[ \E[H_{n-R}] = \int_0^1 \frac{1-x^{n-R}}{1-x} dx = \int_0^1 \frac{1-x^n \E[x^{-R}]}{1-x} dx. \]
  One can write $\E[x^{-R}] = (q/x+1-q)^k \approx x^{-kq}$ using the probability generating function of $R$, and observing $\lim_{q \rightarrow 0} q/x+1-q = 1$ and $\lim_{q \rightarrow 1} q/x+1-q = x^{-1}$. Substituting this approximation above gives
  \begin{equation*}
    \E[H_{n-R}] \approx \int_0^1 \frac{1-x^{n-kq}}{1-x} = H_{n-kq}.
  \end{equation*}
\end{proof}

Majority of the claims presented in the paper were stated without a proof. In the sequel, we present these deferred proofs. In each proof, random variable $X$ denotes the execution time of a single task copy.

\subsection{Proof of Theorem~\ref{thm_k_cd_Exp_T_C}}
\label{subsec:subsec_proof_thm_k_cd_Exp_T_C}
This Theorem assumes task execution times are i.i.d. $\mathrm{Exp}(s, \mu)$, and presents the latency and cost of executing a job of $k$ tasks by employing replicated redundancy after waiting some time $\Delta$.
\noindent
\begin{proof}
  \vspace{0.2em}
  \noindent
  \textbf{Distribution of job execution time, \eqref{eqn:eq_k_cd_Exp_tail}.}
  When a task lasts longer than $\Delta$, $c$ fresh replicas for the task are launched.
  Let us denote the completion time of a task with $Y$, and denote the execution time of a single task copy with $X$.
  Then by the law of total probability
  \begin{longaligned}[\label{eqn:eq_k_cd_Pr_T_g_t}]
    \Pr\{&Y > t\} = \Pr\{Y > t \mid X \leq \Delta\}\Pr\{X \leq \Delta\} \\
      &\qquad\qquad + \Pr\{Y > t \mid X > \Delta\}\Pr\{X > \Delta\} \\
    \stackrel{(a)}{=}& \Pr\{t \leq X \leq \Delta\} + \Pr\{X_{c+1:1} > t-\Delta\}\Pr\{X > \Delta\} \\
    =& \mathbbm{1}(t \leq \Delta)(e^{-\mu t} - e^{-\mu\Delta}) + \mathbbm{1}(t > \Delta) e^{-\mu(c+1)(t - \Delta)}e^{-\mu\Delta}. \longalignedtag
  \end{longaligned}
  Equality $(a)$ results from the following event equality
  \begin{equation*}
  \begin{split}
    & \{Y > t \mid X \leq \Delta\} = \{X > t \mid X \leq \Delta\}, \\
    & \{Y > t \mid X > \Delta\} = \{\Delta + X_{c+1:1} > t\}.
  \end{split}
  \end{equation*}
  For job completion, all of its $k$ tasks need to complete. Thus, $T \sim Y_{k:k}$, which yields
  \[ \Pr\{T \leq t\} = \Pr\{Y \leq t\}^k = (1 - \Pr\{Y > t\})^k. \]
  Substituting \eqref{eqn:eq_k_cd_Pr_T_g_t} in which yields \eqref{eqn:eq_k_cd_Exp_tail}.
  
  \vspace{0.2em}
  \noindent
  \textbf{Latency, \eqref{eqn:eq_k_cd_Exp_ET}.}
  By the law of total expectation
  \begin{longaligned}[\label{eqn:eq_k_cd_Exp__ET_condsum1}]
    \E[T] &= \E[T \mid T \leq \Delta]\Pr\{T \leq \Delta\} \\
    &\quad + \E[T \mid T > \Delta]\Pr\{T > \Delta\} \\
    &= \E[X_{k:k} \mid X_{k:k} \leq \Delta]\Pr\{X_{k:k} \leq \Delta\} \\
    &\quad + \E[T \mid X_{k:k} > \Delta]\Pr\{X_{k:k} > \Delta\}. \longalignedtag
  \end{longaligned}
  where we used the equivalence
  $\{T \leq \Delta\} = \{X_{k:k} \leq \Delta\}$ and $\{T > \Delta\} = \{X_{k:k} > \Delta\}$, which follows by observing that replica tasks are launched only after time $\Delta$.
  
  When no redundancy is employed, latency is given as
  \begin{longaligned}
    \E[X_{k:k}] &= \E[X_{k:k} \mid X_{k:k} \leq \Delta]\Pr\{X_{k:k} \leq \Delta\} \\
    &\quad + \E[X_{k:k} \mid X_{k:k} > \Delta]\Pr\{X_{k:k} > \Delta\}.
  \end{longaligned}
  Using \eqref{eqn:eq_k_cd_Exp__ET_condsum1} we get
  \begin{longaligned}[\label{eqn:eq_k_cd_Exp__ET_condsum2}]
    \E[T] - \E[X_{k:k}] &= \E[T \mid X_{k:k} > \Delta]\Pr\{X_{k:k} > \Delta\} \\
    &- \E[X_{k:k} \mid X_{k:k} > \Delta]\Pr\{X_{k:k} > \Delta\}. \longalignedtag
  \end{longaligned}
  
  Let us denote the number of tasks that complete before time $\Delta$ as $R$, hence $R \sim \mathrm{Binomial}(k, q)$ for $q = 1 - e^{-\mu\Delta}$.
  The left hand side of the difference in \eqref{eqn:eq_k_cd_Exp__ET_condsum2} follows as
  \begin{longaligned}[\label{eqn:eq_k_cd_Exp__ET_condsum2_lhs}]
    \E[T \mid X_{k:k} >& \Delta]\Pr\{X_{k:k} > \Delta\} \\
    &= \sum_{r=0}^{k-1} \E[\Delta + X_{k-r:k-r} \mid R=r]\Pr\{R=r\} \\
    &\stackrel{(a)}{=} \Delta \Pr\{R < k\} + \frac{\E[H_{k-R}]}{(c+1)\mu} \\
    &\stackrel{(b)}{\approx} \Delta \Pr\{R < k\} + \frac{H_{k-kq}}{(c+1)\mu}. \longalignedtag
  \end{longaligned}
  Equality $(a)$ comes by using the mean of exponential order statistics, and $(b)$ is by  Lemma~\ref{lm_approx_bin_mean_of_harmonic__ibeta}. Same steps can be carried out to get the right hand side of the difference in \eqref{eqn:eq_k_cd_Exp__ET_condsum2} as
  \begin{equation}
  \begin{split}
    \E[X_{k:k} \mid X_{k:k} > \Delta]&\Pr\{X_{k:k} > \Delta\} \\
    \approx&~ \Delta \Pr\{R < k\} + \frac{H_{k-kq}}{\mu}.
  \end{split}
  \label{eqn:eq_k_cd_Exp__ET_condsum2_rhs}
  \end{equation}
  Putting \eqref{eqn:eq_k_cd_Exp__ET_condsum2_lhs}, \eqref{eqn:eq_k_cd_Exp__ET_condsum2_rhs} and \eqref{eqn:eq_k_cd_Exp__ET_condsum2} together gives \eqref{eqn:eq_k_cd_Exp_ET}.
  
  \vspace{0.2em}
  \noindent
  \textbf{Cost with task cancellation, $\E[C^c]$ in \eqref{eqn:eq_k_cd_Exp_EC}.}
  When a task completes before time $\Delta$, it contributes to cost as much as its execution time $\E[X \mid X \leq \Delta]$.
  When a task completes after time $\Delta$, it will initially incur $\Delta$ to cost, then will be joined by $c$ fresh replicas at time $\Delta$ and will all together contribute an additional $(c+1)\E[X_{(c+1):1}]$ to cost.
  Given that job consists of $k$ tasks, we can write
  \begin{equation*}
  \begin{split}
    \E[C^c] =&~ k\Bigl(\Pr\{X \leq \Delta\}\E[X \mid X \leq \Delta] \\
    &\quad~ + \Pr\{X > \Delta\}\left(\Delta + (c+1)\E[X_{(c+1):1}] \right) \Bigr) \\
    =&~ k\biggl(\frac{1}{\mu}(1 - e^{-\mu\Delta} - \mu\Delta e^{-\mu\Delta}) \\
    &\quad~ + e^{-\mu\Delta}\left(\Delta + (c+1)\frac{1}{(c+1)\mu}\right) \biggr) = \frac{k}{\mu}.
  \end{split}
  \end{equation*}
  
  \vspace{0.2em}
  \noindent
  \textbf{Cost without task cancellation, $\E[C]$ in \eqref{eqn:eq_k_cd_Exp_EC}.}
  Each of the initial $k$ tasks will execute to completion, hence will contribute $\E[X]$ to cost.
  If a task takes longer than $\Delta$, it will be joined by $c$ fresh replicas, which will contribute an additional $c \cdot \E[X]$ to cost.
  Thus, we can write
  \begin{equation*}
  \begin{split}
    \E[C] = k\left(\E[X] + \Pr\{X > \Delta\}\;c\E[X] \right) = \left(c(1-q) + 1\right)\frac{k}{\mu}.
  \end{split}
  \end{equation*}
\end{proof}

\subsection{Proof of Theorem~\ref{thm_k_cd_SExp_T_C}}
This Theorem assumes task execution times are i.i.d. $\mathrm{SExp}(s, \mu)$, and presents the latency and cost of executing a job of $k$ tasks by employing replicated redundancy after waiting some time $\Delta$.
\begin{proof}
  \vspace{0.2em}
  \noindent
  \textbf{Job completion time, \eqref{eqn:eq_k_cd_SExp_tail__ET}.}
  Let $T_e$ ($T$) be the job completion time when task execution times are distributed as $\mathrm{Exp}(\mu)$ ($\mathrm{SExp}(s, \mu)$).
  Launching $k$ tasks of a job at time $0$, we can treat $s$ as the deterministic portion of their execution, so an $\mathrm{Exp}(\mu)$ timer for each task is started at time $s$. Replicas for the remaining tasks are launched at time $\Delta$, and at time $s + \Delta$ the deterministic portion of their execution time is completed and an $\mathrm{Exp}(\mu)$ timer for each replica is started. This is as if launching $k$ tasks each with $\mathrm{Exp}(\mu)$ execution time at time $s$, then launching the replicas of the remaining tasks at time $s + \Delta$. Thus, we have $T = T_e + s$, which proves \eqref{eqn:eq_k_cd_SExp_tail__ET}.
  
  \vspace{0.2em}
  \noindent
  \textbf{Cost with task cancellation, \eqref{eqn:eq_k_cd_SExp_ECwcancel}.}
  When a task completes before time $\Delta$, it will contribute $\E[X \mid X \leq \Delta]$ to cost. When it completes after time $\Delta$, it will be joined by $c$ fresh replicas and all together will contribute $\Delta + (c+1)\E[Y]$ to cost, where $Y$ denotes the residual lifetime of the job after time $\Delta$, so $Y \sim \min\left\{X_{c:1}, (X-\Delta \mid X > \Delta)\right\}$.
  Multiplying a single task's total contribution to cost by the number of tasks,
  \begin{equation}
  \begin{split}
    \E[C^c] &= \Pr\{X \leq \Delta\}\E[X \mid X \leq \Delta] \\
    &+ \Pr\{X > \Delta\}(c+1)\E[Y].
  \end{split}
  \label{eqn:eq_k_cd_G_ECwcancel}
  \end{equation}
  
  We next derive $\E[X \mid X \leq \Delta]$ and $\E[Y]$.
  Let $F(x)$ and $F_c(x)$ be the CDF and tail distributions of $X$. It is easy to see that $\E[X \mid X \leq \Delta] = 0$ for $\Delta \leq s$. For $\Delta > s$
  \begin{longaligned}
    & \E[X \mid X \leq \Delta] = \Delta - \frac{1}{F(\Delta)}\int_0^{\Delta} F(x) dx \\
    &= \Delta - \frac{1}{q}\int_{s}^{\Delta} (1 - e^{-\mu(x-s)}) dx = \Delta + \frac{1}{\mu} - \frac{\Delta-s}{q}.
  \end{longaligned}
  
  Next we find $\E[Y]$, firstly for $\Delta \leq s$
  \begin{longaligned}
    \E[Y] &= \frac{1}{F_c(\Delta)}\int_0^{\infty} F_c(y)F_c(y + \Delta) dy \\
    &= \int_0^{s - \Delta} 1dy + \int_{s-\Delta}^{s} e^{-\mu(y + \Delta - s)}dy \\
      &\quad + \int_{s}^{\infty} e^{-\mu\left((c+1)(y-s) + \Delta\right)}dy \\
    &= s - \Delta + \frac{1}{\mu}\left(1 - \frac{c}{c+1}e^{-\mu\Delta}\right).
  \end{longaligned}
  and secondly for $\Delta > s$
  \begin{longaligned}
    \E[Y] &= \frac{1}{F_c(\Delta)}\int_0^{\infty} F_c(y)F_c(y + \Delta) dy \\
    &= \int_0^{s} e^{-\mu y}dy + \int_{s}^{\infty} e^{-\mu\left((c+1)y - c\;s\right)}dy \\
    &= \frac{1}{\mu}\left(1 - \frac{c}{c+1}e^{-\mu s}\right).
  \end{longaligned}
  Substituting these in \eqref{eqn:eq_k_cd_G_ECwcancel} yields \eqref{eqn:eq_k_cd_SExp_EC}.
  
  \vspace{0.2em}
  \noindent
  \textbf{Cost without task cancellation, \eqref{eqn:eq_k_cd_SExp_EC}.}
  Each task will contribute $\E[X]$ to cost, and will be joined by $c$ fresh replicas if it takes longer than $\Delta$ amount of time to complete, in which case its replicas will contribute an additional $c\E[X]$ to cost.
  Given that job consists of $k$ tasks, we have
  \[ \E[C] = k\left(\E[X] + \Pr\{X > \Delta\} \;c\E[X]\right). \]
  Substituting $\Pr\{X > \Delta\} = 1 - q$ and $\E[X] = s + 1/\mu$ above yields \eqref{eqn:eq_k_cd_SExp_EC}.
\end{proof}

\subsection{Proof of Theorem~\ref{thm_k_nd_Exp_T_C}}
This Theorem assumes task execution times are i.i.d. $\mathrm{Exp}(s, \mu)$, and presents the latency and cost of executing a job of $k$ tasks by employing coded redundancy after waiting some time $\Delta$.
\begin{proof}
  \vspace{0.2em}
  \noindent
  \textbf{Distribution of job execution time in \eqref{eqn:eq_k_nd_Exp_tail__ET}.}
  By the law of total probability,
  \begin{longaligned}[\label{eqn:eq_k_nd_Exp_ET_tail__condsum}]
    \Pr\{T > t\} &= \Pr\{T > t \mid T \leq \Delta\}Pr\{T \leq \Delta\} \\
    &\quad + \Pr\{T > t \mid T > \Delta\}\Pr\{T > \Delta\}. \longalignedtag
  \end{longaligned}
  where the left hand side of the sum above is
  \begin{longaligned}[\label{eqn:eq_k_nd_Exp_tail__condsum_lhs}]
    \Pr\{T > t &\mid T \leq \Delta\}\Pr\{T \leq \Delta\} \\
    &= \Pr\{X_{k:k} > t \mid X_{k:k} \leq \Delta\}\Pr\{X_{k:k} \leq \Delta\} \\
    &= \mathbbm{1}(t \leq \Delta)\left(\Pr\{X_{k:k} \leq \Delta\} - \Pr\{X_{k:k} \leq t\} \right) \\
    &= \mathbbm{1}(t \leq \Delta)\left((1-e^{-\mu\Delta})^k - (1-e^{-\mu t})^k\right). \longalignedtag
  \end{longaligned}
  Let $R$ the number of tasks completed before time $\Delta$, so $R \sim \mathrm{Bin}(k, q)$ for $q = 1 - e^{-\mu\Delta}$.
  Then, the right hand side of the sum in \eqref{eqn:eq_k_nd_Exp_ET_tail__condsum} becomes
  \begin{longaligned}[\label{eqn:eq_k_nd_Exp_tail__condsum_rhs_1}]
    & \Pr\{T > t \mid T > \Delta\}\Pr\{T > \Delta\} \\
    &= \Pr\{T > t \mid R < k\}\Pr\{R < k\} \\
    &= \sum_{r=0}^{k-1} \Pr\{T > t \mid R=r\}\Pr\{R=r\} \\
    &\stackrel{(a)}{=} \sum_{r=0}^{k-1} \left(1 - \Pr\{X_{n-r:k-r} + \Delta \leq t\}\right)\Pr\{R=r\} \\
	&= (1 - q^k) - \sum_{r=0}^{k-1} \Pr\{X_{n-r:k-r} \leq t - \Delta\}\Pr\{R=r\} \\
    &\stackrel{(b)}{=} (1 - q^k) - \sum_{r=0}^{k-1} I(\bar{q};k-r,n-k+1) \Pr\{R=r\} \\
    &= (1 - q^k) - \E[I(\bar{q};k-R,n-k+1)] \\
    &\quad + q^k I(\bar{q};0,n-k+1). \longalignedtag
  \end{longaligned}
  Equality $(a)$ comes from observing that given that $k-r$ tasks remain and $n-k$ new parity tasks are introduced at time $\Delta$, both old (because exponential distribution is memoryless) and new tasks can be treated as fresh task copies, which equate the residual job lifetime to $X_{n-r:k-r}$.
  Equality $(b)$ is by the definition of order statistics and defining $\bar{q} = \mathbbm{1}(t > \Delta)(1 - e^{-\mu(t-\Delta)})$, and writing the inner finite sum as regularized incomplete beta function.
  By Lemma \ref{lm_approx_bin_mean_of_harmonic__ibeta}, we have the approximation
  \[ \E[I(\bar{q}; k-R, n-k+1)] \approx I(\bar{q}; k(1-q), n-k+1), \]
  substituting which in \eqref{eqn:eq_k_nd_Exp_tail__condsum_rhs_1} gives
  \begin{longaligned}[\label{eqn:eq_k_nd_Exp_tail__condsum_rhs_2}]
    & \Pr\{T \geq t \mid T > \Delta\}\Pr\{T > \Delta\} \\
    &= (1 - q^k) - I(\bar{q}; k(1-q), n-k+1) + q^k I(\bar{q}; 0, n-k+1) \\
    &= I(1-\bar{q}; k(1-q), n-k+1) - q^k I(1-\bar{q}; 0, n-k+1), \longalignedtag
  \end{longaligned}
  where the second equality is obtained from the identity $I(1-q; n, m) = 1 - I(q; m, n)$.
  Finally, putting \eqref{eqn:eq_k_nd_Exp_tail__condsum_lhs} and \eqref{eqn:eq_k_nd_Exp_tail__condsum_rhs_2} together yields the distribution given in \eqref{eqn:eq_k_nd_Exp_tail__ET}.
  
  \vspace{0.2em}
  \noindent
  \textbf{Latency in \eqref{eqn:eq_k_nd_Exp_tail__ET}.}
  By the law of total expectation,
  \begin{equation}
  \begin{split}
    \E[T] &= \E[T \mid T \leq \Delta]\Pr\{T \leq \Delta\} \\
    &+ \E[T \mid T > \Delta]\Pr\{T > \Delta\}.
  \end{split}
  \label{eqn:eq_k_nd_Exp_ET_condsum}
  \end{equation}
  Left hand side of the sum above is
  \begin{longaligned}[\label{eqn:eq_k_nd_Exp_ET_condsum_lhs}]
    & \E[T \mid T \leq \Delta]\Pr\{T \leq \Delta\} \\
    &= \int_0^{\Delta} \Pr\{T \geq t \mid T \leq \Delta\}\Pr\{T \leq \Delta\} dt \\
    &= \Pr\{T \leq \Delta\}\Delta - \int_0^{\Delta} \Pr\{T < t\} dt \\
    &= (1-q^k)\Delta - \int_0^{\Delta} (1-e^{-\mu t}) dt \\
    &\stackrel{(a)}{=} (1-q^k)\Delta - \frac{1}{\mu} \int_{1-q}^1 u^{-1}(1-u)^k du \\
    &= (1-q^k)\Delta \\
      &\quad - \frac{1}{\mu} \left(\int_0^1 u^{-1}(1-u)^k du - \int_0^{1-q} u^{-1}(1-u)^k du \right) \\
    &= (1-q^k)\Delta - \frac{1}{\mu} \left(B(0,k+1) - B(1-q;0,k+1)\right) \\
    &= (1-q^k)\Delta - \frac{1}{\mu} B(0,k+1)\left(1 - I(1-q;0,k+1)\right) \\
    &\stackrel{(b)}{=} (1-q^k)\Delta - \frac{1}{\mu} B(q;k+1,0). \longalignedtag
  \end{longaligned}
  Equality $(a)$ is by doing a change of variables with $u = e^{-\mu x}$, $(c)$ is by the identity $I(q; m, n) = 1 - I(1-q; n, m)$, and $(d)$ is by the identity $B(m, n) = B(n, m)$.
  
  Right hand side of the sum in \eqref{eqn:eq_k_nd_Exp_ET_condsum} is
  \begin{longaligned}[\label{eqn:eq_k_nd_Exp_ET_condsum_rhs}]
    & \E[T \mid T > \Delta]\Pr\{T > \Delta\} \\
    &\stackrel{(a)}{=} \sum_{r=0}^{k-1} \E[T \mid R=r, R < k]\Pr\{R=r \mid R < k\}\Pr\{R < k\} \\
    &\stackrel{(b)}{=} \sum_{r=0}^{k-1} \E[\Delta + X_{n-r:k-r}]\Pr\{R=r\} \\
    &\stackrel{(c)}{=} \left(1-\Pr\{R=k\}\right)\Delta + \sum_{r=0}^{k} \left(H_{n-r} - H_{n-k}\right)\Pr\{R=r\} \\
    &= q^k \Delta + \E[H_{n-R} - H_{n-k}] \approx q^k \Delta + H_{n-kq} - H_{n-k}. \longalignedtag
  \end{longaligned}
  Equality $(a)$ is by the law of total expectation, $(b)$ is by observing that $T \mid R=r \sim \Delta + X_{n-r:k-r}$ due to the memoryless property of exponential distribution, $(c)$ is by the expected value of exponential order statistics; $\E[X_{n-r:k-r}] = H_{n-r} - H_{n-k}$.
  The last step follows from the approximation $\E[H_{n-R}] \approx H_{n-kq}$ given in Lemma \ref{lm_approx_bin_mean_of_harmonic__ibeta}.
  Putting \eqref{eqn:eq_k_nd_Exp_ET_condsum_lhs} and \eqref{eqn:eq_k_nd_Exp_ET_condsum_rhs} in \eqref{eqn:eq_k_nd_Exp_ET_condsum} gives the latency given in \eqref{eqn:eq_k_nd_Exp_tail__ET}.
  
  \vspace{0.2em}
  \noindent
  \textbf{Cost without task cancellation, $\E[C]$ in \eqref{eqn:eq_k_nd_Exp_EC}.}
  Each of the initial $k$ tasks will execute until completion and contribute $1/\mu$ to cost.
  If job does not complete until time $\Delta$, additional $n-k$ coded tasks will be launched and each will also contribute $1/\mu$ to cost.
  Thus, we have
  \begin{longaligned}
    \E[C] &= \frac{k}{\mu} + \Pr\{T > \Delta\}\frac{n-k}{\mu} \\
    &\stackrel{(a)}{=} \frac{k}{\mu} + (1 - q^k)\frac{n-k}{\mu} = \frac{k}{\mu} q^k + \frac{n}{\mu} (1 - q^k)
  \end{longaligned}
  Equality $(a)$ comes from
  \[ \Pr\{T > \Delta\} = \Pr\{X_{k:k} > \Delta\} = 1 - q^k. \]
  
  \vspace{0.2em}
  \noindent
  \textbf{Cost with task cancellation, $\E[C^c]$ in \eqref{eqn:eq_k_nd_Exp_EC}.}
  By the law of total expectation,
  \begin{longaligned}[\label{eqn:eq_k_nd_Exp__ECwcancel_condsum_rhs}]
    \E[C^c] &= \E[C^c \mid T \leq \Delta]\Pr\{T \leq \Delta\} \\
    &\quad + \E[C^c \mid T > \Delta]\Pr\{T > \Delta\} \\
    &\stackrel{(a)}{=} \E\left[C \mid T \leq \Delta\right]\Pr\{T \leq \Delta\} \\
    &\quad + \E\left[C-\frac{n-k}{\mu} \mid T > \Delta\right]\Pr\{T > \Delta\} \\
    &= \E[C] - \frac{n-k}{\mu}\Pr\{T > \Delta\} = \frac{k}{\mu}. \longalignedtag
  \end{longaligned}
  Equality $(a)$ comes by noticing that the cost with task cancellation can be found by subtracting total residual lifetime of the tasks that continue after job completion from the cost without task cancellation.
  In the case without task cancellation, $n-k$ tasks run until completion after job completion if the job takes longer than $\Delta$ to complete. Using the memoryless property of exponential distribution, sum of the expected residual lifetime of these remaining tasks is $(n-k)/\mu$.
\end{proof}

\subsection{Proof of Theorem~\ref{thm_k_nd_SExp_T_C}}
This Theorem assumes task execution times are i.i.d. $\mathrm{SExp}(s, \mu)$, and presents the latency and cost of executing a job of $k$ tasks by employing coded redundancy after waiting some time $\Delta$.
\begin{proof}
  \vspace{0.2em}
  \noindent
  \textbf{Distribution and expected value of job execution time, \eqref{eqn:eq_k_nd_SExp_tail__ET}.}
  Derivations follow from the exact same arguments outlined in the Proof of Thm.~\ref{thm_k_cd_SExp_T_C}, which we do not repeat here.
  
  \vspace{0.2em}
  \noindent
  \textbf{Cost without task cancellation.}
  When $\Delta \leq s$, no task can complete before time $\Delta$ and redundant tasks are always going to be launched. Since all tasks will run until completion (with no cancellation), then the cost is simply the overall number of task times the average lifetime of each task.
  
  Next, we consider the case with $\Delta > s$.
  Until time $s$, all of the initial $k$ tasks are spending the deterministic portion of their lifetime. If the residual lifetime of the job after time $s$, which is distributed as $X_{k:k}$ with $X \sim \mathrm{Exp}(\mu)$, less than $\Delta-s$, then redundant tasks will not be launched, otherwise they will be.
  Thus, we can write the cost as
  \begin{longaligned}
    \E[C] &= k\E[X] + \Pr\{X_{k:k} > \Delta-s\}(n-k)\E[X] \\
    &= (k + (1 - q^k)(n-k))\left(s + \frac{1}{\mu}\right).
  \end{longaligned}
  
  \vspace{0.2em}
  \noindent
  \textbf{Cost with task cancellation.}
  We first consider the case with $\Delta \leq s$.
  Each of the initial $k$ tasks will complete the deterministic portion of their lifetime at time $s$.
  Thus, no task can complete before time $\Delta$ and redundant tasks are always going to be launched at time $\Delta$.
  Each redundant task will complete the deterministic portion of its lifetime at time $s + \Delta$.
  Given these, job execution after time $s$ can be expressed with an exponential-setup; as if $k$ tasks were launched at time $s$ and $n-k$ coded tasks are launched at time $s + \Delta$ with execution times of all task copies distributed as $\mathrm{Exp}(\mu)$.
  
  We already know from Thm.~\ref{thm_k_nd_Exp_T_C} that cost in the exponential-setup is $k/\mu$. In addition to that, we need to consider two factors that contribute to cost in our shifted-exponential setup:
  i) $k$ initial tasks run from time $0$ and $s$, hence contributing $ks$ to cost,
  ii) $n-k$ redundant tasks run from time $\Delta$ to time $\Delta + s$ if job takes longer than $\Delta + s$ ($T > \Delta + s$) to complete, otherwise, they get cancelled some time between time $s$ and $\Delta + s$, hence contributing $(n-k)\left(s - \Pr\{T \leq \Delta + s\}(s - \E[T - \Delta \mid T \leq \Delta + s])\right)$.
  Event $\{T \leq \Delta + s\}$ can be expressed as $\{\tilde{X}_{k:k} \leq \Delta\}$, and $T - \Delta = \Delta - \tilde{X}_{k:k}$ for $\tilde{X} \sim \mathrm{Exp}(\mu)$.
  Putting all the observations given above together
  \begin{longaligned}
    \E[C^c] &= \frac{k}{\mu} + ks + (n-k) \\
      &\times \left(s - \Pr\{\tilde{X}_{k:k} \leq \Delta\}\left(\Delta - \E[\tilde{X}_{k:k} \mid \tilde{X}_{k:k} \leq \Delta]\right)\right) \\
    &= \frac{k}{\mu} + ns - (n-k) \\
      &\times \Pr\{\tilde{X}_{k:k} \leq \Delta\}\left(\Delta - \E[\tilde{X}_{k:k} \mid \tilde{X}_{k:k} \leq \Delta] \right)
  \end{longaligned}
  We firstly derive
  \begin{longaligned}
    \Pr&\{\tilde{X}_{k:k} \leq \Delta\}\E[\tilde{X}_{k:k} \mid \tilde{X}_{k:k} \leq \Delta] = \int_{0}^{\Delta} x\Pr\{\tilde{X}_{k:k} = x\}dx \\
    &= \int_{0}^{\Delta} x k\Pr\{\tilde{X} \leq x\}^{k-1}\Pr\{\tilde{X} = x\} dx \\
    &= k\mu \int_{0}^{\Delta} x (1 - e^{-\mu x})^{k-1} e^{-\mu x}dx \\
    &\stackrel{(a)}{=} k\mu \sum_{i=0}^{k-1} (-1)^{k-1-i} \binom{k-1}{i} \int_{0}^{\Delta} x e^{-\mu(k-i)x}dx \\
    &= k\mu \sum_{i=0}^{k-1} (-1)^{k-1-i} \binom{k-1}{i} \\
      &\qquad\qquad \times \left(\frac{1 - e^{-\mu(k-i)\Delta}}{\mu(k-i)} - \Delta e^{-\mu(k-i)\Delta} \right) \\
    &= k\mu\Biggl(\frac{1}{\mu}\sum_{i=0}^{k-1} (-1)^{k-1-i} \binom{k-1}{i} \frac{1 - e^{-\mu(k-i)\Delta}}{k-i} \\
      &\qquad\quad - \Delta e^{-\mu\Delta}\left(1 - e^{-\mu\Delta}\right)^{k-1} \Biggr)
  \end{longaligned}
  Equality $(a)$ comes from writing out the binomial expansion of $(1 - e^{-\mu x})^{k-1}$ and interchanging the sum and integral.
  Substituting $q = 1 - e^{-\mu\Delta}$ in the above expressions, we get
  \begin{longaligned}
    \E[C^c] &= \frac{k}{\mu} + ns - (n-k)q^k \\
      &\quad \times \left(\Delta + k\mu\left(\frac{\eta}{\mu q^k} - \Delta\left(\frac{1}{q}-1\right)\right)\right),
  \end{longaligned}
  where
  \[ \eta = \sum_{i=0}^{k-1} (-1)^{k-1-i} \binom{k-1}{i} \frac{1 - (1-q)^{k-i}}{k-i}. \]
  
  Next, we consider the case with $\Delta > s$.
  Let us define $\Delta^{\prime} = \Delta + s/k$ and $\tilde{q} = 1 - e^{-\mu\Delta}$.
  By the law of total expectation,
  \begin{longaligned}[\label{eqn:eq_k_nd_SExp__ECwcancel_1}]
    & \E[C^c] = \E[C^c \mid T \leq \Delta]\Pr\{T \leq \Delta\} \\
      &\quad + \E[C^c \mid \Delta < T \leq \Delta^{\prime}]\Pr\{\Delta < T \leq \Delta^{\prime}\} \\
      &\quad + \E[C^c \mid T > \Delta^{\prime}]\Pr\{T > \Delta^{\prime}\} \\
    &\stackrel{(a)}{=} \E[C \mid T \leq \Delta]\Pr\{T \leq \Delta\} \\
      &\quad + \E\biggl[C - (n-k)\left(\Delta^{\prime}-T+\frac{1}{\mu}\right) \\ 
      &\qquad\quad~ \mid \Delta < T \leq \Delta^{\prime}\biggr]\Pr\{\Delta < T \leq \Delta^{\prime}\} \\
      &\quad + \E\left[C - \frac{(n-k)}{\mu} \mid T > \Delta^{\prime}\right]\Pr\{T > \Delta^{\prime}\} \\
    &= \E[C] - (n-k)\\
      &\quad \times \Bigl(\E[\Delta^{\prime}-T \mid \Delta < T \leq \Delta^{\prime}]\Pr\{\Delta < T \leq \Delta^{\prime}\} \\
        &\qquad\quad + \frac{1}{\mu}\Pr\{T > \Delta\}\Bigr) \\
    &= \E[C] - (n-k)\Bigl(\E[\Delta^{\prime}-T \mid \Delta < T \leq \Delta^{\prime}](\tilde{q}^k-q^k) \\
      &\qquad\qquad\qquad\qquad~ + \frac{1}{\mu}(1-q^k)\Bigr). \longalignedtag
  \end{longaligned}
  Equality $(a)$ comes from observing:
  1) $C^c \mid T \leq \Delta = C \mid T \leq \Delta$ since no redundant tasks are launched if job completes before time $\Delta$,
  2) $(C^c \mid \Delta < T \leq \Delta^{\prime}) = (C - (n-k)(\Delta^{\prime}-T+1/\mu) \mid \Delta < T \leq \Delta^{\prime})$ since the $n-k$ coded tasks that are launched at time $\Delta$ is canceled at the job completion time $T$ before they could even finish the deterministic portion ($s/k$) of their lifetime,
  3) $C^c \mid T > \Delta^{\prime} = (C - (n-k)/\mu \mid T > \Delta^{\prime})$ since all the canceled tasks finish the deterministic portion of their lifetime before they get canceled and the residual lifetime of each is distributed as $\mathrm{Exp}(\mu)$.
  
  In the following, we derive an approximation for $\E[\Delta^{\prime}-T \mid \Delta < T \leq \Delta^{\prime}]$.
  \begin{longaligned}[\label{eqn:eq_k_nd_SExp__ECwcancel_2}]
    & \E[\Delta^{\prime}-T \mid \Delta < T \leq \Delta^{\prime}] \stackrel{(a)}{=} \E[\Delta-T_e \mid \Delta-s/k < T_e \leq \Delta] \\
    &\stackrel{(b)}{=} \sum_{r=0}^{k-1} \E[\Delta-T_e \mid R=r]\Pr\{R=r\} \\
    &\stackrel{(c)}{=} \sum_{r=0}^{k-1} \left(s/k - \E[X_{k-r:k-r} \mid X_{k-r:k-r} \leq s/k]\right)\Pr\{R=r\} \\
    &= \sum_{r=0}^{k-1} \left(s/k - \E[X_{k-r:k-r} \mid X_{k-r:k-r} \leq s/k]\right)\Pr\{R=r\} \\
    &\stackrel{(d)}{=} \E\left[\int_0^{s/k} \left(\frac{F(x)}{F(s/k)}\right)^{k-R} dx\right] \\
    &\stackrel{(e)}{=} \frac{1}{\mu} \E\left[\int_0^{\alpha} \left(\frac{u}{\alpha}\right)^{k-R}(1-u)^{-1} du\right] \\
    &= \frac{1}{\mu} \int_0^{\alpha} \E\left[\left(\frac{u}{\alpha}\right)^{k-R}\right](1-u)^{-1} du \\
    &\stackrel{(f)}{\approx} \frac{1}{\mu} \int_0^{\alpha} \left(\frac{u}{\alpha}\right)^{k(1-q)}(1-u)^{-1} du \\
    &= \frac{1}{\mu \alpha^{k(1-q)}} B(\alpha;k-kq+1,0). \longalignedtag
  \end{longaligned}
  Equality $(a)$ results from the fact $T = T_e + s/k$.
  Equality $(b)$ is by denoting the number of tasks completed before $\Delta$ as $R \sim \mathrm{Binomial}(k, q)$ and using the memoryless property of exponential distribution. 
  Equality $(c)$ is by observing that $\Delta-T_e \mid R=r$ is simply the time between $\Delta$ and the maximum of $k-r$ exponential timers that started at $\Delta-s/k$, which is less than $\Delta$.
  Expectation in $(d)$ is with respect to $R$ where $F$ denotes the CDF of $\mathrm{Exp}(\mu)$.
  Equality $(e)$ is by setting $\alpha = F(s/k)$ and by doing a change of variables with $u = F(x)$.
  Equality $(f)$ is by substituting the approximation
  \[ \E\left[(u/\alpha)^{k-R}\right] = (q + (1-q)u/\alpha)^k \approx (u/\alpha)^{k(1-q)}, \]
  which comes from observing $\lim_{q \to 0} (q + (1-q)u/\alpha)^k = x^k$ and $\lim_{q \to 1} (q + (1-q)u/\alpha)^k = 1$.
  Substituting \eqref{eqn:eq_k_nd_SExp__ECwcancel_2} in \eqref{eqn:eq_k_nd_SExp__ECwcancel_1} yields the cost expression.
\end{proof}

\subsection{Proof of Theorem~\ref{thm_k_cn_ET__EC}}
This Theorem assumes task execution times are i.i.d. $\mathrm{Pareto}(s, \alpha)$, and presents the latency and cost of executing a job of $k$ tasks by launching the job together with $c$ replicas for each task or $n-k$ coded tasks.
\begin{proof}[Proof Sketch]
  Let task execution times be distributed as random variable $X$.
  When job is launched by adding $c$ replicas for each task, lifetime of each task is given by $X_{c+1:1}$.
  Then latency follows as $\E[(X_{c+1:1})_{k:k}]$ and cost as $k \E[X_{c+1:1}]$. Closed form expressions for both come from the first principles of order statistics.
  
  When job is launched by adding $n-k$ coded tasks, latency follows as $\E[X_{n:k}]$ and cost as $\E\left[\sum_{i=1}^k X_{n:i} + (n-k)X_{n:k}\right]$. Again the closed form expressions for both come from the first principles of order statistics.
\end{proof}

\subsection{Proof of Theorem~\ref{thm_coding_vs_rep_less_latency_cost}}
This Theorem states that executing a distributed job with coded redundancy achieves lower cost and latency than executing it with replicated redundancy.
\begin{proof}
  Proof is simple and omitted here for the case with Shifted-Exponential task execution times.
  We here present the proof for the case with Pareto task execution times.
  
  \vspace{0.2em}
  \noindent
  \textbf{Latency.}
  Let us denote the job execution time with $T_{(k, (c+1)k)}$ when $kc$ coded tasks are added into execution, and with $T_{(k, c)}$ when $c$ replicas are launched for each of its $k$ tasks.
  Firstly, we show
  \[ \E[T_{(k, (c+1)k)}] < \E[T_{(k, c)}]. \]
  Using the expressions given in Thm.~\ref{thm_k_cn_ET__EC},
  \begin{longaligned}
    & \E[T_{(k, (c+1)k)}]/\E[T_{(k, c)}] = \frac{\Gamma((c+1)k+1)}{\Gamma((c+1)k+1-1/\alpha)} \\
    &\qquad\qquad \times \frac{\Gamma(ck+1-1/\alpha)}{\Gamma(ck+1)} \frac{\Gamma(k+1-1/(c+1)\alpha)}{\Gamma(k+1)\Gamma(1-1/(c+1)\alpha)}.
  \end{longaligned}
  which goes to $1$ as $c \to \infty$. Then, if we could show that this ratio given above is monotonically increasing with $c$, we could conclude that $\E[T_{(k, (c+1)k)}] < \E[T_{(k, c)}]$ for all $c$.
  Weierstrass' definition of $\Gamma$ function \cite[eq.~5.8.2]{NIST:DLMF} gives
  \[ \ln(\Gamma(z)) = -\ln(z) - \gamma z + \sum_{i=1}^{\infty} \frac{z}{i} - \ln\left(1+\frac{z}{i}\right). \]
  We then take $\ln$ of the ratio given above and show that it is monotonically increasing with $c$. Note that multiplication with $\Gamma(k+1)$ does not affect the monotonicity and only serves us by simplifying the expressions.
  \begin{longaligned}[\label{eqn:eq_gamma_ratio1}]
    \ln(\Gamma(k+1) & \E[T_{(k, (c+1)k)}] / \E[T_{(k, c)}]) = \\
    & \gamma k + \sum_{i=1}^{\infty} \frac{k}{i} - \ln(f(0)) - \sum_{i=1}^{\infty} \ln(f(i)). \longalignedtag
  \end{longaligned}
  where
  \begin{longaligned}
    f(i) &= \left(\frac{1 + k/(ck+1)}{1 + k/(ck+1-1/\alpha)} \right) \left(1 + \frac{k}{1-1/(c+1)\alpha} \right).
  \end{longaligned}
  It is easy to see that $f(i)$ monotonically decreases with $c$ and the rate of reduction decreases with increasing $c$. Thus, $-\ln(f(0))$ and $-\ln(f(i))$ in eq.~\ref{eqn:eq_gamma_ratio1} monotonically increases with $c$.
  Overall, this proves that $\E[T_{(k, (c+1)k)}] / \E[T_{(k, c)}]$ monotonically increases with $c$ at a decreasing rate and goes to $1$ as $c \to \infty$.
  
  \vspace{0.2em}
  \noindent
  \textbf{Cost.}
  Let us denote the cost of job execution with $C_{(k, (c+1)k)}$ when $kc$ coded tasks are added into execution, and with $C_{(k, c)}$ when $c$ replicas are launched for each of its $k$ tasks.
  Secondly, we show
  \[ \E[C_{(k, (c+1)k)}] < \E[C_{(k, c)}]. \]
  Using the expressions given in Thm.~\ref{thm_k_cn_ET__EC} and after some simple algebra, the inequality above corresponds to
  \[ \frac{\Gamma(n+1)}{\Gamma(n-k+1)} \frac{\Gamma(n-k+1-1/\alpha)}{\Gamma(n+1-1/\alpha)} > 1 + \frac{k}{\alpha n-k}, \]
  which can be written as
  \[ \frac{B(k, n-k+1-1/\alpha)}{B(k, n-k+1)} > 1 + \frac{k}{\alpha n-k}. \]
  Writing the right hand side of the inequality in terms of $B$ function, we get
  \begin{equation}
    \frac{B(k, n-k+1-1/\alpha) / B(k, n-k+1)}{B(k, \alpha n-k) / B(k, \alpha n-k+1)} > 1.
  \label{eqn:eq_ratio_of_betas_g_1}
  \end{equation}
  Expressions involve four $B$ functions, all of which having their first parameter being equal to $k$.
  For all $k \geq 1$, we have
  \[ \frac{1/\alpha}{n-k+1} > \frac{1}{\alpha n-k+1}. \]
  This means for \eqref{eqn:eq_ratio_of_betas_g_1} that the relative difference between the second parameters in numerator ($n-k+1-1/\alpha$, $n-k+1$) is greater than the relative difference between the second parameters ($\alpha n-k$, $\alpha n-k+1$) in denominator.
  This together with $\alpha n-k > n-k$ and the fact that $B(k, x)$ is a convex function decreasing with $k$ shows that \eqref{eqn:eq_ratio_of_betas_g_1} holds.
\end{proof}

\subsection{Proof of Corollary~\ref{cor_k_cn_Pareto_reduc_in_ET_for_baseline_EC}}
This Corollary presents the necessary and sufficient conditions on the tail index for replicated or coded redundancy to be able to reduce latency without exceeding the baseline cost of running with no redundancy.
\begin{proof}{}
  \vspace{0.2em}
  \noindent
  \textbf{Replication.}
  We firstly consider adding $c$ replicas for each task.
  Latency is a decreasing function of $c$. Cost however decreases with $c$ up to a level, then it starts growing monotonically in $c$.
  Let us define $c_{\max}$ such that
  \[ \E[C_{c=c_{\max}}] < \E[C_{c=0}] < \E[C_{c=c_{\max}+1}]. \]
  Then by definition, $\E[T_{\min}] = \E[T_{c=c_{\max}}]$.
  We have the following inequality hold
  \begin{equation*}
    \E[C_c] < \E[C_{c=0}] \iff s k\left(\frac{(c+1)^2\alpha}{(c+1)\alpha-1} - \frac{\alpha}{\alpha-1}\right) < 0
  \end{equation*}
  if and only if $c < 1/(\alpha-1)-1$. Since $c$ is a non-negative integer, we write this condition as
  \[ c_{\max} = \max\Set{\floor*{1/(\alpha-1)}-1, 0}. \]
  
  \vspace{0.2em}
  \noindent
  \textbf{Coding.}
  We secondly consider adding $n-k$ coded tasks into the execution.
  Cost expression in this case does not allow to find $n_{\max}$ such that
  \[ \E[C_{n=n_{\max}}] < \E[C_{n=k}] < \E[C_{n=n_{\max}+1}]. \]
  Instead we can express $\E[C]$ as a function of $\E[T]$ as
  \[ \E[C] = s\frac{n}{\alpha-1}\left(\alpha - \frac{n-k}{n s}\E[T]\right). \]
  We can then find an approximation for $\E[T_{\min}]$ directly by relating $\E[C_{n_{\max}}]$ to $\E[C_{n=k}]$. We have
  \[ \E[C_{n_{\max}}] < \E[C_{n=k}] \iff \E[T_{n_{\max}}] < s\alpha + \frac{\E[T_{n=k}]}{n_{\max}-k} \]
  from which \eqref{eqn:eq_k_n_Pareto_reduc_in_ET_for_base_EC} follows. Upper bound given in \eqref{eqn:eq_k_n_Pareto_reduc_in_ET_for_base_EC__ineq} follows from this by setting $n_{\max} = k+1$.
  
  We next find the necessary and sufficient conditions on the tail index. 
  The inequality 
  \[ \E[C_n] - \E[C_k] \leq 0 \]
  holds if and only if
  \[ \alpha \leq \frac{\Gamma(n+1)}{\Gamma(n+1-1/\alpha)} \frac{\Gamma(n-k+1-1/\alpha)}{\Gamma(n-k+1)}. \]
  Denoting the right hand side of the above inequality as $h(k, n, \alpha)$ and using Gautschi's inequality \cite[eq.~5.6.4]{NIST:DLMF}, we obtain
  \[ \left(\frac{n}{n-k+1}\right)^{1/\alpha} \leq h(k, n, \alpha) \leq \left(\frac{n+1}{n-k}\right)^{1/\alpha}. \]
  From the last two inequalities above, necessary and sufficient conditions on $\alpha$ given in \eqref{eqn:eq_k_n_neccond_on_a} and \eqref{eqn:eq_k_n_suffcond_on_a} follow.
\end{proof}

\subsection{Proof of Theorem~\ref{thm_suffcond_ETgainpain}}
This Theorem presents necessary and/or sufficient conditions on the growth of the tail heaviness of task execution times in order to be able to reduce latency by employing a higher job expansion rate with replication or coding.
\begin{proof}
  \vspace{0.2em}
  \noindent
  \textbf{Coding.}
  Firstly, we consider expanding jobs with \textit{coded} tasks at a rate of $r_i$, in which case task execution times are distributed as $\mathrm{Pareto}(s, \alpha_i)$.
  Then, latency of executing a job of $k$ tasks is given as
  \[ \E[T_i] = s\frac{\Gamma(n_i+1)}{\Gamma(n_i+1-1/\alpha_i)}\frac{\Gamma(n_i-k+1-1/\alpha_i)}{\Gamma(n_i-k+1)}. \]
  Given $\alpha_i > 1$ and using Gautschi's inequality \cite{GammaIneqs:Gautschi59}, we obtain the following inequalities
  \begin{equation*}
  \begin{split}
    n_i^{1/\alpha_i} < &\frac{\Gamma(n_i+1)}{\Gamma(n_i+1-1/\alpha_i)} < (n_i+1)^{1/\alpha_i}, \\
    (n_i-k)^{1/\alpha_i} < &\frac{\Gamma(n_i-k+1)}{\Gamma(n_i-k+1-1/\alpha_i)} < (n_i-k+1)^{1/\alpha_i}.
  \end{split}
  \end{equation*}
  Using these, we obtain
  \begin{equation*}
    s\left(\frac{n_i}{n_i-k+1}\right)^{1/\alpha_i} < \E[T_i] < s\left(\frac{n_i+1}{n_i-k}\right)^{1/\alpha_i}.
  \end{equation*}
  Given $r_j > r_i$, it is then straightforward to derive the following upper and lower bounds
  \begin{equation}
    \frac{\E[T_j]}{\E[T_i]} < \left(\frac{n_j+1}{n_j-k}\right)^{1/\alpha_j} \left(\frac{n_i}{n_i-k+1}\right)^{-1/\alpha_i},
  \label{eq:eq_Tj_over_Ti_ub}
  \end{equation}
  \begin{equation}
    \frac{\E[T_j]}{\E[T_i]} > \left(\frac{n_j}{n_j-k+1}\right)^{1/\alpha_j} \left(\frac{n_i+1}{n_i-k}\right)^{-1/\alpha_i}.
  \label{eq:eq_Tj_over_Ti_lb}
  \end{equation}
  Sufficient condition \eqref{eq:eq_suffcond_ETred_coding} to yield a reduction in latency is obtained using the upper bound \eqref{eq:eq_Tj_over_Ti_ub} as
  \begin{equation*}
  \begin{split}
    & \left(\frac{n_j+1}{n_j-k}\right)^{1/\alpha_j} \left(\frac{n_i}{n_i-k+1}\right)^{-1/\alpha_i} \leq 1 \\
    & \iff \frac{\alpha_i}{\alpha_j} \leq \log\left(\frac{n_i}{n_i-k+1}\right)/\log\left(\frac{n_j+1}{n_j-k}\right).
  \end{split}
  \end{equation*}
  Similarly, sufficient condition \eqref{eq:eq_necessandsuffcond_ETred_rep} to incur an increase in latency is obtained using the lower bound \eqref{eq:eq_Tj_over_Ti_lb} as
  \begin{equation*}
  \begin{split}
    & (\frac{n_j}{n_j-k+1})^{1/\alpha_j} \left(\frac{n_i+1}{n_i-k}\right)^{-1/\alpha_i} \geq 1 \\
    & \iff \frac{\alpha_i}{\alpha_j} \geq \log\left(\frac{n_i+1}{n_i-k}\right)/\log\left(\frac{n_j}{n_j-k+1}\right).
  \end{split}
  \end{equation*}
  
  \vspace{0.2em}
  \noindent
  \textbf{Replication.}
  Secondly, we consider expanding jobs with \textit{replica} tasks at a rate of $r_i$, in which case task execution times are distributed as $\mathrm{Pareto}(s, \alpha_i)$.
  Then, latency of executing a job of $k$ tasks is given as
  \[ \E[T_i] = s k! \frac{\Gamma(x_i)}{\Gamma(k + x_i)} \]
  where $x_i = 1 - 1/(r_i \alpha_i)$.
  Then, increasing the expansion rate from $r_i$ to $r_j$, latency is reduced (i.e., $\E[T_j]/\E[T_i] < 1$) if and only if
  \[ \frac{\Gamma(x_j)}{\Gamma(k + x_j)} \frac{\Gamma(k + x_i)}{\Gamma(x_i)} < 1, \]
  which is equivalent to
  \[ B(k, x_j) / B(k, x_i) < 1. \]
  Beta $B(k, x)$ function is monotonically decreasing in $x > 0$, hence the above inequality is equivalent to $x_i < x_j$, from which the necessary and sufficient condition given in \eqref{eq:eq_necessandsuffcond_ETred_rep} follows.
\end{proof}

\subsection{Proof of Theorem~\ref{thm_k_wrelaunch_T_C}}
This Theorem assumes task execution times are i.i.d. $\mathrm{Pareto}(s, \alpha)$, and presents the latency and cost of executing a job of $k$ tasks by relaunching remaining tasks after waiting some time $\Delta$.
\begin{proof}[Sketch]
  Defining random variable $R$ as the number of tasks completed before $\Delta$, $R \sim \mathrm{Binomial}(k, q)$ where $q = \mathbbm{1}(\Delta > \lambda)(1 - (\frac{\lambda}{\Delta})^{\alpha})$.
  Derivations of the distribution, latency and cost of job execution follow from the law of total probability or expectation by conditioning on $R$.
  We omit the details here because the results presented in this Theorem are special cases of those given in Thm.~\ref{thm_k_cnd_wrelaunch_ET__EC}. We refer the reader to the proof of Thm.~\ref{thm_k_cnd_wrelaunch_ET__EC} for detailed derivations.
\end{proof}

\subsection{Proof of Lemma~\ref{lm_k_wrelaunch_ET__opt_d_suff_a}}
This Lemma presents sufficient conditions for straggler relaunch to reduce cost and latency of distributed job execution, and an asymptotically exact approximate for optimal relaunch time.
\begin{proof}
  As $k \to \infty$, $q^k \to 0$ and $I(1-q, 1-1/\alpha, k) \to 1$,
  \[ \implies \E[T] \to \Delta - g(k, \alpha)s/\Delta. \]
  Approximate $\Delta^*$ given in \eqref{eqn:eq_k_wrelaunch_approx_optd} comes from firstly taking the limiting value of $\E[T]$ given above as an approximation for the latency, then solving for $\Delta$ by setting the first order derivative of approximate latency with respect to $\Delta$ to zero.
  
  Approximation for $p^*$ follows by observing that the number of tasks that complete before $\Delta^*$ is $R \sim \mathrm{Binomial}(k, q^*)$ where $q^* = 1 - (s/\Delta^*)^{\alpha}$. Average fraction of the tasks that are relaunched follows by writing out $p^* = 1 - q^* = (s/\Delta^*)^{\alpha}$, and then substituting the approximate $\Delta^*$ found above.
  
  Next, we show the sufficient condition \eqref{eqn:eq_k_wrelaunch_suffcond} for straggler relaunch to be able to reduce the cost and latency.
  We know
  \[ \E[T_{no rel}] = s k!\Gamma(1-1/\alpha)/\Gamma(k+1-1/\alpha). \]
  Then we have
  \begin{equation*}
  \begin{split}
    \E[T - T_{no rel}] &= \Delta(1-q^k) \\
    &\quad - (1-s/\Delta) \E[T_{no rel}]I(1-q; 1-1/\alpha, k), 
  \end{split}
  \end{equation*}
  which is equal to $\Delta - (1-s/\Delta) \E[T_{no rel}]$ in the limit $k \to \infty$.
  Using this limit value as an approximation,
  \begin{equation}
    \E[T - T_{no rel}] \lesssim 0 \iff \E[T_{no rel}] \gtrsim \Delta/(1 - s/\Delta).
  \label{eq:eq_ET_ETnorel}
  \end{equation}
  Maximum value of the right hand side in the above inequality is $4s$ (obtained when $\Delta = 2s$), which gives us \eqref{eqn:eq_k_wrelaunch_suffcond}.
  Substituting the exact expression for $\E[T_{no rel}]$ in \eqref{eq:eq_ET_ETnorel},
  \[ s \Gamma(1-1/\alpha)\Gamma(k+1)/\Gamma(k+1-1/\alpha) > 4s, \]
  which holds if the following looser inequality holds
  \[ \Gamma(k+1)/\Gamma(k+1-1/\alpha) > 4. \]
  Lower bounding the left hand side using Gautschi's inequality \cite[eq.~5.6.4]{NIST:DLMF} yields an even looser sufficient condition as
  \[ k^{1/\alpha} > 4, \]
  which implies \eqref{eqn:eq_k_wrelaunch_suffcond_a}.
\end{proof}

\subsection{Proof of Theorem~\ref{thm_k_cn_wrelaunch_ET}}
This Theorem assumes task execution times are i.i.d. $\mathrm{Pareto}(s, \alpha)$, and presents the latency of executing a job of $k$ tasks by launching the job together with $c$ replicas for each task or $n-k$ coded tasks and relaunching all remaining tasks (initial or redundant) after waiting some time $\Delta$.
\begin{proof}
  \vspace{0.2em}
  \noindent
  \textbf{Case $\mathbf{\Delta \leq s}$.}
  When $\Delta \leq s$, relaunching takes place before the minimum task completion time $s$. Since none of the tasks can complete before $s$, relaunching will cause a complete restart of all the tasks (initial and redundant).
  Therefore, regardless of the type of redundant tasks used, the latency will be $\Delta$ plus the latency of the case with no relaunch as given in Thm.~\ref{thm_k_cn_ET__EC}.
  In the remainder of the proof, we consider $\Delta > s$.
  
  \vspace{0.2em}
  \noindent
  \textbf{Replication, $\mathbf{\Delta > s}$.}
  When $c$ replicas are added for each of the $k$ tasks, a task will complete as soon as the fastest of its $c+1$ copies completes. Thus, execution time of each task is distributed as $X_{c+1:1} \sim \mathrm{Pareto}(s, (c+1)\alpha)$.
  Then, latency (given in \eqref{eqn:eq_k_c_wrelaunch_ET}) follows from replacing each $\alpha$ in the latency expression given in \eqref{eqn:eq_k_wrelaunch_tail__ET} with $(c+1)\alpha$.
  
  \vspace{0.2em}
  \noindent
  \textbf{Coding, $\mathbf{\Delta > s}$.}
  Next, we consider launching the job by adding $n-k$ coded tasks.
  Let us denote the number of tasks completed before time $\Delta$ as $R \sim \mathrm{Binomial}(k, q)$ where $q = \mathbbm{1}(\Delta > s)(1 - (s/\Delta)^{\alpha})$. By the law of total expectation,
  \begin{longaligned}[\label{eqn:eq_k_n_wrelaunch_ET_totalexpec}]
    \E[T] &= \E[T \mid T \leq \Delta]\Pr\{T \leq \Delta\} + \E[T \mid T > \Delta]\Pr\{T > \Delta\} \\
    &= \E[X_{n:k} \mid R \geq k]\Pr\{R \geq k\} \\
      &\quad + \E[\Delta + X_{n-R:k-R} \mid R < k]\Pr\{R < k\}. \longalignedtag
  \end{longaligned}
  
  Let $\E[T_{norel}]$ denote the latency for the case with no straggler relaunch. We know
  \[ \E[T_{norel}] = s \frac{n!}{(n-k)!}\frac{\Gamma(n-k+1-1/\alpha)}{\Gamma(n+1-1/\alpha)} \]
  Also, by the law of total expectation,
  \begin{longaligned}[\label{eqn:eq_k_n_wrelaunch_ET_totalexpec_2}]
    \E[T_{norel}] &= \E[T_{norel} \mid T_{norel} \leq \Delta]\Pr\{T_{norel} \leq \Delta\} \\
    &\quad + \E[T_{norel} \mid T_{norel} > \Delta]\Pr\{T_{norel} > \Delta\} \\
    &= \E[X_{n:k} \mid R \geq k]\Pr\{R \geq k\} \\
    &\quad + \E[(X \mid X > \Delta)_{n-R:k-R} \mid R < k]\Pr\{R < k\}. \longalignedtag
  \end{longaligned}
  Taking the difference between \eqref{eqn:eq_k_n_wrelaunch_ET_totalexpec} and \eqref{eqn:eq_k_n_wrelaunch_ET_totalexpec_2}
  \begin{longaligned}[\label{eqn:eq_k_n_wrelaunch_ET_diff}]
    & \E[T] - \E[T_{norel}] \\
    &= \Bigl(\E[\Delta + X_{n-R:k-R} \mid R < k] \\
      &\quad - \E[(X \mid X > \Delta)_{n-R:k-R} \mid R < k] \Bigr)\Pr\{R < k\} \\
    &= \Delta \Pr\{R \leq k-1\} \\
      &\quad + \sum_{r=0}^{k-1} \left(\E[X_{n-r:k-r}] - \E[(X \mid X > \Delta)_{n-r:k-r}]\right)\Pr\{R=r\} \\
    &\stackrel{(a)}{=} \Delta I(1-q; n-k+1, k) \\
      &\quad + (s-\Delta)\frac{\Gamma(n-k+1-1/\alpha)}{\Gamma(n-k+1)} \\
      &\qquad \times \sum_{r=0}^{k-1} \frac{\Gamma(n-r+1)}{\Gamma(n-r+1-1/\alpha)} \Pr\{R=r\}. \longalignedtag
  \end{longaligned}
  where $(a)$ follows from observing that $(X \mid X > \Delta) \sim \mathrm{Pareto}(\Delta, \alpha)$ and using the mean of Pareto order statistics.
  \begin{longaligned}
    & \sum_{r=0}^{k-1} \frac{\Gamma(n-r+1)}{\Gamma(n-r+1-1/\alpha)} \binom{n}{r} q^r (1-q)^{n-r} \\
    &= \sum_{r=0}^{k-1} \frac{\Gamma(n+1)}{\Gamma(n-r+1-1/\alpha)\Gamma(r+1)} q^r (1-q)^{n-r} \\
    &= \frac{\Gamma(n+1)}{\Gamma(n+1-1/\alpha)}(1-q)^{1/\alpha} \\
      &\quad \times \sum_{r=0}^{k-1} \binom{n-1/\alpha}{r} q^r (1-q)^{n-r-1/\alpha} \\
    &= \frac{s}{\Delta}\frac{\Gamma(n+1)}{\Gamma(n+1-1/\alpha)} I(1-q; n-k+1-1/\alpha, k),
  \end{longaligned}
  substituting which in \eqref{eqn:eq_k_n_wrelaunch_ET_diff} gives
  \begin{longaligned}
    \E[T] &- \E[T_{norel}] = \Delta I(1-q; n-k+1, k) \\
      &\quad + (s-\Delta)\frac{\Gamma(n-k+1-1/\alpha)}{\Gamma(n-k+1)} \\
      &\quad \times \frac{s}{\Delta}\frac{\Gamma(n+1)}{\Gamma(n+1-1/\alpha)} I(1-q; n-k+1-1/\alpha, k) \\
    &= \Delta I(1-q; n-k+1, k) \\
      &\quad + \left(s/\Delta - 1\right)I(1-q; n-k+1-1/\alpha, k)\E[T_{norel}].
  \end{longaligned}
  from which \eqref{eqn:eq_k_n_wrelaunch_ET} follows.
\end{proof}

\subsection{Proof of Theorem~\ref{thm_k_cnd_wrelaunch_ET__EC}}
This Theorem presents the latency and cost for job execution when replicated or coded redundancy is introduced at time $\Delta$ together with straggler relaunch.
\begin{proof}
  \vspace{0.2em}
  \noindent
  \textbf{Replication.}
  Firstly, consider adding $c$ replicas for each remaining task together with performing straggler relaunch after waiting some time $\Delta$.
  
  We at first derive $\E[T]$.
  Let $\E[T_{wo}]$ be the latency of job execution when straggler relaunch is performed at time $\Delta$ with no redundancy employed.
  By the law of total expectation,
  \begin{longaligned}
    \E[T_{wo}] &= \E[T_{wo} \mid R=k]\Pr\{R=k\} \\
    &\quad + \sum_{r=0}^{k-1} \E[T_{wo} \mid R=r]\Pr\{R=r\} \\
    &= \E[X_{k:k} \mid R=k]\Pr\{R=k\} \\
    &\quad + \sum_{r=0}^{k-1} \E[\Delta + X_{k-r:k-r}]\Pr\{R=r\}.
  \end{longaligned}
  Repeating the same for $\E[T]$, we get
  \begin{longaligned}
    \E[T] &= \E[T \mid R=k]\Pr\{R=k\} + \sum_{r=0}^{k-1} \E[T \mid R=r]\Pr\{R=r\} \\
    &= \E[X_{k:k} \mid R=k]\Pr\{R=k\} \\
    &\quad + \sum_{r=0}^{k-1} \E\left[\Delta + (X_{(c+1):1})_{k-r:k-r}\right]\Pr\{R=r\}.
  \end{longaligned}
  where $X_{(c+1):1} \sim \mathrm{Pareto}(s, (c+1)\alpha)$. Taking the difference between these two,
  \begin{longaligned}
    & \E[T] - \E[T_{wo}] \\
    &= \sum_{r=0}^{k-1} \E[(X_{(c+1):1})_{k-r:k-r} - X_{k-r:k-r}] \Pr\{R=r\} \\
    &= \E\left[\E[(X_{(c+1):1})_{k-R:k-R} - X_{k-R:k-R} \mid R]\right] \\
    &= s\Gamma(1-1/(c+1)\alpha)\E\left[\frac{\Gamma(k-R+1)}{\Gamma(k-R+1-1/(c+1)\alpha)}\right] \\
      &\quad - s\Gamma(1-1/\alpha)\E\left[\frac{\Gamma(k-R+1)}{\Gamma(k-R+1-1/\alpha)}\right] \\
    &= s\frac{\Gamma(1-1/(c+1)\alpha)}{\Gamma(-1/(c+1)\alpha)}\E\left[\frac{\Gamma(k-R+1)\Gamma(-1/(c+1)\alpha)}{\Gamma(k-R+1-1/(c+1)\alpha)}\right] \\
      &\quad - s\frac{\Gamma(1-1/\alpha)}{\Gamma(-1/\alpha)}\E\left[\frac{\Gamma(k-R+1)\Gamma(-1/\alpha)}{\Gamma(k-R+1-1/\alpha)}\right] \\
    &= s\frac{\Gamma(1-1/(c+1)\alpha)}{\Gamma(-1/(c+1)\alpha)}\E[B(k-R+1, -1/(c+1)\alpha)] \\
      &\quad - s\frac{\Gamma(1-1/\alpha)}{\Gamma(-1/\alpha)}\E[B(k-R+1, -1/\alpha)] \\
    &\stackrel{(a)}{\approx} s\frac{\Gamma(1-1/(c+1)\alpha)}{\Gamma(-1/(c+1)\alpha)}B(k-kq+1, -1/(c+1)\alpha) \\
      &\quad - s\frac{\Gamma(1-1/\alpha)}{\Gamma(-1/\alpha)}B(k-kq+1, -1/\alpha).
  \end{longaligned}
  where $(a)$ follows from Lemma \ref{lm_approx_bin_mean_of_harmonic__ibeta}. Using this difference and the expression for $\E[T_{wo}]$ given previously in \eqref{eqn:eq_k_wrelaunch_tail__ET}, $\E[T]$ follows as given in \eqref{eqn:eq_k_cd_wrelaunch_ET}.
  
  We next derive cost without task cancellation.
  Let $R$ be the number of tasks completed before time $\Delta$, then $R \sim \mathrm{Binomial}(k, q)$ where $q = \mathbbm{1}(\Delta > s)(1 - (s/\Delta)^{\alpha})$. By iterated expectation,
  \begin{longaligned}
    & \E[C] = \E\left[\E[C \mid R]\right] \\
    &= \E\Bigl[\E\Bigl[\sum_{i=1}^R (X_i \mid X_i \leq \Delta) + (k-R)\Delta \\
      &\qquad\quad + \sum_{i=1}^{k-R} (c+1)X_{k-R:i} \mid R\Bigr]\Bigr] \\
    &= \E\Bigl[R\; \E[X_i \mid X_i \leq \Delta] + (k-R)\Delta \\
      &\qquad\quad + (c+1)(k-R)\E[X]\Bigr] \\
    &= kq\;\E[X_i \mid X_i \leq \Delta] + k(1-q)\Delta + k(c+1)(1-q)\E[X] \\
    &\stackrel{(a)}{=} \frac{k\alpha}{(\alpha-1)}(s - \Delta(1-q)) \\
      &\quad + k(1-q)\Delta + k(c+1)(1-q)s\frac{\alpha}{\alpha-1}.
  \end{longaligned}
  where $(a)$ comes from $\E[X \mid X \leq \Delta] = \frac{\alpha}{q(\alpha-1)}(s - \Delta(1-q))$ and $\E[X] = s\alpha/(\alpha-1)$.
  
  We next derive cost with task cancellation. By iterated expectation,
  \begin{longaligned}
    \E[C^c] &= \E\left[\E[C^c \mid R]\right] \\
    &= \E\Bigl[\E\Bigl[\sum_{i=1}^R (X_i \mid X_i \leq \Delta) + (k-R)\Delta \\
      &\qquad\quad~ + \sum_{i=1}^{k-R} (c+1)(X_{c+1:1})_{k-R:i} \mid R\Bigr]\Bigr] \\
    &= \E\bigl[R\; \E[X \mid X \leq \Delta] + (k-R)\Delta \\
    &\qquad~ + (c+1)(k-R)\E[X_{c+1:1}]\bigr] \\
    &= kq\;\E[X \mid X \leq \Delta] + k(1-q)\Delta \\
    &\quad~ + k(c+1)(1-q)\E[X_{c+1:1}] \\
    &= \frac{k\alpha}{(\alpha-1)}(s - \Delta(1-q)) + k(1-q)\Delta \\
      &\quad + k(c+1)(1-q)s\frac{(c+1)\alpha}{(c+1)\alpha-1}.
  \end{longaligned}
  where $(a)$ comes from
  \begin{equation*}
  \begin{split}
    \E[X \mid X \leq \Delta] &= \frac{\alpha}{q(\alpha-1)}(s - \Delta(1-q)), \\
    \E[X_{c+1:1}] &= s\frac{(c+1)\alpha}{(c+1)\alpha-1}.
  \end{split}
  \end{equation*}
  
  \vspace{0.2em}
  \noindent
  \textbf{Coding.}
  Secondly, we consider adding $n-k$ coded tasks into the job execution together with performing straggler relaunch after waiting some time $\Delta$.
  
  We at first derive $\E[T]$.
  When $\Delta \leq s$, since no task can complete before time $\Delta \leq s$, all $k$ tasks that are initially launched at time $0$ will be relaunched at time $\Delta$ together with the newly launched $n-k$ coded tasks. Equivalently, we think of the system performing the following: wait for $\Delta$ amount of time, then launch $n$ tasks, among which any $k$ is sufficient for the job completion. Therefore, job completion time is simply $\Delta + \E[X_{n:k}]$. Then, the mean of Pareto order statistics yields the expression given in \eqref{eqn:eq_k_nd_wrelaunch_ET} for $\Delta \leq s$.
  
  When $\Delta > s$, using the law of total expectation and denoting the number of tasks completed before time $\Delta$ as $R \sim \mathrm{Binomial}(k, q)$ for $q = \mathbbm{1}(\Delta > s)(1 - (s/\Delta)^{\alpha})$, we have
  \begin{longaligned}[\label{eqn:eq_k_nd_wrelaunch_ET_totalexpec}]
    \E[T] &= \E[T \mid T \leq \Delta]\Pr\{T \leq \Delta\} + \E[T \mid T > \Delta]\Pr\{T > \Delta\} \\
    &= \E[X_{k:k} \mid R = k]\Pr\{R = k\} \longalignedtag \\
      &\quad + \E[\Delta + X_{n-R:k-R} \mid R < k]\Pr\{R < k\}.
  \end{longaligned}
  
  Let $\E[T_{wo}]$ be the latency when only straggler relaunch is performed without adding any redundant task.
  Exact expression of $\E[T_{wo}]$ is given in \eqref{eqn:eq_k_wrelaunch_tail__ET}. By the law of total expectation,
  \begin{longaligned}[\label{eqn:eq_k_nd_wrelaunch_ET_totalexpec_2}]
    \E[T_{wo}] &= \E[T_{wo} \mid T_{wo} \leq \Delta]\Pr\{T_{wo} \leq \Delta\} \\
      &\quad + \E[T_{wo} \mid T_{wo} > \Delta]\Pr\{T_{wo} > \Delta\} \\
    &= \E[X_{k:k} \mid R = k]\Pr\{R = k\} \longalignedtag \\
      &\quad + \E[\Delta + X_{k-R:k-R} \mid R < k]\Pr\{R < k\}.
  \end{longaligned}
  
  Taking the difference between \eqref{eqn:eq_k_nd_wrelaunch_ET_totalexpec} and \eqref{eqn:eq_k_nd_wrelaunch_ET_totalexpec_2},
  \begin{longaligned}[\label{eqn:eq_k_nd_wrelaunch_ET_diff}]
    & \E[T] - \E[T_{wo}] = \Bigl(\E[X_{n-R:k-R} \mid R < k] \\
      &\qquad\qquad\qquad\qquad - \E[X_{k-R:k-R} \mid R < k]\Bigr)\Pr\{R < k\} \\
    &= \sum_{r=0}^{k-1} \left(\E[X_{n-r:k-r}] - \E[X_{k-r:k-r}]\right)\Pr\{R=r\} \\
    &\stackrel{(a)}{=} \sum_{r=0}^{k-1} \Bigl(s\frac{\Gamma(n-k+1-1/\alpha)}{(n-k)!}\frac{(n-r)!}{\Gamma(n-r+1-1/\alpha)} \\
      &\qquad\qquad - s\Gamma(1-1/\alpha)\frac{(k-r)!}{\Gamma(k-r+1-1/\alpha)}\Bigr)\Pr\{R=r\} \\
    &= s\;\frac{\Gamma(n-k+1-1/\alpha)}{(n-k)!}\E\left[\frac{(n-R)!}{\Gamma(n-R+1-1/\alpha)}\right] \longalignedtag \\
      &\quad - s\Gamma(1-1/\alpha)\E\left[\frac{(k-R)!}{\Gamma(k-R+1-1/\alpha)}\right],
  \end{longaligned}
  where $(a)$ comes from the mean of Pareto order statistics and observing that $\E[X_{n-r:k-r}] - \E[X_{k-r:k-r}] = 0$ for $r = k$. Then,
  \begin{longaligned}
    \E&\left[\frac{(k-R)!}{\Gamma(k-R+1-1/\alpha)}\right] \\
    &= \sum_{r=0}^k \frac{(k-r)!}{(k-r-1/\alpha)!} \binom{k}{r} q^r(1-q)^{k-r} \\
    &= \frac{q^k}{\Gamma(1-1/\alpha)} + \sum_{r=0}^{k-1} \frac{k!}{(k-r-1/\alpha)!r!} q^r(1-q)^{k-r} \\
    &= \frac{q^k}{\Gamma(1-1/\alpha)} \\
      &\quad + \frac{k!}{\Gamma(k+1-1/\alpha)}\frac{s}{\Delta} \sum_{r=0}^{k-1} \binom{k-1/\alpha}{r} q^r(1-q)^{k-r-1/\alpha} \\
    &= \frac{q^k}{\Gamma(1-1/\alpha)} + \frac{k!}{\Gamma(k+1-1/\alpha)}\frac{s}{\Delta} I(1-q;1-1/\alpha,k).
  \end{longaligned}
  Substituting this in \eqref{eqn:eq_k_nd_wrelaunch_ET_diff} gives us
  \begin{longaligned}
    \E[T] &- \E[T_{wo}] \\
    &= s\frac{\Gamma(n-k+1-1/\alpha)}{(n-k)!}\E\left[\frac{(n-R)!}{\Gamma(n-R+1-1/\alpha)}\right] \\
      &\quad - s\Gamma(1-1/\alpha)\E\left[\frac{(k-R)!}{\Gamma(k-R+1-1/\alpha)}\right] \\
    &= s\frac{\Gamma(n-k+1-1/\alpha)}{(n-k)!}\E\left[\frac{(n-R)!}{\Gamma(n-R+1-1/\alpha)}\right] \\
      &\quad - s q^k - k\frac{s^2}{\Delta}\frac{\Gamma(1-1/\alpha)\Gamma(k)}{\Gamma(k+1-1/\alpha)}I(1-q;1-1/\alpha,k) \\
    &= s\frac{\Gamma(n-k+1-1/\alpha)}{(n-k)!}\E\left[\frac{(n-R)!}{\Gamma(n-R+1-1/\alpha)}\right] \\
      &\quad - s q^k - k\frac{s^2}{\Delta}B(1-1/\alpha,k)I(1-q;1-1/\alpha,k).
  \end{longaligned}
  Then, using the expression of $\E[T_{wo}]$ given in \eqref{eqn:eq_k_wrelaunch_tail__ET},
  \begin{longaligned}[\label{eqn:eq_k_nd_wrelaunch_ET_almost}]
    & \E[T] = s\frac{\Gamma(n-k+1-1/\alpha)}{(n-k)!}\E\left[\frac{(n-R)!}{\Gamma(n-R+1-1/\alpha)}\right] \\
      &- s q^k - k\frac{s^2}{\Delta}B(1-1/\alpha,k)I(1-q;1-1/\alpha,k) + \Delta(1-q^k) \\
      &+ ks B(1-1/\alpha,k)\left(1 - \left(1-\frac{s}{\Delta}\right)I(1-q;1-1/\alpha,k)\right) \\
    &= s\;\frac{\Gamma(n-k+1-1/\alpha)}{(n-k)!}\E\left[\frac{(n-R)!}{\Gamma(n-R+1-1/\alpha)}\right] \longalignedtag \\
      &\quad - s q^k + \Delta(1-q^k) + ks B(q;k,1-1/\alpha).
  \end{longaligned}
  Although the equality $B(x,y) = \Gamma(x)\Gamma(y)/\Gamma(x+y)$ can be shown with the integral evaluation of $B$ and $\Gamma$ functions only for $Re\{x\}, Re\{y\} > 0$, {\it analytic continuation} of these functions allows us do the following manipulation
  \begin{longaligned}
    & \frac{\Gamma(n-k+1-1/\alpha)}{(n-k)!}\E\left[\frac{(n-R)!}{\Gamma(n-R+1-1/\alpha)}\right] \\
    &= \frac{\Gamma(n-k+1-1/\alpha)}{\Gamma(n-k+1)\Gamma(-1/\alpha)}\E\left[\frac{\Gamma(n-R+1)\Gamma(-1/\alpha)}{\Gamma(n-R+1-1/\alpha)}\right] \\
    &= \frac{\E[B(n-R+1,-1/\alpha)]}{B(n-k+1,-1/\alpha)} \stackrel{(a)}{\approx} \frac{B(n-kq+1,-1/\alpha)}{B(n-k+1,-1/\alpha)}.
  \end{longaligned}
  where $(a)$ follows from Lemma \ref{lm_approx_bin_mean_of_harmonic__ibeta}. Substituting which in \eqref{eqn:eq_k_nd_wrelaunch_ET_almost} gives us
  \begin{longaligned}
    &\E[T] = \Delta(1-q^k) \\
    &\quad + s\left(\frac{B(n-kq+1,-1/\alpha)}{B(n-k+1,-1/\alpha)} + kB(q;k,1-1/\alpha) - q^k\right).
  \end{longaligned}
  
  Next, we derive cost without task cancellation.
  By the law of total expectation,
  \begin{longaligned}
    \E[C] &= \E[C \mid R=k]\Pr\{R=k\} \\
      &\quad + \sum_{r=0}^{k-1} \E[C \mid R=r]\Pr\{R=r\} \\
    &= \E\left[\sum_{i=1}^k X_i \mid X_i \leq \Delta\right]\Pr\{R=k\} \\
      &\quad + \sum_{r=0}^{k-1} \E\Bigl[\sum_{i=1}^r (X_i \mid X_i \leq \Delta) + (k-r)\Delta \\
      &\qquad\qquad\quad + \sum_{i=1}^{n-r} X_{n-r:i}\Bigr]\Pr\{R=r\} \\
    &= \left(\sum_{i=1}^k \E[X \mid X \leq \Delta]\right)\Pr\{R=k\} \\
      &\quad + \sum_{r=0}^k \E\Bigl[\sum_{i=1}^r (X_i \mid X_i \leq \Delta) + (k-r)\Delta \\
      &\qquad\qquad\quad + \sum_{i=1}^{n-r} X_{n-r:i}\Bigr]\Pr\{R=r\} \\
      &\quad - \E\left[\sum_{i=1}^k (X_i \mid X_i \leq \Delta) + \sum_{i=1}^{n-k} X_{n-k:i}\right]\Pr\{R=k\} \\
    &= \E\left[\E\left[\sum_{i=1}^R (X_i \mid X_i \leq \Delta) + (k-R)\Delta + \sum_{i=1}^{n-R} X_{n-R:i}\right]\right] \\
      &\quad - \E\left[\sum_{i=1}^{n-k} X_{n-k:i}\right]\Pr\{R=k\} \\
    &= \E\left[R\;\E[X \mid X \leq \Delta] + (k-R)\Delta + (n-R)\E[X]\right] \\
      &\quad - (n-k)\E[X] q^k \\
    &= kq\;\E[X \mid X \leq \Delta] + (k-kq)\Delta + (n-kq)\E[X] \\
      &\quad - (n-k)\E[X] q^k \\
    &\stackrel{(a)}{=} \frac{k\alpha}{\alpha-1}(s - \Delta(1-q)) + (k-kq)\Delta + (n-kq)\frac{s\alpha}{\alpha-1} \\
      &\quad - (n-k)q^k\frac{s\alpha}{\alpha-1} \\
    &= \frac{\alpha}{\alpha-1}\left(ks - k\Delta(1-q) + ns - ks q - ns q^k + ks q^k\right) \\
      &\quad + k(1-q)\Delta \\
    &= \frac{\alpha}{\alpha-1}\left(ks(1-q+q^k) + ns(1-q^k)\right) - \frac{k\Delta(1-q)}{\alpha-1}.
  \end{longaligned}
  where $(a)$ comes from substituting the following
  \begin{longaligned}
    & \E[X \mid X \leq \Delta] = \int_{s}^{\Delta} x\Pr\{X=x \mid X \leq \Delta\}dx \\
    &= \int_{s}^{\Delta} x \frac{\Pr\{X=x\}}{\Pr\{X \leq \Delta\}} dx = \frac{\alpha s^{\alpha}}{q} \int_{s}^{\Delta} x^{-\alpha}dx \\
    &= \frac{s^{\alpha}}{q}\frac{\alpha}{\alpha-1}\left(s^{1-\alpha} - \Delta^{1-\alpha}\right) = \frac{\alpha}{q(\alpha-1)}\left(s - \Delta(1-q)\right),
  \end{longaligned}
  where $\Pr\{X = x\}$ denotes the PDF of $X$.
  
  Next, we derive cost with task cancellation.
  By the law of total expectation,
  \begin{longaligned}[\label{eqn:eq_k_nd_wrelaunch_ECwcancel_lawoftotal}]
    & \E[C^c] = \E[C^c \mid R=k]\Pr\{R=k\} \\
      &\quad + \sum_{r=0}^{k-1} \E[C^c \mid R=r]\Pr\{R=r\} \\
    &= \E\left[\sum_{i=1}^k X_i \mid X_i \leq \Delta\right]\Pr\{R=k\} \\
      &\quad + \sum_{r=0}^{k-1} \E\Biggl[\sum_{i=1}^r (X_i \mid X_i \leq \Delta) \\
      &\quad + (k-r)\Delta + \sum_{i=1}^{k-r} X_{n-r:i} + (n-k)X_{n-r:k-r}\Biggr]\Pr\{R=r\} \\
    &= \left(\sum_{i=1}^k \E[X \mid X \leq \Delta]\right)\Pr\{R=k\} \\
      &\quad + \sum_{r=0}^{k} \E\Biggl[\sum_{i=1}^r (X_i \mid X_i \leq \Delta) \\
      &\quad + (k-r)\Delta + \sum_{i=1}^{k-r} X_{n-r:i} + (n-k)X_{n-r:k-r}\Biggr]\Pr\{R=r\} \\
      &\quad - \E\left[\sum_{i=1}^k (X_i \mid X_i \leq \Delta) + (n-k)X_{n-k:0}\right]\Pr\{R=k\} \\
    &= \E\Biggl[\E\Biggl[\sum_{i=1}^R (X_i \mid X_i \leq \Delta) + (k-R)\Delta + \sum_{i=1}^{k-R} X_{n-R:i} \\
      &\qquad\qquad + (n-k)X_{n-R:k-R} \mid R\Biggr]\Biggr] - s(n-k)q^k \\
    &= kq\;\E[X \mid X \leq \Delta] + k(1-q)\Delta \\
      &\quad + \E\left[\sum_{i=1}^{k-R} \E[X_{n-R:i} \mid R]\right] \\
      &\quad + (n-k)\E\left[\E[X_{n-R:k-R} \mid R]\right] - s(n-k)q^k. \longalignedtag
  \end{longaligned}
  We expand the second term of the sum above as
  \begin{longaligned}
    & \E\left[\sum_{i=1}^{k-R} \E[X_{n-R:i} \mid R]\right] \\
    &= \E\left[\sum_{i=1}^{k-R} s\frac{\Gamma(n-R+1)}{\Gamma(n-R+1-1/\alpha)}\frac{\Gamma(n-R-i+1-1/\alpha)}{\Gamma(n-R-i+1)}\right] \\
    &\stackrel{(a)}{=} \E\left[s\frac{\Gamma(n-R+1)}{\Gamma(n-R+1-1/\alpha)} \sum_{j=n-k+1}^{n-R} \frac{\Gamma(j-1/\alpha)}{\Gamma(j)}\right] \\
    &= \E\Biggl[s\frac{\Gamma(n-R+1)}{\Gamma(n-R+1-1/\alpha)} \\
      &\qquad\quad \times \left(\sum_{j=0}^{n-R} \frac{\Gamma(j-1/\alpha)}{\Gamma(j)} - \sum_{j=0}^{n-k} \frac{\Gamma(j-1/\alpha)}{\Gamma(j)}\right)\Biggr] \\
    &\stackrel{(b)}{=} \E\Biggl[s\frac{\Gamma(n-R+1)}{\Gamma(n-R+1-1/\alpha)} \frac{\alpha}{\alpha-1} \\
      &\qquad\quad \times \left(\frac{\Gamma(n-R+1-1/\alpha)}{\Gamma(n-R)} - \frac{\Gamma(n-k+1-1/\alpha)}{\Gamma(n-k)}\right)\Biggr] \\
    &= \frac{\alpha}{\alpha-1}s(n-kq) \\
      &\quad - \frac{\alpha}{\alpha-1}s\frac{\Gamma(n-k+1-1/\alpha)}{\Gamma(n-k)}\E\left[\frac{\Gamma(n-R+1)}{\Gamma(n-R+1-1/\alpha)}\right].
  \end{longaligned}
  where $(a)$ comes from defining $j=n-R-i+1$, and $(b)$ from Lemma~\ref{lm_sum_gamma_ratios}.
  We also have
  \begin{longaligned}
    \E&\left[\E[X_{n-R:k-R} \mid R]\right] \\
    &\quad = s\frac{\Gamma(n-k+1-1/\alpha)}{\Gamma(n-k+1)}\E\left[\frac{\Gamma(n-R+1)}{\Gamma(n-R+1-1/\alpha)}\right].
  \end{longaligned}
  Substituting these in \eqref{eqn:eq_k_nd_wrelaunch_ECwcancel_lawoftotal} gives us
  \begin{longaligned}
    & \E[C^c] = kq\;\E[X \mid X \leq \Delta] + k(1-q)\Delta \\
      &\quad + \E\left[\sum_{i=1}^{k-R} \E[X_{n-R:i} \mid R]\right] + (n-k)\E\left[E[X_{n-R:k-R} \mid R]\right] \\
      &\quad - s(n-k)q^k \\
    &= \frac{\alpha}{\alpha-1}k(s-\Delta(1-q)) + k(1-q)\Delta \\
      &\quad - s(n-k)q^k + \frac{\alpha}{\alpha-1}s(n-kq) \\
      &\quad - \frac{\alpha}{\alpha-1}s\frac{\Gamma(n-k+1-1/\alpha)}{\Gamma(n-k)}\E\left[\frac{\Gamma(n-R+1)}{\Gamma(n-R+1-1/\alpha)}\right] \\
      &\quad + (n-k)s\frac{\Gamma(n-k+1-1/\alpha)}{\Gamma(n-k+1)}\E\left[\frac{\Gamma(n-R+1)}{\Gamma(n-R+1-1/\alpha)}\right] \\
    &= \frac{\alpha}{\alpha-1}(k(1-q)(s-\Delta) + ns) + k(1-q)\Delta - s(n-k)q^k \\
      &\quad - \frac{s}{\alpha-1}\frac{\Gamma(n-k+1-1/\alpha)}{\Gamma(n-k)}\E\left[\frac{\Gamma(n-R+1)}{\Gamma(n-R+1-1/\alpha)}\right].
  \end{longaligned}
  
  In the proof of Thm.~\ref{thm_k_cnd_wrelaunch_ET__EC}, we showed that
  \begin{longaligned}
    \frac{\Gamma(n-k+1-1/\alpha)}{\Gamma(n-k+1)} & \E\left[\frac{\Gamma(n-R+1)}{\Gamma(n-R+1-1/\alpha)}\right] \\
    &\approx \frac{B(n-kq+1, -1/\alpha)}{B(n-k+1, -1/\alpha)}.
  \end{longaligned}
  Using which we can finalize the derivation as
  \begin{longaligned}
    \E[C^c] &= \frac{\alpha}{\alpha-1}(k(1-q)(s-\Delta) + ns) + k(1-q)\Delta \\
      &\quad - s(n-k)q^k - \frac{s}{\alpha-1}(n-k)\frac{B(n-kq+1, -1/\alpha)}{B(n-k+1, -1/\alpha)}.
  \end{longaligned}

\end{proof}

\end{document}